\documentclass[12pt]{article}
\usepackage[utf8]{inputenc}
\usepackage{jones}

\ifpdf
\hypersetup{
    pdftitle={Fourier Analysis of Iterative Algorithms},
    pdfauthor={Chris Jones, Lucas Pesenti}
}
\fi
\usepackage{tabularx}
\newcolumntype{Y}{>{\centering\arraybackslash}X}

\newcommand{\rootpic}{\vcenter{\hbox{\includegraphics[height=2ex]{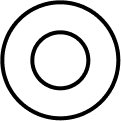}}}}
\newcommand{\smallrootpic}{\vcenter{\hbox{\includegraphics[height=1.5ex]{images/singleton.png}}}}
\newcommand{\ijrootpic}{\vcenter{\hbox{\includegraphics[height=2ex]{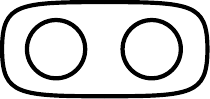}}}}
\newcommand{\smallijrootpic}{\vcenter{\hbox{\includegraphics[height=1.5ex]{images-pdf/ij-roots.pdf}}}}

\newcommand{\scalar}{\textnormal{scalar}}
\newcommand{\float}{\textnormal{float}}

\newcommand{\eqinf}{\,\overset{\infty}{=}\,}

\newcommand{\sdots}{\;\dots,\;}

\renewcommand{\toas}{\overset{\textnormal{a.s.}}{\longrightarrow}}
\renewcommand{\tod}{\overset{\textnormal{d}}{\longrightarrow}}

\title{Fourier Analysis of Iterative Algorithms}

\author{Chris Jones\thanks{Bocconi University. \texttt{chris.jones@unibocconi.it}.} \and Lucas Pesenti\thanks{Bocconi University. \texttt{lucas.pesenti@phd.unibocconi.it} }}

\begin{document}
\pagestyle{empty}

\maketitle{}

\thispagestyle{empty}

\begin{abstract}
    We study a general class of nonlinear iterative algorithms which includes power iteration, belief propagation and approximate message passing, and many forms of gradient descent. When the input is a random matrix with i.i.d. entries, we use Boolean Fourier analysis to analyze these algorithms as low-degree polynomials in the entries of the input matrix.
    Each symmetrized Fourier character represents all monomials with a certain shape as specified by a small graph, which we call a \emph{Fourier diagram}.

    We prove fundamental asymptotic properties of the Fourier diagrams: over the randomness of the input,
    all diagrams with cycles are negligible; the tree-shaped diagrams form a basis of \emph{asymptotically independent Gaussian vectors}; and, when restricted to the trees, iterative algorithms exactly follow an idealized Gaussian dynamic.
    We use this to prove a state evolution formula, giving a ``complete'' asymptotic description of the algorithm's trajectory.

    The restriction to tree-shaped monomials mirrors the assumption of the  \emph{cavity method}, a 40-year-old non-rigorous technique in statistical physics which has served as one of the most important techniques in the field.
    We demonstrate how to implement cavity method derivations by 1)~restricting the iteration to its
    tree approximation, and 2)~observing
    that heuristic cavity method-type arguments hold
    rigorously on the simplified
    iteration. Our proofs use combinatorial arguments similar to the trace method from random matrix theory.
    
    Finally, we push the diagram analysis to a number of iterations that scales with the dimension $n$ of the input matrix, proving that the tree approximation still holds for a simple variant of power iteration all the way up to $n^{\Omega(1)}$ iterations.
\end{abstract}

 \newpage

\tableofcontents
\newpage
\cleardoublepage
\pagenumbering{arabic}
\pagestyle{plain}

\section{Introduction}

We study nonlinear iterative algorithms which take as input a matrix $A\in\R^{n\times n}$, maintain a vector state $x_t\in\R^n$, and at each step
\begin{enumerate}
    \item either multiply the state by $A$, 
    \[x_{t+1} = Ax_t\,,\]
    \item or apply the same function $f_t:\R^{t+1}\to\R$ to each coordinate of the previous states,
    \[x_{t+1}=f_t(x_t, \dots, x_0)\,.\] 
\end{enumerate}
This class of algorithms has been coined \emph{general first-order methods} (GFOM)~\cite{celentano2020estimation, montanari2022statistically}.
GFOM algorithms are a simple, widespread, practically efficient, and incredibly powerful computational model.
Alternating linear and nonlinear steps can describe first-order optimization algorithms including power iteration and many types 
of gradient descent (see \cite{celentano2020estimation, gerbelot2022rigorous}).
This definition also captures belief propagation and other message-passing algorithms which play a central role not only in the design of Bayes-optimal algorithms for planted signal recovery~\cite{feng2022unifying}, but also recently in average-case complexity theory for the optimization of random polynomials~\cite{surveyOptimizingSpinGlass}.

In machine learning and artificial intelligence, deep neural networks exhibit a similar structure which alternates multiplying weight matrices and applying nonlinear functions.
Remarkably, viewed from this level of generality, the line blurs between neural networks and the gradient descent algorithms used to train them.

Despite the widespread use of GFOM and deep neural networks,
developing a mathematical theory for these algorithms continues to be a major challenge.
Thus far, it has been difficult to isolate mathematical theorems which describe key phenomena but avoid being too specific to any one setting, model, or algorithm.
That being said, one effective theory has emerged
at the interface of computer science, physics, and statistics for
studying a class of
nonlinear iterations known as Belief Propagation~(BP) and Approximate Message Passing~(AMP) algorithms.
This theory is most developed for inputs $A$ that are \emph{dense random matrices with i.i.d entries}, also known as a \emph{mean-field models} in physics, and which can be considered the simplest possible model of random data.

The analysis of BP and AMP algorithms in this setting can 
be summarized by the \emph{state evolution} formula~\cite{donoho2009message, bolthausen2014iterative}.
This is an impressive ``complete''
description of the trajectory of the iterates $x_t \in \R^n$, in the limit $n\to\infty$.
Specifically, state evolution
defines a sequence of {\em scalar} random variables $X_t$ such that for essentially \emph{any} symmetric quantity of interest
related to $x_t$,
the expectation of a corresponding expression in $X_t$
approximates the quantity with an error
that goes to $0$ as $n\to\infty$.
This yields analytic formulas for quantities such as the loss function or objective value achieved by $x_t$, the norm of $x_t$, the correlation between $x_s$ and $x_t$ across
iterations,
or the fraction of $x_t$'s coordinates which lie in the interval $[-1,+1]$.
The ability to precisely analyze the trajectory and the fixed points of message-passing algorithms (through $X_t$ with large $t$) has been key to their applications.

State evolution for BP/AMP iterations was originally predicted using a powerful and influential
technique from statistical physics known as the \emph{cavity method}.
Variants of BP have been studied in physics as ``non-linear dynamical systems'' as far back as the work of Bethe~\cite{bethe1935statistical}, although the algorithmic perspective came into prominence only later. 
The cavity method and the related replica method were devised in the 1980s~\cite{Parisi79,parisi1980sequence,cavityMethod86, mezard1987spinglasstheoryandbeyond}, initially as a tool to compute thermodynamic properties of disordered systems, and later as a tool for analyzing belief propagation algorithms.
Since their introduction, the cavity method and the replica method have served as two of the most fundamental
tools in the statistical physics toolbox.

The deployment of these techniques has undoubtedly been a huge success; there are many survey articles offering various perspectives from different fields \cite{yedidia2003understanding, MezardMontanari, koller2009probabilistic, zdeborova2016statistical, gabrie2020mean, feng2022unifying, zou2022concise, charbonneau2023spin}.
However, the reality is that the situation is not as unified as the above picture would suggest, due to a major issue:
\emph{the physical methods are not mathematically rigorous}.

At present, there exists a significant gap between how results are established in the physical and mathematical literature.
The two general types of results are: 1) simple non-rigorous arguments based on the cavity/replica method; 2) mathematically rigorous arguments that confirm the physical predictions, but with technically sophisticated proofs that can't closely follow the path of the physical reasoning.
For example, the state evolution formula  was first proven by Bolthausen~\cite{bolthausen2014iterative} using a Gaussian conditioning technique which is fairly technically involved.
Although many proofs have been found for predictions of the cavity and replica methods, none can be said to clearly explain the success of the physicists' techniques.

It has appeared that the physicists have some secret advantage yet unmatched by rigorous mathematics.
Is there a simple and rigorous mathematical framework that explains why the assumptions made by physicists always seem to work?

\subsection{Our contributions}

We introduce a new method to analyze nonlinear iterations based on Fourier analysis, when the input to the algorithm is a random matrix with i.i.d entries. Our framework
gives proofs that are able to
closely follow heuristic physics derivations.

Our strategy is to replace the original iteration $(x_t)_{t \geq 0}$
by an idealized version
\[x_t \approx \widehat x_t\,,\]
which we call the \emph{tree approximation} to $x_t$.
The analysis then follows a
two-step structure:
\begin{enumerate}
    \item The tree approximation $\widehat x_t$ tracks the original
    iteration $x_t$ up to a uniform $\widetilde{O}(n^{-\frac 1 2})$ entrywise error.
    Hence, any reasonable asymptotic
    result established on
    $(\widehat x_t)_{t\ge 0}$ (such as the joint distribution of their entries) automatically extends to $(x_t)_{t\ge 0}$.
    \item Cavity method-type reasoning can be 
    rigorously applied to the
    tree approximation. 
    In cases where $\widehat x_t$ has already been analyzed in physics,
    one can essentially copy
    the heuristic physics derivation. 
\end{enumerate}

Analyzing $\widehat x_t$ is a
significant simplification compared to the entire state $x_t$---in fact, we show 
that the former follows an explicit {\em Gaussian dynamic}.
The simplification directly yields
a \emph{state evolution} formula for GFOM algorithms,
as well as rigorous
implementations of physics-based cavity method arguments (in the algorithmic or ``replica symmetric'' setting of the method).
In other words, our new notion of tree
approximation matches implicit assumptions
of the cavity method and gives a way to
justify them.

We define the tree approximation $\widehat x_t$ essentially as follows: we expand
the entries of $x_t$ as polynomials in the entries of the
input matrix~$A \in \R^{n \times n}$.
If we represent the monomials 
(e.g. $A_{12}A_{23}A_{24}$) as graphs in the natural way, then $\widehat x_t$
consists of only the monomials appearing in $x_t$ whose graph is a tree. 
Hence, we will show that the state of an
iterative algorithm can be tightly approximated using the much smaller set of tree monomials.

\vspace{-8pt}
\paragraph{Fourier diagrams.}
We view iterates of a GFOM with polynomial
non-linearities as 
vector-valued polynomials in the
entries of $A$. These polynomials have 
a special symmetry:
they are invariant
under permutations of the row/column indices of~$A$.

The polynomial representation can be visualized using \emph{Fourier diagrams},
each of which is a small graph representing all the monomials with a given shape. For example, here are three Fourier diagrams along with
the vectors associated with them.

\begin{table}[h]
\begin{tabularx}{\textwidth}{YYY}
     \includegraphics[scale=1]{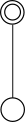} & \includegraphics[scale=1]{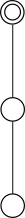} & \includegraphics[scale=1]{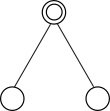}\\\vspace{3mm}
     $\displaystyle Z_i \defeq \sum_{\substack{ j = 1\\\text{$i, j$ distinct}}}^n A_{ij}$& $\displaystyle Z_i' \defeq \sum_{\substack{ j,k = 1\\\text{$i, j, k$ distinct}}}^n A_{ij}A_{jk}$ & $\displaystyle Z_i'' \defeq \sum_{\substack{ j,k = 1\\\text{$i, j, k$ distinct}}}^n A_{ij}A_{ik}$ 
\end{tabularx}
\end{table}

In general, a Fourier diagram is an undirected rooted multigraph $\alpha=(V(\alpha),E(\alpha))$ which represents the vector $Z_\al\in\R^n$ whose entries are:
\begin{align}
    Z_{\al,i} := \sum_{\substack{\textnormal{injective }\ph\colon V(\al) \to [n]\\\ph(\smallrootpic)=i}} \prod_{\{u,v\} \in E(\al)} A_{\ph(u)\ph(v)}\,,\quad\textnormal{ for all $i\in [n]$\,.}\label{eq:diagram-short}
\end{align}
We use $\rootpic \in V(\al)$ to notate the root vertex.

The symmetry of the GFOM operations ensures
that in the polynomial representation of an iterate $x_t$,
all monomials corresponding to the same Fourier diagram
come with the same coefficient.
Therefore, any iterate $x_t$ of a GFOM with 
polynomial non-linearities
can be expressed as a linear combination of Fourier diagrams,
in which case we say that it is written \emph{in the Fourier diagram basis}. 

We emphasize that these diagrams are
constructed by summing over {\em injective} embeddings $\ph\colon V(\al)\to[n]$, a crucial detail for the
results that follow. The term ``Fourier'' reflects that
this basis of polynomials is a symmetrized version of the
standard Fourier basis from Boolean function analysis (see \cref{sec:sym-Fourier-analysis}).

\vspace{-8pt}
\paragraph{Asymptotic diagram analysis.}
It turns out that something special happens to the
Fourier diagram basis in the limit $n\to\infty\,$, when
$A$ is a symmetric matrix with independent
mean-0, variance-$\frac 1 n$ entries. Informally, the entries of the diagrams become mutually independent, and the following properties hold.
\begin{itemize}
    \item The Fourier diagrams with cycles are negligible.
    \item The Fourier tree diagrams with one branch from the root are independent Gaussian vectors.
    \item The Fourier tree diagrams with several branches from the root are Hermite polynomials in the Gaussians represented by the branches.
\end{itemize}
Most importantly, \textbf{the only non-negligible contributions come from the trees}.
Based on this classification, we define the
{\em tree approximation} $\widehat x_t$ of an expression $x_t$ written
in the Fourier diagram basis to be obtained by discarding all diagrams with cycles.
That is, as polynomials in $A\,$, $\widehat x_t$ consists of the tree-shaped monomials in $x_t$.

The reason that cyclic Fourier diagrams are negligible is combinatorially intuitive:
cyclic diagrams sum over fewer terms than tree-shaped diagrams.
For example, the left diagram is a sum over ${\displaystyle\approx }\,n^4$ terms, each mean-zero with magnitude $\sim n^{-2}$, so the overall magnitude is $\Theta(1)$. The right diagram is a sum over ${\displaystyle\approx }\,n^3$ terms, again of magnitude $n^{-2}$, so the overall of order of the diagram is $\Theta(n^{-1/2})$.
\[    \vcenter{\hbox{\includegraphics[height=10ex]{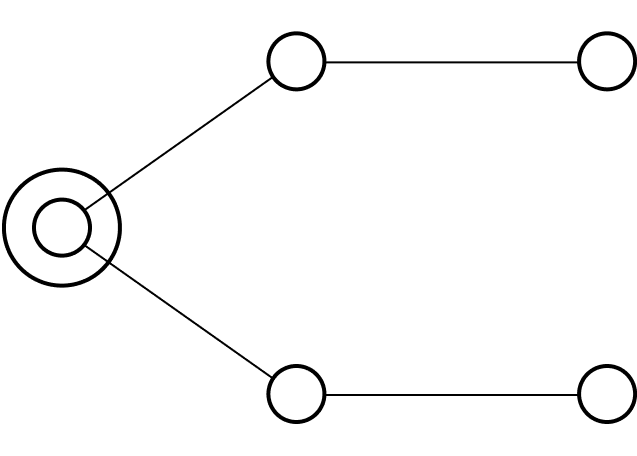}}} \qquad \qquad \vcenter{\hbox{\includegraphics[height=10ex]{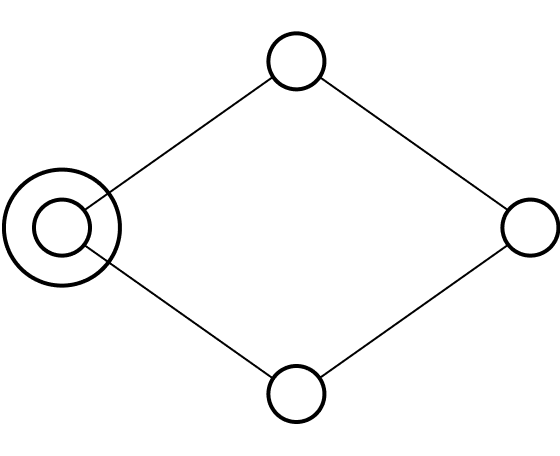}}}
\]

We now state our main theorems. In all of them, we assume that $A$ is a symmetric matrix with independent mean-0 variance-$\frac 1n$ entries (see \cref{assump:A-entries}). First, we formalize the classification above by proving that all joint moments of the Fourier diagrams converge to those of the corresponding random variables in a Gaussian space.

\begin{theorem}[Classification theorem; see \cref{thm:classification}]
\label{thm:intro-classification}
    For any $k\ge 0$ independent of $n$, for all connected Fourier diagrams $\al_1, \dots, \al_k$ and $i_1, \ldots, i_k\in [n]$ (allowing repetitions in $\al_j$ and $i_j$),
    \[
        \E_{A} \left[\prod_{j = 1}^k Z_{\al_j,i_j}\right] = \E\left[\prod_{j = 1}^k Z^{\infty}_{\al_j,i_j}\right] + O(n^{-\frac 1 2})\,,
    \]
    where for any connected Fourier diagram $\alpha$ and $i\in [n]$,
    \begin{enumerate}
        \item $Z^\infty_{\al,i} = 0$, if $\alpha$ has a cycle.
        \item $Z^\infty_{\al,i} \sim \calN(0, \left|\Aut(\al)\right|)$ independently, if $\al$ is a tree whose root has degree $1$.
        \item $Z^\infty_{\al,i} = \prod_{\tau} h_{d_\tau}(Z_{\tau,i}^\infty\;; \left|\Aut(\tau)\right|)$ if $\al$ is a tree consisting of $d_\tau$ copies of each tree $\tau$ from case 2 merged at the root, where $h_{d_\tau}$ are the Hermite polynomials (defined in \cref{sec:preliminaries}).
    \end{enumerate}
    \label{thm:main1}
\end{theorem}

Next, we prove that the tree approximation of
a GFOM closely tracks the original iteration. 
This addresses the first of the two steps
from the overview of our method.

\begin{theorem}[Tree approximation of GFOMs; see \cref{thm:state-evolution}]
\label{thm:intro-tree-approx}
    Let $t$ be a constant, $x_t\in \R^n$ be the state of a GFOM with polynomial non-linearities, and $\widehat{x}_t \in \R^n$ be the state obtained by performing the GFOM operations on only the tree diagrams. Then with high probability over $A$,
    \[
        \norm{x_t - \widehat{x}_t}_\infty = \widetilde{O}(n^{-\frac 1 2})\,.
    \]
    \label{thm:main2}
\end{theorem}

The statement of this theorem exactly isolates a key and subtle point: not only are the cyclic
diagrams negligible at time $t$, but they will never blow up to affect
the state at any future time.
The fact that ``approximation
errors do not propagate'' is what gives us the ability to pass the algorithm to an asymptotic limit.\footnote{This directly addresses a question raised in the seminal paper of Donoho, Maleki, and Montanari on approximate message passing~\cite[Section III.E]{donoho2010message}.} 

The proof of \cref{thm:intro-tree-approx} is intuitive.
According to the diagram classification theorem,
we can tease out the approximation error for $x_t$ as the monomials with cycles,
whereas the approximating quantity $\widehat{x}_t$ consists of the tree monomials.
When a GFOM operation is applied, the cycles persist in all cyclic monomials,
and hence they continue to be negligible.

As a direct consequence of these results, we deduce a very strong form of state evolution for all GFOM algorithms.
The theorem below paints a nearly complete picture of the evolution of $x_t$
as $n$ independent trajectories of a single random variable $X_t$ which is an ``idealized Gaussian dynamic''
in correspondence with $\widehat{x}_t$.

\begin{theorem}[General state evolution; see \cref{thm:general-state-evolution}]
    \label{thm:intro-general-state}
    Let $t$ be a constant, $x_t\in \R^n$ be the state of a GFOM with polynomial non-linearities, and let $X_t$ be the asymptotic state of $x_t$ (\cref{def:asymptotic-state}). Then:
    \begin{enumerate}[(i)]
        \item For each $i \in [n]$, $(x_{0,i}, \dots, x_{t,i}) \tod (X_0, \dots, X_t)$.
        Furthermore, the coordinates' trajectories $\{(x_{0,i},\dots, x_{t,i}) : i \in [n]\}$ are asymptotically independent.
        \item With high probability over $A$,
    \[
        \frac{1}{n} \sum_{i = 1}^n x_{t,i} = \E[X_t] + \widetilde{O}(n^{-\frac 1 2})\,.
    \]
    \end{enumerate}
\end{theorem}
Quantities such as the norm of $x_t$ can be computed using part (ii)
along with one additional GFOM iteration that squares $x_t$ componentwise.
Without much extra work, \cref{thm:intro-general-state} also encapsulates other key features of previous state evolution formulas including quantitative error bounds (similar to the main result of \cite{rush2018finite}) and universality (the main result of \cite{bayati2015universality}).\footnote{Similarly to \cite{bayati2015universality}, our technical analysis assumes that the nonlinearities in the GFOM are polynomial functions, but other works have been able to handle the larger class of \emph{pseudo-Lipschitz} non-linearities. We do not find this
assumption to be too restrictive since
it is known in many cases that we can approximate the non-linearities 
by polynomials~\cite[Appendix B]{ivkov2023semidefinite}.}

\vspace{-8pt}
\paragraph{The cavity method.} To explain the cavity method in one sentence, it allows you to assume
that ``loopy'' message-passing algorithms
on random dense graphs behave as if on a tree, gaining extra properties such as the independence of incoming messages.
It turns out that the assumption of being on a tree matches the restriction to tree-shaped monomials in $A$, leading to a way to rigorously implement
simple cavity method reasoning.

We formalize two types of cavity method arguments. For the first one, we introduce
a combinatorial notion of asymptotic equality $\eqinf$
which can rigorously replace heuristic approximations in the cavity method.

\begin{definition}[$\eqinf$, informal version of \cref{def:asymptotic-equality}]
    Let $x \eqinf y$ if $x-y$ is a sum of constantly many diagrams with cycles.
\end{definition}

As an application of this definition, we implement the cavity method argument
that belief propagation and approximate message passing are asymptotically equivalent for dense random matrices.

\begin{theorem}[Equivalence of BP and AMP; see \cref{thm:bp-amp}]\label{thm:bp-amp-intro}
    Let $m_{t}^{\textnormal{BP}}$ and $m_t^{\textnormal{AMP}}$ be the iterates
    of respectively the belief propagation
    and the approximate message passing iterations on the same non-linearities (see \cref{eq:bp-memory,eq:amp-output}). Then with high probability over $A$,
    \begin{align*}
        \norm{m_{t}^{\textnormal{BP}} - m_t^{\textnormal{AMP}}}_\infty = \widetilde{O}(n^{-\frac 12})\,.
    \end{align*}
\end{theorem}

We also use $\eqinf$ to prove a fundamental
assumption of the cavity method for belief
propagation iterations on dense models, 
namely that
the messages incoming at a vertex
are asymptotically independent.

\begin{theorem}[Asymptotic independence of incoming messages; see \cref{thm:cavity-formal}]
    Let $m_t^{\textnormal{BP}}$ be the iterates
    of a belief propagation iteration
    (\cref{eq:bp-memory}). For any $j\in [n]$, the incoming messages
    at $j$, $\{m^t_{i\to j}:i\in [n],i\neq j\}$, are
    asymptotically independent.
\end{theorem}

The second way that we formalize the
cavity method reasoning is through the idealized Gaussian dynamic $X_t$ in \cref{thm:intro-general-state}.
We recover the vanilla form of state evolution for approximate message passing,
since in this case, $X_t$ has a simple description.

\begin{theorem}[Asymptotic state of AMP; see \cref{thm:amp-state}]\label{thm:intro-amp-state}
    Consider the AMP iteration
    \begin{align}
        x_{t+1} = Af_t(x_t, \ldots, x_0) - \frac 1 n \sum_{s=1}^t \sum_{i=1}^n \frac{\partial f_t}{\partial x_s}(x_{t,i},\ldots,x_{0,i}) f_s(x_s, \ldots, x_0)\,.\label{eq:amp-intro}
    \end{align}
    The asymptotic state of $(x_0, x_1, \ldots)$ is a centered Gaussian vector
    $(X_0, X_1, \ldots)$ with covariances given by the recurrence, for all $s, t$,
    \[
        \E \left[X_s X_t\right]= \E \left[f_{s-1}(X_{s-1},\ldots,X_0)f_{t-1}(X_{t-1},\ldots,X_0)\right]\,.
    \]
\end{theorem}
The subtracted term in \cref{eq:amp-intro} is called the \emph{Onsager correction}
which, as we show, is carefully
designed to cancel out a backtracking
term in the asymptotic tree space 
(\cref{lem:taylor-expand}).

We emphasize that \cref{thm:bp-amp-intro,thm:intro-amp-state}
are {\em known}. They were originally predicted with the cavity
method, then later confirmed by rigorous proofs (in \cite{bayati2015universality} and \cite{bolthausen2014iterative, bayati2011dynamics, celentano2020estimation}, respectively).
The main message about our proofs is the new and quite comprehensive perspective obtained through the tree approximation,
providing a clear way in which GFOM algorithms on dense random inputs ``can be assumed to occur on a tree''. 

Finally, we provide an exposition in \cref{sec:montanari} of the breakthrough \emph{iterative AMP} algorithm devised by Montanari to compute ground states of the Sherrington--Kirkpatrick model \cite{montanari2021optimization, AM20, AMS20:pSpinGlasses}.
We explain from the diagram perspective how the algorithm is the optimal choice among algorithms which ``extract'' a Brownian motion from the input.

\vspace{-7pt}
\paragraph{Taking the tree approximation farther.}
The asymptotic theory above applies to an iterative algorithm running for a \textit{constant} number of iterations.
Although this ``short-time'' setting is used in a large majority of previous
works in this area, there is interest in extending the analysis
to, say, $O(\log n)$ iterations, which may be enough to capture planted recovery from
random initialization and distinct phases
of learning algorithms~\cite{LFW23}.

Can we use the tree-like Fourier characters to analyze the long-time behavior?
We show in \cref{sec:polyn-iterations} that some care needs to be taken.
First, we prove a positive result, that the tree approximation continues to hold for $n^{\Omega(1)}$ iterations for a simple belief propagation algorithm (debiased power iteration, or asymptotically equivalently, power iteration on the non-backtracking walk matrix).

\begin{theorem}[See \cref{thm:power-iteration}]
Generate $x_t \in \R^n$ from the debiased power iteration and let $\widehat{x}_t$ be the tree approximation to $x_t$.
Then there exist universal constants $c, \delta > 0$ such that for all $t \leq cn^\delta$,
\[\norm{x_t - \widehat{x}_t}_\infty \toas 0\,.\]
\end{theorem}

However, we also identify some problems with the technology which suggest that new ideas will be needed to completely capture the long-time setting. We observe that the asymptotic Gaussian classification theorem is no longer valid for diagrams of size $t \approx \log n$.
Finally, we identify a further threshold at $t \approx \sqrt{n}$ iterations
beyond which the tree approximation we use seems to break down.

\vspace{-7pt}
\paragraph{Conclusion.}
We demonstrate that for iterative algorithms
running on dense random inputs,
{trees are all you need}.
The tree-shaped Fourier diagrams form
an asymptotic
basis of independent Gaussian vectors
associated to an arbitrary Wigner matrix.
This basis seems extremely useful, and
we are not aware of any previous works on it.

We note that from the outset, it is not at all clear
how to find this basis. Individual monomials (i.e. Boolean Fourier characters) such as $A_{12}A_{23}A_{34}$ and $A_{12}A_{23}A_{13}$ have the same magnitude, and the asymptotic negligibility of the cyclic terms including $A_{12}A_{23}A_{13}$ only appears after summing up
the total contribution of all monomials with the same shape. Furthermore, grouping
terms in a different way does not identify our notion of tree approximation,
such as by allowing
repeated indices in \cref{eq:diagram-short}
(as done in~\cite{bayati2015universality,ivkov2023semidefinite}).
In this repeated-label representation, there
is no clear notion of tree approximation of
iterative algorithms (in fact, with this alternative definition, the iterates of
a GFOM can always
be represented {\em exactly} with trees!)
or of the simplified Gaussian dynamic on trees, which is central
to our approach.

As we show, the Fourier tree approximation
leads to streamlined proofs of several arguments
based on the cavity method. We believe that
this framework has potential to generalize well beyond the Wigner case and to address outstanding
open problems in the area---such as the long-time setting mentioned above.

\subsection{Related work}\label{sec:related_work}

\paragraph{Comparison with prior work.}
Our analysis is based on the recent ``low-degree paradigm'' in which
algorithms are analyzed as low-degree multivariate polynomial functions of the input \cite{kunisky2019notes}.
Several recent works have used a similar approach for iterative algorithms~\cite{bayati2015universality,montanari2022equivalence,ivkov2023semidefinite}, with subtle but crucial differences to our work.

Bayati, Lelarge, and Montanari~\cite{bayati2015universality} decompose the AMP iterates into certain ``nonreversing'' labeled trees. 
They also observe that the Onsager correction corresponds to a backtracking term. Montanari and Wein~\cite[Section 4.2]{montanari2022equivalence} introduce an orthogonal diagram
basis (similar to our Fourier diagram basis) 
in their proof that no low-degree polynomial
estimator
can outperform AMP for rank-1 matrix estimation.
Ivkov and Schramm~\cite{ivkov2023semidefinite} use a repeated-label representation to show that AMP can be
robustified.

Diagrammatically, the main advantage of our method is the precise choice of the Fourier diagram basis.
By summing over {injective} a.k.a. {self-avoiding} labelings $\varphi$ in \cref{eq:diagram-short}, each diagram exactly describes all monomials with a given shape.
When working with other polynomial basis, for example diagrams with repeated labels \cite{bayati2015universality,ivkov2023semidefinite} (see \cref{sec:repeated-labels}),
the key properties of the Fourier diagram basis (the family of asymptotically independent Gaussian vectors and the associated Gaussian dynamic) do not seem clearly visible.
In particular, previous work does not show 
the tree approximation.

Our results stated above which are cavity method-based reproofs of existing results are \cref{thm:bp-amp-intro},
which essentially follows from 
\cite[Proposition 3]{bayati2015universality},
and \cref{thm:intro-amp-state}, which 
was first proven by Bayati and Montanari \cite{bayati2011dynamics}.
Notably, Bayati, Lelarge, and Montanari \cite{bayati2015universality} use an approach based on the moment method as we do.
Their proof is somewhat more technical, it does not use the Fourier diagram basis,
and it is not able to clearly follow the 
simple cavity method argument that we
reproduce in \cref{sec:bp-amp-heuristic}.

    We also compare our state evolution formula for GFOM in \cref{thm:intro-general-state}
    with a state evolution formula for GFOM proven by \cite{celentano2020estimation}.
    They give a reduction from GFOM to AMP to derive a state evolution formula for GFOM.
    The corresponding description of the asymptotic state $X_0, \dots, X_t$ is inside a very compressed probability space
    generated by $t$ Gaussians with a certain covariance structure.
    
    Our description of the random variables $X_0, \dots, X_t$ (necessarily having the same distribution)
    has a simpler interpretation inside a larger probability space generated by $(Z_\sig^\infty)_{\sig \in \calS}\,$.
    Both descriptions of the
    asymptotic state $X_t$ are likely to be valuable for different purposes or explicit calculations.
    Our formulation of state evolution also includes the asymptotic independence of the trajectories of different coordinates.

\paragraph{Analyzing algorithms as low-degree polynomials.}
Our technical framework is adapted from the average-case analysis of
\emph{Sum-of-Squares} algorithms.
The Sum-of-Squares algorithm is a powerful meta-algorithm
for combinatorial optimization and statistical inference \cite{raghavendra2018high, FKP:semialgebraicproofs}.
Sum-of-Squares has been successfully analyzed on i.i.d. random inputs using \emph{graph matrices},
which are a Fourier basis for matrix-valued functions of a random matrix $A$
in the same way that our diagram basis is a basis
for vector-valued functions of $A$.

The theory appears much more pristine in the current setting, so we
hope that the current results can bring some new clarity to the technically challenging works on Sum-of-Squares.
Many key ideas on graph matrices are present in a pioneering work
by Barak et al. which analyzes the Sum-of-Squares algorithm
for the Planted Clique problem~\cite{BHKKMP16:PlantedClique} (building on earlier work \cite{deshpande2015improved, meka2015sum, hopkins2018integrality}).
Analytical ideas were subsequently isolated by Ahn, Medarametla, and Potechin~\cite{AMP20} and Potechin and Rajendran~\cite{PR20:Machinery, potechin2022sub} and developed in several
more works~\cite{GJJPR20:SherringtonKirkpatrickPlantedAffinePlanes, RajendranTulsiani20, JPRTX21:SparseIndependentSet,
JP22:InnerProductPolynomials, jones2022symmetrized, jones2023sum, xu2024ultrasparse}.
Several recent works have made explicit connections between AMP, Sum-of-Squares, and low-degree polynomials \cite{montanari2022equivalence, ivkov2023semidefinite, singh2023highentropy, singh2024sum}.
Another similar class of diagrammatic techniques are \emph{tensor networks} \cite{moitra2019spectral, kunisky2024tensor}.

\vspace{-8pt}
\paragraph{Statistical physics and the cavity method.}
The cavity and replica methods are widely used in statistical physics to compute the free energy, complexity, etc. of Gibbs distributions on large systems,
or similarly to compute the satisfiability threshold, number of solutions, etc. for many non-convex random optimization problems.
For an introduction to statistical physics methods in computer science, we recommend the surveys~\cite{martin2001statistical, zdeborova2016statistical, gabrie2020mean}, the book~\cite{MezardMontanari}, and the 40-year retrospective~\cite{charbonneau2023spin}.
The cavity method is described in \cite{mezard2003cavity} and \cite[Part V]{MezardMontanari}.

Rigorously verifying the predictions of the physical methods has
been far from easy for mathematicians.
To highlight some major landmarks in the literature over the past decades, tour-de-force proofs
of the Parisi formula for the free energy of the
SK model were developed by
Talagrand \cite{Talagrand06, talagrandBook} and Panchenko \cite{panchenko2013sherrington}.
Ding, Sly, and Sun~\cite{DSS16:MaxIndependentSets, ding2014satisfiability, ding2015proof} identified
the satisfiability threshold for several random optimization problems including $k$-SAT with large $k$.
Ding and Sun~\cite{ding2019capacity} and Huang~\cite{huang2024capacity}
rigorously analyze the storage capacity of the Ising perceptron, assuming a numerical condition.

Note that the results above are strictly outside the regime of the
current work. They require the replica method in ``replica symmetry breaking'' settings,
whereas we study the simpler but related cavity method in the replica symmetric setting.
$k$-SAT is also a sparse (a.k.a. dilute) model
whereas our results are for dense (a.k.a. mean-field) models.
Despite these differences, our results tantalizingly suggest that it may be possible
to validate the physical techniques in a more direct and generic way than taken by current approaches.

Other authors have also directly considered the cavity assumption,
albeit using a less combinatorial approach.
Both proofs of the Parisi formula implement analytic forms of the cavity calculation
(\cite[Section 1.6]{talagrandBook} and \cite[Section 3.5]{panchenko2013sherrington}).
The cavity method can also be partially justified for sparse models
in the replica symmetric regime using that the interactions
are locally treelike with high probability \cite{bayati2006rigorous, coja2017information}.

Diagrammatic methods are common in physics,
and in fact they have been used in the vicinity of belief propagation
even since a seminal 1977 paper by Thouless--Anderson--Palmer~\cite{thouless1977solution} which introduced the TAP equations of \cref{eq:tap}.
A version of the tree approximation actually appears briefly in their diagrammatic formula for the free energy of the SK model in Section 3. 
However, it has not been clear how or whether these
arguments could be made rigorous, and to date, rigorous proofs have
not directly followed these approaches.

\vspace{-8pt}
\paragraph{Belief propagation and AMP.}
Belief propagation originates
in computer science and statistics from Pearl \cite{pearl1988probabilistic}.
In the current setting, we can view the underlying graphical model as the complete graph, with correlations between the variables induced by the random matrix $A$.
State evolution was first predicted for BP algorithms in this setting
by Kabashima \cite{kabashima2003cdma} and Donoho--Maleki--Montanari \cite{donoho2009message}. Since the first rigorous proof of state evolution by Bolthausen \cite{bolthausen2014iterative}, his Gaussian conditioning technique has been extended to prove state evolution for
many variants of AMP \cite{bayati2011dynamics, javanmard2013state, ma2017analysis, takeuchi2019rigorous, berthier2020nonseparable, takeuchi2019unified, AMS20:pSpinGlasses, takeuchi2021bayes, lu2021householder, feng2022unifying, fan2022approximate, gerbelot2023graph, huang2023optimization}.

A notably different proof of state evolution by Bayati, Lelarge, and Montanari \cite{bayati2015universality} uses a moment-based approach
which is closer to ours
(see also follow-up proofs~\cite{chen2021universality, dembo2021diffusions, wang2022universality, dudeja2023universality}).
These proofs and also ours show universality statements which the Bolthausen conditioning method cannot handle.

All of the above works restrict themselves to a constant number of
iterations,
although some recent papers push the analysis of AMP in some settings to a superconstant number of iterations 
\cite{rush2018finite, cademartori2023non, wu2023lower,wu2024sharp}.
Very recently, \cite{LW22, LFW23} managed to analyze $t = \widetilde{\Om}(n)$ iterations of AMP in the spiked Wigner model.
This last line of work is especially intriguing, given that our
approach seems to break down at $t \approx \sqrt{n}$ (\cref{sec:combi-obstructions}).

The perspective that we take is slightly different from most of these papers.
Whereas previous works analyze the asymptotic \emph{distribution} of the AMP iterates over the randomness of $A$,
we give an explicit function $\widehat x_t$ which exactly satisfies a ``Gaussian dynamics'' and asymptotically approximates the iterates.
This general approach provides more information and we hope that it has increased potential
for generalization.

On first-order iterations which are not BP/AMP algorithms, a smaller number of physical analyses have
been performed using the more general techniques of \emph{dynamical mean field theory} \cite{martin1973statistical}.
We refer to the survey \cite{gabrie2020mean}.
Most analyses rely on heuristic arguments, although some more recent works \cite{celentano2021high, gerbelot2022rigorous, liang2022high} prove rigorous results.

Finally, we note that the tree approximation
bears similarities to the suppression of
noncrossing partitions in free probability~\cite{speicherBook}.
Unlike the traditional viewpoint of free
probability, the combinatorial cancellations
behind the tree approximation
occur directly on the trajectory of 
random objects (the
iterates of the algorithm), and
not only for averaged quantities associated with them.

\subsection{Organization of the paper}

After background preliminaries in \cref{sec:preliminaries},
we introduce the diagrams in \cref{sec:diagram-calculus} and describe
their key properties without proofs.

In \cref{sec:combinatorial}, we present the full diagram analysis: we 
define the key notion of asymptotic
equality $\eqinf$, and we prove three central theorems: the classification of the diagrams (\cref{thm:classification}), the tree approximation for GFOMs (\cref{thm:state-evolution}), and a general state evolution
formula (\cref{thm:general-state-evolution}).

In \cref{sec:cavity-method}, we demonstrate
the connection with the cavity method by
proving the equivalence between
belief propagation and approximate
message passing (\cref{thm:bp-amp}), the
independence of incoming messages in 
belief propagation
(\cref{thm:cavity-formal}), the state
evolution formula for AMP (\cref{thm:amp-state}),
and the analysis of the iterative AMP algorithm
of Montanari.

Finally, \cref{sec:polyn-iterations} investigates algorithms running for a large number of iterations.

\cref{app:non-asymptotic,app:non-asymptotic-analysis,sec:star} contain omitted proofs and calculations.

\vspace{-8pt}
\paragraph{Acknowledgments.}
We are pleased to thank numerous colleagues for feedback and discussions as this work developed, including Giorgi Kanchaveli, Carlo Lucibello, Enrico Malatesta, Goutham Rajendran, Alon Rosen, and Riccardo Zecchina.
Work supported in part by the European Research Council (ERC) under the European
Union’s Horizon 2020 research and innovation programme (grant agreement Nos. 834861 and 101019547).
CJ is also a member of the Bocconi Institute for Data Science and Analytics (BIDSA).

\section{Preliminaries}
\label{sec:preliminaries}

To maintain generality, we specify the input (a random matrix) and the algorithm (a first-order iteration), but we do not specify an objective/energy function, and for this reason our results are in the flavor of random matrix theory. 
While the setting of this paper is a null model without any hidden signal, 
we expect that our techniques can also be applied to planted recovery problems.
A concrete algorithmic application to 
keep in mind in the null model
is the optimization of
random degree-2 polynomials that we revisit in
\cref{sec:montanari}.

Our results will apply universally to a Wigner random matrix model (they hold regardless
of the specific choice of $\mu, \mu_0$ below).

\begin{assumption}[Assumptions on matrix entries]
\label{assump:A-entries}
    Let $\mu$ and $\mu_0$ be two subgaussian\footnote{A distribution $\mu$ on $\R$ is \emph{subgaussian} if there exists a constant $C > 0$ such that for all $q \in \N$, $\E_{X \sim \mu}[|X|^q] \leq C^q q^{q/2}$.} distributions on $\R$ such that $\E_{X \sim \mu}[X] = 0$ and $\E_{X \sim \mu}[X^2] = 1$.
    
    Let $A$ be a random $n\times n$ symmetric matrix with independent entries (up to the symmetry) which are either $\sqrt{n} A_{ii} \sim \mu_0$ on the diagonal or $\sqrt{n} A_{ij} \sim \mu$
    off the diagonal.
\end{assumption}

The subgaussian assumption on
$\mu$ and $\mu_0$ can be relaxed to require
only the existence of the
$q$-th moment of $\mu$ for some 
large enough constant
$q\in \N$ that depends only 
on the number of iterations
and the degree of the nonlinearities appearing
in the algorithm. In this case, our statements
of the form ``$\|x_n-y_n\|_\infty=\widetilde O(n^{-1/2})$ with high probability''\footnote{We say a
sequence of events $(A_n)_{n\ge 0}$ occurs with high probability if $\Pr(A_n)\ge 1-1/\text{poly}(n)$.} weaken to ``$\|x_n-y_n\|_\infty\toas 0$''.

\begin{definition}[Convergence of random vectors]
    Let $(X_n)_{n\in \N}$ and $Z$ be random vectors. 
    \begin{itemize}
        \item We write $X_n\toas Z$ if $X_n$ converges to $Z$ almost surely, i.e. $\lim_{n \to \infty}X_n$ exists and equals $Z$ with probability 1.
        \item We write $X_n\tod Z$ if $X_n$ converges to $Z$ in distribution, i.e. for every real-valued bounded continuous function $f$, $\lim_{n\to\infty} \E f (X_n)$ exists and equals $\E f(Z)$.
    \end{itemize}
\end{definition}

We can derive convergence in distribution 
of random vectors by computing their
moments.

\begin{lemma}[Method of moments {\cite[Theorems 29.4, 30.1, and 30.2]{billingsley}}]
\label{lem:method-of-moments}
    Let $X_n\in\R^d$ for $n\in \N$ and $Z\in\R^d$ be random
    vectors such that for any $q_1,\ldots,q_d\in\N$,
    \[
        \E \left[\prod_{i=1}^d X_{n,i}^{q_i}\right]\underset{n\to\infty}{\longrightarrow} \E \left[\prod_{i=1}^d Z_i^{q_i}\right].
    \]
    Suppose that for all $i\in [n]$, $Z_i$ 
    has the distribution of a polynomial in
    Gaussian random variable. Then $X_n\tod Z$.  
\end{lemma}

We will refer to the generalized (probabilist's) Hermite polynomials as $h_k(\,\cdot\,; \sigma^2)$, where $h_k$ is the degree-$k$ monic orthogonal polynomial for $\calN(0, \sigma^2)$. If $Z_i$ is an independent $\calN(0, \sigma_i^2)$ random variable for all $i\in\calI$, then
$\left(\prod_{i \in \calI} h_{k_i}(Z_i;\sigma_i^2)\right)_{k\in\N^\calI}$ is an orthogonal basis for polynomials in $(Z_i)_{i\in\calI}$ with respect to the expectation over $(Z_i)_{i\in\calI}$.

The Gaussian distribution and Hermite polynomials have combinatorial interpretations related to matchings.

\begin{lemma}
    \label{lem:gaussian-moments}
    For $Z \sim \calN(0, \sigma^2)$, 
    \[
        \E\left[Z^{q}\right] = \abs{\calP\calM(q)}\sig^{\frac q 2} = \begin{cases}
        (q-1)!!\cdot \sig^{\frac q 2} & \text{if $q$ is even}\\
        0 & \text{if $q$ is odd}
    \end{cases}\,,
    \]
    where $\calP\calM(q)$ is the set of perfect matchings on $q$ objects and $(q-1)!!=\frac {q!} {2^{q/2} (q/2)!}$.
\end{lemma}

\begin{lemma}[{\cite[Theorem 3.4 and Example 3.18]{Janson:GaussianHilbertSpaces}}]
\label{fact:hermite-matchings}
    For all $k\ge 0$ and $x\in\R$,
    \[
        h_k(x; \sig^2) = \sum_{M \in \calM(k)} (-1)^{|M|} \sig^{2|M|} x^{k-2\abs{M}}\,,
    \]
    where $\calM(k)$ is the set of (partial) matchings on $k$ objects (including the empty matching and perfect matchings).
\end{lemma}

\begin{lemma}[{\cite[Theorem 3.15 and Example 3.18]{Janson:GaussianHilbertSpaces}}]\label{fact:hermite-product}
    For any $k_1,\ldots,k_l\ge 0$ and $x\in\R$,
        \begin{align*}
        h_{k_1}(x; \sigma^2)\cdots h_{k_\el}(x; \sigma^2) = \sum_{M \in \calM(k_1,\dots, k_\el)} h_{k - 2\abs{M}}(x; \sigma^2) \sigma^{2\abs{M}}\,,
        \end{align*}
        where $\calM(k_1, \dots, k_\el)$ is the set of (partial)
        matchings on $k=k_1 + \cdots + k_\el$ objects divided into $\el$ blocks of sizes $k_1,\ldots,k_\el$
        such that no two elements from the same block are matched.
\end{lemma}

Finally, we recall:

\begin{lemma}[Gaussian integration by parts]
    \label{lem:gaussian-ipp}
    Let $(Z_1,\ldots,Z_k)$ be a centered Gaussian vector. Then for all smooth $f:\R^k\to\R$,
    \[
        \E \left[Z_1 f(Z_1,\ldots,Z_k)\right] = \sum_{i=1}^k \E \left[Z_1 Z_i\right]\E \left[\frac{\partial f}{\partial z_i}(Z_1,\ldots,Z_k)\right].
    \]
\end{lemma}

\section{The Diagram Basis}
\label{sec:diagram-calculus}

Here we give the key properties of the Fourier diagrams on a high level, delaying formal statements and proofs to the next section.
\begin{itemize}
    \item In \cref{sec:example}, we give an example.
    \item In \cref{sec:nonasymptotic-diagram-basis}, we define the class of diagrams and describe their behavior both for fixed $n$ and in the limit $n\to\infty$. 
    \item In \cref{sec:calculus-rules}, we summarize how iterative algorithms behave asymptotically.
    \item In \cref{sec:sym-Fourier-analysis}, we explain how the diagram basis can be derived from standard discrete Fourier analysis.
\end{itemize}

\subsection{Example of using diagrams}
\label{sec:example}

We show how to compute the vector $A(A\vec{1})^{2}$ in the diagram basis, where $\vec{1} \in \R^n$ denotes the all-ones vector and the square function is applied componentwise.
Calculation with diagrams is a bit like a symbolic version of the trace method from random matrix theory \cite{bordenave2019lecture}.

For simplicity, we assume in this subsection that $A$ satisfies \cref{assump:A-entries} with $A_{ii}=0$ for all $i\in [n]$.

We will use {rooted} multigraphs to represent vectors.\footnote{Graphs with multiple
roots can be used to represent matrices and tensors, although
we will not need those here.}
Multigraphs may include multiedges and self-loops.
In our figures, the root will be drawn as a circled vertex $\rootpic$. The vector $\vec{1}$ will correspond to the singleton graph with one vertex (the root): $\rootpic$.
Edges will correspond to $A_{ij}$ terms.

The vector $A\vec{1}$ will be represented by the graph consisting of a single edge, with one of the endpoints being the root:

\newcolumntype{B}{>{\(\displaystyle}c<{\)}@{}}

\begin{center}
\begin{tabular}{BBB}
    (A\vec{1})_i = \sum_{j = 1}^n A_{ij} &{}={}& \sum_{\substack{j = 1\\i,j\text{ distinct}}}^n A_{ij}\\
    &\equiv& {\vcenter{\hbox{\includegraphics[height=3ex]{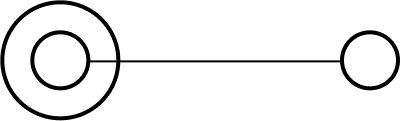}}}}
\end{tabular}
\end{center}
where the second equality uses the assumption that $A$ has zero diagonal.
Now to apply the square function componentwise, we can decompose:
\begin{center}\begin{tabular}{BBBBB}
(A\vec{1})^{2}_i &{}={}&  \sum_{\substack{j,k=1\\i,j,k\text{ distinct}}}^n A_{ij} A_{ik} &{}+{}& \sum_{\substack{j=1\\i,j\text{ distinct}}}^n A_{ij}^2\\
&\equiv& \vcenter{\hbox{\includegraphics[height=5ex]{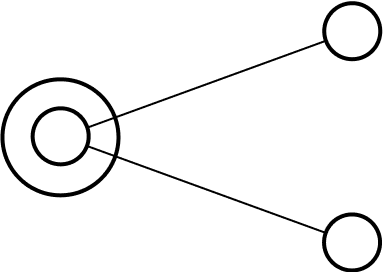}}} &+& \vcenter{\hbox{\includegraphics[height=3ex]{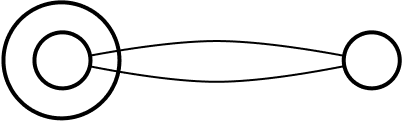}}}
\end{tabular}\end{center}

Moving on, we apply $A$ to this representation by casing on whether the new index $i$ matches one of the previous indices. We group terms together using the symmetry of $A$ and the fact that $A_{ii}=0$.

\begin{center}\begin{tabular}{BBBBBBBBB}
(A(A\vec{1})^{2})_i &{}={}& \sum_{\substack{j,k,\el=1\\i,j,k,\el\text{ distinct}}}^n A_{ij} A_{jk} A_{j\el} &{}+{}& 2\sum_{\substack{j,k=1\\i,j,k\text{ distinct}}}^n A_{ij}^2 A_{jk}&{}+{}& \sum_{\substack{j,k=1\\i,j,k\text{ distinct}}}^n A_{ij} A_{jk}^2 &{}+{}& \sum_{\substack{j=1\\i,j\text{ distinct}}}^n A_{ij}^3\\    &\equiv& \vcenter{\hbox{\includegraphics[height=5ex]{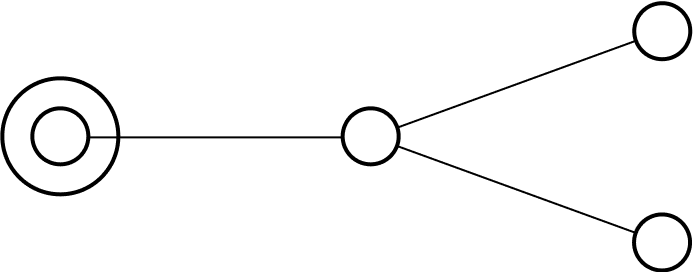}}} &+2&\; \vcenter{\hbox{\includegraphics[height=3ex]{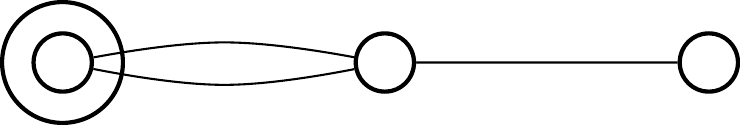}}} &+&  \vcenter{\hbox{\includegraphics[height=3ex]{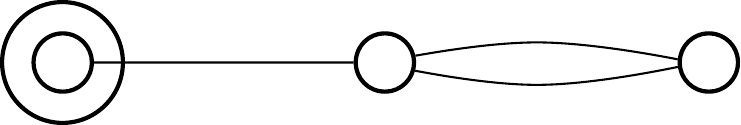}}} &+&  \vcenter{\hbox{\includegraphics[height=3.5ex]{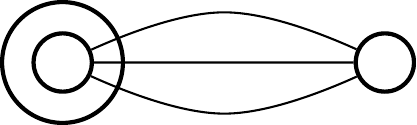}}}
\end{tabular}\end{center}

This is the non-asymptotic Fourier diagram representation of $A(A\vec{1})^2$.
In the limit $n \to \infty$, only some of these terms contribute
to the \textit{asymptotic} Fourier diagram basis representation. Asymptotically, \textit{hanging} double edges can be removed from a diagram\footnote{To be convinced of this, the reader can think of the case where the entries of $A$ are uniform $\pm \frac 1 {\sqrt n}$.},
so that the third diagram in the representation above satisfies, as $n\to\infty$,
\begin{align*}
    \vcenter{\hbox{\includegraphics[height=3ex]{images-pdf/single-edge-double-edge.pdf}}} \;\eqinf\; \vcenter{\hbox{\includegraphics[height=3ex]{images/edge-rooted.png}}}\,.
\end{align*}

The second and fourth diagrams in the representation of $A(A\vec{1})^2$ have entries on the scale $O(n^{-1/2})$ and so they will be dropped from the asymptotic diagram representation. In total,
\[
    A(A\vec{1})^{2} \;\eqinf\; \vcenter{\hbox{\includegraphics[height=5ex]{images/1,2-tree.png}}} \;+\; \vcenter{\hbox{\includegraphics[height=3ex]{images/edge-rooted.png}}}\,.
\]
We will show that as $n\to\infty$, the left diagram becomes a Gaussian vector with independent entries of variance $2$, and the right diagram becomes a Gaussian vector with independent entries of variance $1$. In fact, these $2n$ entries are asymptotically jointly independent. It can be verified numerically that approximately for large $n$, $A(A\vec{1})^{2}$ matches the sum of these two random vectors, the histogram of each vector's entries is Gaussian, and the vectors are approximately orthogonal.

\subsection{Properties of the diagram basis}
\label{sec:nonasymptotic-diagram-basis}

\begin{definition}
    \label{def:diagram}
    A \emph{Fourier diagram} is an unlabeled undirected multigraph $\al = (V(\al), E(\al))$
    with a special vertex labeled $\rootpic$ which we call the \textit{root}.
    No vertices may be isolated except for the root. We let $\calA$ be the set of all Fourier diagrams.
\end{definition}

\begin{definition}[$Z_\al$]
    \label{def:Zal}
    For a Fourier diagram $\al\in\calA$ with
    root $\rootpic$, define the vector $Z_\al \in \R^n$ by
    \[
        Z_{\al,i} = \sum_{\substack{\ph: V(\al) \to [n]\\\ph \textnormal{ injective}\\\ph(\smallrootpic) = i}} \prod_{\{u,v\} \in E(\al)} A_{\ph(u)\ph(v)}\,,\quad \text{ for all $i\in [n]$\,.}
    \]
\end{definition}

Among all Fourier diagrams, the ones corresponding to trees play a special role. They will constitute the \textit{asymptotic Fourier diagram basis}.

\begin{definition}[$\calS$ and $\calT$]
    \label{def:diagram-notations}
    Let $\calS$ be the set of unlabeled rooted trees such
    that the root has exactly one subtree (i.e. the root has degree 1).
    Let $\calT$ be the set of all unlabeled rooted trees (non-empty, but allowing the singleton).
\end{definition}

\begin{definition}[Proper Fourier diagram]
    A proper Fourier diagram is a Fourier diagram with no multiedges or self-loops (i.e. a rooted simple graph).
\end{definition}

For \emph{proper} Fourier diagrams $\al \in \calA$, the following properties of $Z_\alpha$
hold non-asymptotically i.e. for arbitrary $n$: 
\begin{enumerate}[(i)]
    \item $Z_\alpha$ is a multilinear polynomial in the entries of $A$
    with degree $|E(\alpha)|$ (or $Z_\al = 0$ when $|V(\al)| > n$).
    \item 
    $Z_\alpha$ has the symmetry that $Z_{\alpha,i}(A) = Z_{\alpha, \pi(i)}(\pi(A))$ for all permutations $\pi \in S_n$, where $\pi$ acts on $A$ by permuting the rows and columns simultaneously.
    \item 
    For each $i \in [n]$, the 
    set $\{Z_{\al, i}:\text{proper Fourier diagram }\al \in \calA\}$
    is orthogonal
    with respect to the expectation over $A$. 
    \item In fact, $Z_\alpha$ is a symmetrized
    multilinear Fourier character (see \cref{sec:sym-Fourier-analysis}). This implies the previous properties and it shows that
    the proper diagrams are an orthogonal basis for a class of
    symmetric functions of $A$. 
\end{enumerate}

We represent the algorithmic state as a Fourier diagram expression of the form $x = \sum_{\al \in \calA} c_\al Z_\al$.
To multiply together or apply algorithmic operations on a diagram expression, we case on
which indices repeat, like in the example in \cref{sec:example}.
See \cref{lem:mat-mul,lem:diagrams-product} in \cref{sec:derivation-asymptotic} for a formal derivation of these rules.

Now we turn to the asymptotic properties.
The constant-size tree diagrams $(Z_\tau)_{\tau \in \calT}$ exhibit the following key
properties in the limit $n \to \infty$ and with respect to the randomness of $A$.
\begin{enumerate}[(i)]
    \item The coordinates of $Z_\tau\in \R^n$ for any $\tau\in\calT$ are asymptotically independent and identically distributed.
    \item The random variables $Z_{\sig,1}$ for $\sigma \in \calS$ (the tree diagrams with one subtree) are asymptotically independent Gaussians
    with variance $\abs{\Aut(\sig)}$, where $\Aut(\sig)$ are
    the graph automorphisms of $\sig$ which fix the root.
    \item 
    The random variable $Z_{\tau, 1}$ for $\tau \in \calT$ (the tree diagrams with multiple subtrees) is asymptotically equal to the multivariate Hermite polynomial $\prod_{\sig \in \calS} h_{d_\sig}(Z_{\sig,1}; \abs{\Aut(\sig)})$
    where $d_\sig$ is the number of children of the root whose subtree (including the root) equals $\sig \in \calS$.
\end{enumerate}
\noindent
The remaining Fourier diagrams not in $\calT$ can be understood using the further asymptotic properties:

\begin{itemize}
    \item[(iv)] For any diagram $\al \in \calA$, if $\al$ has a \emph{hanging double edge} i.e. a double edge with one non-root endpoint of degree exactly 2,
    letting $\al_0$ be the diagram with the hanging double edge and hanging vertex removed,
    then $Z_{\al}$ is asymptotically equal to $Z_{\al_0}$.
    For example, the following diagrams are asymptotically equal:
    \begin{center}
    \begin{tabular}{ccccc}
         $\vcenter{\hbox{\includegraphics[height=3ex]{images/singleton.png}}}$ &  $\eqinf$ & $\vcenter{\hbox{\includegraphics[height=3ex]{images/double-edge.png}}}$ & $\eqinf$ & $\vcenter{\hbox{\includegraphics[height=15ex]{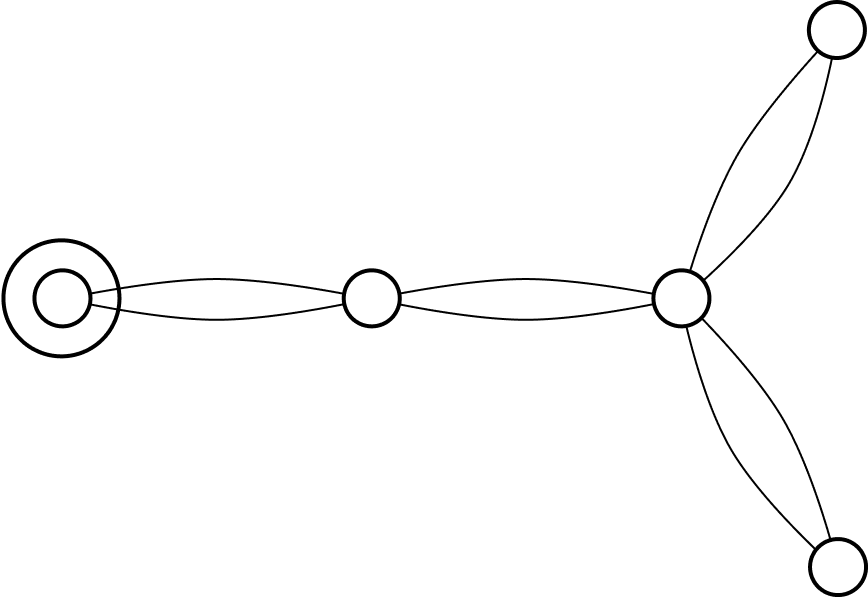}}}$\\
         1 & $\approx$ &  $\displaystyle\sum_{\substack{j=1\\i \neq j}}^n A_{ij}^2$ & $\approx$ & $\displaystyle\sum_{\substack{j,k,\el,m = 1\\i,j,k,\el,m\text{ distinct}}}^n A_{ij}^2 A_{jk}^2 A_{k\el}^2 A_{km}^2$
    \end{tabular}
    \end{center}

    \item[(v)] For any \emph{connected} $\al \in \calA$, if removing the hanging trees of double
    edges from $\al$ creates a diagram in $\calT$, then by the previous property, 
    $Z_\al$ is asymptotically equal to that diagram.
    If the result is not in $\calT$, then $Z_\al$ is asymptotically negligible.
    \item[(vi)]
    The disconnected diagrams have only a minor and negligible role in the algorithms that we consider.
    See \cref{sec:classification} for the description
    of these random variables.
\end{itemize}

To summarize the properties, given a sum $x$ of connected diagrams,
by removing the hanging double trees,
and then removing all diagrams not in $\calT$, the expression admits an \emph{asymptotic} Fourier diagram
basis representation of the form
\begin{align}
    x \eqinf \sum_{\tau \in \calT} c_\tau Z_\tau\,,\label{eq:asymptotic-expansion}
\end{align}
for some coefficients $c_\tau\in \R$ independent of $n$ and $A$.
We call this the \emph{tree approximation} to $x$.
Note that all tree diagrams have order 1 variance regardless of their size, which can be counter-intuitive.

\subsection{Asymptotic state evolution}
\label{sec:calculus-rules}

The main appeal of the tree approximation
in \cref{eq:asymptotic-expansion} is that when
restricted to the tree-shaped diagrams, the
GFOM operations have a very simple interpretation:
they implement an idealized {\em Gaussian 
dynamics} which we describe now.

The idealized GFOM moves through an ``asymptotic Gaussian probability space''
which is naturally the one corresponding to the $n \to \infty$ limit
of the diagrams.
Based on the diagram classification, this consists of an infinite family of
independent Gaussian vectors $(Z_\sig)_{\sig \in \calS}\,.$
However, due to symmetry, all of the coordinates follow the same dynamic, so we can compress the representation of the dynamic down to a one-dimensional random variable $X_t$ which is the asymptotic distribution of each coordinate $x_{t,i}\,$.
We call $X_t$ the \emph{asymptotic state} of $x_t$.

For example, Approximate Message Passing (AMP)
is a special type of GFOM whose
iterates are asymptotically Gaussian i.e. $X_t$ is a Gaussian random variable for all $t$ (in general GFOMs, although $X_t$ is defined in terms of Gaussians, it is not necessarily Gaussian).

With that prologue, the algorithmic operations restricted to the trees and the corresponding evolution of the asymptotic state $X_t$ are as follows.
Two important operations on a tree-shaped diagram are extending/contracting
the root by one edge.

\begin{definition}[$+$ and $-$ operators]
    \label{def:plus-minus-operators}
    We define $+:\calT\to\calS$ and $-:\calS\to\calT$ by:
    \begin{itemize}
        \item If $\tau\in\calT$, let $\tau^+$ be the diagram
        obtained by extending the root by one edge (i.e.
        adding one new vertex and one edge connecting
        it to the root of $\tau$, and re-rooting 
        $\tau^+$ at this new vertex).
        \item If $\tau\in \calS$, let $\tau^-$
        be the diagram obtained by contracting the 
        root by one edge (i.e. removing the root vertex
        and the unique edge from it, and re-rooting
        $\tau^-$ at the endpoint of that edge).
    \end{itemize}
\end{definition}

Recall that the possible operations of a GFOM 
are either multiplying the state by $A$ or applying
a function componentwise. 
The effect of these two operations on the tree-shaped
diagrams are:

\begin{itemize}

    \item If $\sigma\in\calS$, then
    $A Z_\sigma$ is asymptotically the sum of the diagrams
    $\sigma^+$ and $\sigma^-$ obtained by respectively 
    extending and
    contracting the root by one edge. For example,
    \[A \ \ \times\ \  \vcenter{\hbox{\includegraphics[height=5ex]{images/1,2-tree.png}}} \ \ \eqinf \ \ \vcenter{\hbox{\includegraphics[height=5ex]{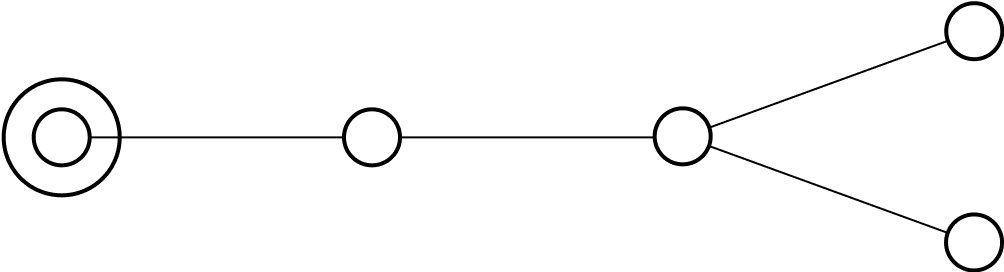}}} \ \  + \ \  \vcenter{\hbox{\includegraphics[height=5ex]{images/2tree.png}}}\]
    
    If $\tau \in \calT \setminus \calS$, then
    $AZ_\tau$ is asymptotically only the $\tau^+$ term. For example,
    \[A \ \ \times\ \  \vcenter{\hbox{\includegraphics[height=5ex]{images/2tree.png}}} \ \ \eqinf \ \ 
    \vcenter{\hbox{\includegraphics[height=5ex]{images/1,2-tree.png}}}\]

    \item From the classification of diagrams, if $\tau\in\calT$ consists of $d_\sigma$ copies of $\sigma\in\calS$, then 
    \begin{align}
        \prod_{\sigma\in\calS} h_{d_\sigma}(Z_\sigma;\left|\Aut(\sigma)\right|)\eqinf Z_\tau \,.\label{eq:hermite}
    \end{align}
    Therefore,
    to compute $f(Z_\sigma:\sigma\in \calS)$, we expand $f$ in the Hermite
    polynomial basis associated to $\calS$, and
    apply the rule \cref{eq:hermite} to all
    the terms. For example,

    \begin{center}
        \includegraphics[scale=0.6]{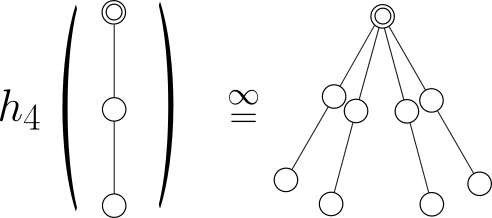}
    \end{center}
\end{itemize}

These operations correspond to the following Gaussian dynamic.

\begin{definition}[Asymptotic Gaussian space, $\Om$]
    Let $(Z_\sigma^\infty)_{\sigma\in\calS}$ be a
    set of independent centered (one-dimensional)
    Gaussian random
    variables with variances $\Var(Z_\sigma^\infty)=\left|\Aut(\sigma)\right|$.

    If $\tau\in\calT$ can be decomposed as $d_\sigma$ copies of each
    $\sigma\in\calS$ branching from the root, we define 
    \[
        Z_\tau^\infty=\prod_{\sigma\in\calS} h_{d_\sigma}(Z^\infty_\sigma;\left|\Aut(\sigma)\right|)\,.
    \]
    We call {\em asymptotic states}
    the elements in the linear span of $(Z^\infty_\tau)_{\tau\in\calT}$. We can view them both as polynomials in the formal
    variables $(Z_\sigma^\infty)_{\sigma\in\calS}$ and
    as real-valued random variables.
    The set of asymptotic states is denoted $\Om$.
\end{definition}

\begin{definition}[Asymptotic state]\label{def:asymptotic-state}
    If $x\in\R^n$ is such that $x\eqinf \sum_{\tau\in\calT} c_\tau Z_\tau\,$, we
    define the {\em asymptotic state} of $x$ by
    \[
        X = \sum_{\tau\in\calT} c_\tau Z_\tau^\infty\,.
    \]

\end{definition}

The state evolution of the algorithm can now be described concisely as:
\begin{itemize}
    \item If $x_t$ has asymptotic state $X_t$, then the asymptotic state of $Ax_t$ is $X_t^+ + X_t^-$.
    Here we extend the $+$ and $-$ operators linearly to sums of
    $Z_\tau$ or $Z^\infty_\tau$
    (let $Z_\tau^- = (Z^\infty_\tau)^- = 0$ if $\tau \in \calT \setminus \calS$).

    \item If $x_{t-1}, \ldots, x_0$ have asymptotic
    states $X_{t-1}, \ldots, X_0$ and $f$ is any polynomial, then the asymptotic state of $f(x_{t-1},\ldots,x_0)$ is $f(X_{t-1},\ldots,X_0)$.
\end{itemize}

\subsection{Perspective: equivariant Fourier analysis}
\label{sec:sym-Fourier-analysis}

The Fourier diagrams form an orthogonal basis that can be derived in a mechanical way using \emph{symmetrization}.

We can use Fourier analysis to express a function or algorithm
with respect to a natural basis.
The unsymmetrized underlying analytical space consists of functions
of the $n^2$ entries of $A$.
Since the entries of $A$ are independent,
the associated Fourier basis is the product
basis for the different entries.
When $A \in \{-1, 1\}^{n \times n}$ is a Rademacher random matrix, the Fourier characters are the multilinear
monomials in $A$.
An arbitrary function $f:\{-1,1\}^{n\times n}\to \R$ is then expressed as
\[f(A) = \sum_{\al \subseteq [n] \times [n]} c_\al \prod_{(i,j) \in \al} A_{ij}\,,\]
where $c_\al$ are the Fourier coefficients of $f$.
When $A$ is a symmetric matrix with zero diagonal, we only need Fourier characters
for the top half of $A$, and the basis simplifies to $\al \subseteq \binom{[n]}{2}$.
That is, the possible $\al$ can be interpreted combinatorially as graphs on the vertex set $[n]$.

An observation that allows us to significantly simplify
the representation is that many of the Fourier coefficients are equal
for \emph{$S_n$-equivariant} algorithms. A function $f : \R^{n \times n} \to \R$ is $S_n$-equivariant if it satisfies $f(\pi(A)) = f(A)$ or if $f : \R^{n \times n} \to \R^n$ satisfies $f(\pi(A)) = \pi(f(A))$ where $\pi$ acts on $A$ by permuting the rows and columns simultaneously.
For scalar-valued functions, considering the action of $S_n$ on the vertex
set of the Fourier characters $[n]$, any two Fourier characters $\al, \beta$ which are in the same
orbit will have the same Fourier coefficient.
Equivalently, if $\al$ and $\beta$ are isomorphic as graphs,
then their Fourier coefficients are the same.
By grouping together all isomorphic Fourier characters, we obtain the symmetry-reduced representation
defining the Fourier diagram basis,
\[f(A) = \sum_{\text{nonisomorphic }\al \subseteq \binom{[n]}{2}} c_\al \left(\sum_{\substack{\text{injective }\ph: V(\al) \to [n]}} \prod_{\{u,v\} \in \al} A_{\ph(u)\ph(v)}\right)\,.\]

Thus by construction, the diagrams are an orthogonal
basis for symmetric low-degree polynomials of $A$. We use this to derive some simple facts in \cref{sec:fourier}.
Asymptotic independence of the Gaussian diagrams
can be predicted based on the fact that the diagrams
are an \emph{orthogonal} basis, and orthogonal Gaussians are independent
(thus we expect a set of independent Gaussians to appear from other types of i.i.d. inputs as well).

The above discussion was for Boolean matrices with $A_{ij} \sim \{\pm 1\}$.
The natural generalization expresses polynomials in the basis of orthogonal polynomials for the entries $A_{ij}$
(e.g. the Hermite polynomials when the $A_{ij} \sim \calN(0,1/n)$ \cite[Section 3.2]{montanari2022equivalence}). 

Our results show that for the first-order algorithms we consider,
only the multilinear part of the basis matters
(i.e. the orthogonal polynomials which are degree 0 or 1 in each variable):
up to negligible error, we can approximate $A_{ij}^2 \approx \frac 1n$
and $A_{ij}^k \approx 0$ for $k \geq 3$. We use the monomial basis\footnote{
The monomial ``basis'' is a misnomer in the cases when $A_{ij}$ satisfies a polynomial identity
such as $A_{ij}^2 = \frac{1}{n}$.
In these cases, representation as a sum of diagrams is not unique. Our expressions should be interpreted as giving explicit sums of diagrams.} to represent higher-degree
polynomials instead of higher-degree orthogonal polynomials in order to simplify the presentation (except for the degree-2 orthogonal polynomial $A_{ij}^2 - \frac{1}{n}$ which expresses some error terms).

\section{Diagram Analysis of \texorpdfstring{$O(1)$}{O(1)} Iterations}
\label{sec:combinatorial}

In this section we develop tools for rigorously analyzing diagrams of constant size,
corresponding to first-order algorithms with constantly
many iterations. These proofs make formal the intuitive ideas developed in \cref{sec:diagram-calculus}.
Longer proofs in this section are delayed to \cref{app:non-asymptotic-analysis} for readability.

\begin{itemize}
    \item In \cref{sec:eqinf}, we give a rigorous definition of the asymptotic equality $\eqinf$.
    \item In \cref{sec:classification}, we prove the classification of
    the asymptotic behavior of the constant-size diagrams.
    \item In \cref{sec:gfom}, we prove the tree approximation
    for the class of GFOM algorithms.
    \item In \cref{sec:asymptotic-space}, we prove a general state evolution formula for GFOM algorithms.
\end{itemize}

\subsection{Equality up to combinatorially negligible diagrams}
\label{sec:eqinf}

The idea behind $\eqinf$ is to make a purely combinatorial definition
so that we can use combinatorial arguments on the diagrams.
First, we have the following key moment bound which estimates the magnitude in $n$
of a diagram $Z_\al$ based on combinatorial properties of $\al$.

\begin{definition}[$I(\al)$]\label{def:isolated}
     For a diagram $\al \in \calA$, let $I(\al)$ be the subset
     of non-root vertices such that every edge incident to that vertex has multiplicity $\geq 2$ or is a self-loop.
\end{definition}

\begin{restatable}{lemma}{diagramMagnitude}
\label{lem:improper-magnitude-simplified}
Let $q\in\N$ be a constant independent of $n$ and $\al \in \calA$ be a constant-size diagram. Then for $i\in [n]$,
        \[
        \abs{\E\left[Z_{\al,i}^q\right]} \leq O\left( n^{\frac q 2\left(|V(\al)| - 1 - |E(\al)|+|I(\al)|\right)}\right)\,.
    \]
\end{restatable}

A similar norm bound for matrices is a crucial ingredient in Fourier analysis of matrix-valued functions \cite{AMP20}.
The proof of \cref{lem:improper-magnitude-simplified} is in \cref{sec:omitted-combneg}.

Based on this computation, we define a \emph{combinatorially negligible diagram} to be one whose moments decay with $n$. Since we will be working with diagram expressions that are linear combinations of different diagrams, the following definition also handles diagrams rescaled by some coefficient depending on $n$.

\begin{definition}[Combinatorially negligible and order 1]
\label{def:combinatorially-negligible}
    Let $(a_n)_{n \in \N}$ be a sequence of real-valued coefficients such that $a_n=\Theta(n^{-k})$ for some $k\ge 0$ with $2k\in\Z$. Let $\alpha\in\calA$ be a constant-size diagram.

    \begin{enumerate}
        \item We say that $a_n Z_\al$ is \emph{combinatorially negligible}
        if
    \[|V(\al)| - 1 - |E(\al)| + |I(\al)| \le 2k-1\,.\]
        For $a_n=0$, we also say that $a_n Z_\al$ is combinatorially negligible.
        \item We say that $a_n Z_\al$ has \emph{combinatorial order 1} if
    \[|V(\al)| - 1 - |E(\al)| + |I(\al)| = 2k\,.\]
    \end{enumerate}
\end{definition}

We will only consider settings where the coefficients are small enough so that all diagram expressions have combinatorial order at most 1 (that is, negligible or order 1).

\begin{definition}[$\eqinf$]\label{def:asymptotic-equality}
    We say that $x\eqinf y$ if there exists real coefficients $(c_\al)_{\al\in \calA}$ depending on $n$ and supported on diagrams of constant size such that
    \[
        x-y = \sum_{\al\in\calA} c_\al Z_\al\,,
    \]
    where $c_\al Z_\al$ is combinatorially negligible for all $\al\in\calA$.
\end{definition}

Later, we will prove results of the form $x \eqinf \widehat x$ where $x$
is the state of an algorithm and $\widehat x$ is some asymptotic approximation of $x$.
In order to interpret these results, we note that $\eqinf$ implies very strong forms of convergence of the error to $0$. The proof of the following lemma can be found in \cref{sec:omitted-combneg}.

\begin{restatable}{lemma}{eqinfAlmostSure}
\label{lem:eqinf-almostsure}
    Suppose that $A=A(n)$ is a sequence of random matrices satisfying \cref{assump:A-entries}.
    If $x$ and $y$ are diagram expressions such that $x\eqinf y$, then $\norm{x-y}_\infty = \widetilde{O}(n^{-1/2})$ with high probability.
\end{restatable}

Next, we prove a very important property of $\eqinf$.
The combinatorially negligible diagrams remain combinatorially negligible
after applying additional algorithmic operations.

\begin{restatable}{lemma}{eqinfProperties}
\label{lem:eqinf-properties}
    If $x,y$ are diagram expressions with $x \eqinf y$, then
    \[
        Ax\eqinf Ay\,.
    \]
    Moreover, if $x_1, \dots, x_t,y_1, \dots, y_t$ are diagram expressions with $x_i \eqinf y_i$ for all $i\in [t]$, then
    \[
        f(x_1, \dots, x_t) \eqinf f(y_1, \dots, y_t)\,,
    \]
    for any polynomial function $f:\R^t\to \R$ applied componentwise.
\end{restatable}

The proof of \cref{lem:eqinf-properties} is in \cref{sec:omitted-combneg}.
The intuitive view of this lemma is that a diagram with a cycle still has the cycle after the algorithmic operations and thus remains negligible.
The proof in \cref{sec:omitted-combneg} is a syntactic version.

We can also show combinatorially that the error of removing a hanging double edge from any diagram is negligible. The proof proceeds by extending the definition of diagrams to allow new types of residual edges that are only used in the analysis (see \cref{sec:edge-labels}).

\begin{restatable}{lemma}{removeDoubleEdge}
\label{lem:remove-double-edge}
    Let $a_n Z_\al$ be a term of combinatorial order at most $1$ such that $\al$ has a hanging double edge. Let $\al_0$ be $\al$ with the hanging double edge and hanging vertex removed. Then
    \[
        a_n Z_\al \eqinf a_n Z_{\al_0}\,.
    \]
\end{restatable}

\subsection{Classification of constant-size diagrams}
\label{sec:classification}

We classify the asymptotic limits of constant-size diagrams and prove that all of their constant-order joint moments
are within $O(n^{-1/2})$ of the asymptotic limit.
In addition to the vector Fourier diagrams from \cref{def:diagram}, we will classify \emph{scalar Fourier diagrams}, which are simply unlabeled undirected multigraphs (the only difference with vector diagrams being that they do not have a root). The notation for scalar diagrams is analogous.

\begin{definition}[Scalar Fourier diagrams]\label{def:scalar-diagrams}
    Let $\calA_{\scalar}$ be the set of all unlabeled undirected multigraphs with no isolated vertices.
    Let $\calT_{\scalar}$ be the set of non-empty unlabeled trees.

Given a scalar Fourier diagram $\alpha\in \calA_{\scalar}$, we define $Z_\alpha\in\R$ by

\[
    Z_\alpha = \sum_{\substack{\varphi:V(\alpha)\to [n]\\\varphi\textnormal{ injective}}} \prod_{\{u,v\}\in E(\alpha)} A_{\varphi(u) \varphi(v)}.
\]

We allow the empty scalar Fourier diagram which represents the constant 1.
\end{definition}

\begin{definition}[$\calF_\scalar$ and $\calF$]
    Let $\calF_{\scalar}$ be the set of unlabeled forests with no isolated vertices.
    Let $\calF$ be the set of unlabeled forests such that one vertex is the special root vertex~$\rootpic$. No vertices may be isolated except for the root.
\end{definition}

The scalar diagrams are not normalized
``correctly'' by default.
$Z_\rho$ for $\rho \in \calF_
{\scalar}$ has order $n^{c/2}$ where $c$ is the number
of connected components in $\rho$.
The proper normalization divides by $n^{c/2}$ to put all the diagrams on the same scale.
The notion of $\eqinf$ and combinatorial negligibility also extend in a natural
way to scalar diagrams. See \cref{sec:scalar-diagrams} for these definitions.

We classify the diagrams in $\calA$ and $\calA_\scalar$.
First, the next lemma identifies which of the diagrams are
non-negligible.
This lemma is for \emph{connected} vector diagrams; scalar diagrams
and disconnected
vector diagrams have a similar characterization
in \cref{lem:scalar-nonnegligible}.

\begin{lemma}
\label{lem:connected-nonnegligible}
Let $\al \in \calA$ be a connected Fourier diagram. Then $Z_\al$ is either combinatorially negligible or combinatorially order 1. Moreover, it is combinatorially order 1 if and only if the following four conditions hold simultaneously:
\begin{enumerate}[(i)]
    \item Every multiedge has multiplicity 1 or 2.
    \item There are no cycles.
    \item The subgraph of multiplicity 1 edges is connected and contains the root if it is nonempty (i.e. the multiplicity 2 edges consist of hanging trees).
    \item There are no self-loops or 2-labeled edges (\cref{sec:edge-labels}).
\end{enumerate}
\end{lemma}
\begin{proof}
    By assumption, every vertex is connected to the root. With the exception of the root,
    we can assign injectively one edge to every vertex in $V\setminus I(\al)$
    and two edges to every vertex in $I(\al)$
    as follows.
    Run a breadth-first search from the root and assign to each vertex the multiedge that was used to discover it.
    This encoding argument implies 
    \[
        (|V(\al)|-|I(\al)|-1) + 2|I(\al)| \le |E(\al)|\,.
    \]
    Hence $Z_\alpha$ is combinatorially negligible or combinatorially order 1, and it is combinatorially order 1 if and only if
    this inequality is an equality.
    This holds if and only if
    there are no cycles, multiplicity $\sgt 2$ edges, self-loops, or 2-labeled edges
    in $\al$,
    and the edges incident to $V(\al)\setminus I(\al)$ in the direction of the root
    all have multiplicity 1.
\end{proof}

As a result, the non-negligible connected diagrams in $\calA$ are asymptotically equal to trees in $\calT$ after using
\cref{lem:remove-double-edge} to remove the hanging
double edges (disconnected diagrams $\al \in \calA$ and scalar diagrams $\al \in \calA_\scalar$ are likewise asymptotically equal to a forest in $\calF$ or $\calF_\scalar$).

The next \cref{thm:classification} completes the classification by showing that the non-negligible diagrams in $\calT, \calF,$ and $\calF_\scalar$
are asymptotically Gaussians and Hermite polynomials.
The proof is in \cref{sec:diagram-classification-proof}.
Also see \cref{all-moments} for a version of the theorem in terms of moments.

\begin{theorem}[Classification]\label{thm:classification}
    Suppose that $A=A(n)$ is a sequence of random matrices satisfying \cref{assump:A-entries}.
    
    The non-negligible scalar diagrams can be classified as follows:
    \begin{itemize}
        \item If $\tau\in\calT_{\scalar}$, then $n^{-\frac 1 2} Z_\tau \overset{d}{\longrightarrow} \calN(0, \abs{\Aut(\tau)})$.
        \item If $\rho \in \calF_{\scalar}$ has $c$ connected components, then 
        \[n^{-\frac c 2} Z_\rho \eqinf \prod_{\tau \in \calT_{\scalar}} h_{d_\tau}(n^{-\frac 1 2}Z_{\tau}; \abs{\Aut(\tau)})\,,\]
        where $d_\tau$ is the number of copies of $\tau$ in $\rho$.
    \end{itemize}
    The non-negligible vector diagrams can be classified as follows:
    \begin{itemize}
        \item If $\sig \in \calS$ and $i\in [n]$, then $Z_{\sig, i} \overset{d}{\longrightarrow} \calN(0, \abs{\Aut(\sig)})$.
        \item If $\tau \in \calT$, then $Z_{\tau} \eqinf \prod_{\sig \in \calS} h_{d_\sig}(Z_{\sig}; \abs{\Aut(\sig)})$
        where $d_\sig$ is the number of isomorphic copies of $\sigma$ starting from the root of $\tau$,
        and the Hermite polynomial is applied componentwise.
        \item If $\al\in \calF$ has $c$ 
        \textit{floating components}
        (connected components which are not the component of the root), letting $\al_{\smallrootpic}$ be the component of the root (a vector diagram) and $\al_{\text{float}}$ be the floating part (a scalar diagram), then $n^{-\frac c 2}Z_\al \eqinf n^{-\frac c 2}Z_{\al_{\float}} Z_{\al_{\smallrootpic}}$.
    \end{itemize}
    Moreover, the random variables
    \[
        \left\{Z_{\sig, i} : \sig \in \calS, i \in [n]\right\} \cup \left\{n^{-\frac 1 2} Z_{\tau} : \tau \in \calT_{\scalar}\right\}
    \]
    are asymptotically independent (\cref{def:asymptotic-indpt}).
\end{theorem}

Finally, we formalize what we mean by \textit{asymptotic independence} of vectors whose dimension can grow with $n$.
\begin{definition}[Asymptotic independence]
\label{def:asymptotic-indpt}
    A family of real-valued random variables $(X_{n,i})_{n\in\N,i\in \calI_n}$ is \textit{asymptotically independent} if:
    \[ \forall q \in \N.\; \exists \eps=\eps(q) \underset{n\to\infty}{\longrightarrow} 0. \; \forall k  \in \N^{\calI_n} : \sum_{i \in \calI_n} k_i = q.\;
    \abs{\E\left[ \prod_{i \in \calI_n} X_{n,i}^{k_i}\right] - \prod_{i \in \calI_n}\E\left[X_{n,i}^{k_i}\right]} \le \eps(q)\,.\]
    Note that $\calI_n$ may be infinite.
\end{definition}

\subsection{Tree approximation of GFOMs}
\label{sec:gfom}

The class of general first-order methods is defined as follows.
\begin{definition}[General first-order method]
\label{def:gfom}
The input is a matrix $A \in \R^{n \times n}$.
The state of the algorithm at time $t$
is a vector $x_t \in \R^n$. Initially, $x_0=\vec{1}$.
At each time $t$, we can execute one of the following two operations:
\begin{enumerate}
    \item Multiply by $A$ ($x_{t+1} = Ax_t$).
    \item Apply coordinatewise a polynomial\footnote{Restriction to polynomial functions is a technical assumption which is not present in the full definition.} function independent of $n$, $f_t : \R^{t+1} \to \R$ to $(x_t, x_{t-1}, \dots, x_0)$ (for all $i \in [n]$, $x_{t+1,i}=f_t(x_{t,i},\ldots,x_{0,i})$).
\end{enumerate}
\end{definition}

Inductively following the rules given explicitly in \cref{sec:derivation-asymptotic},
we may represent the algorithmic state $x_t$ of a GFOM in the diagram
basis.
Define the \emph{tree approximation} $\widehat{x}_t$ to be the analogous diagram expression obtained by performing the algorithmic operations on only the tree diagrams, removing hanging double edges and removing the cyclic diagrams that appear (see \cref{def:asymptotic-state-explicit} for the formal definition).

\begin{restatable}[Tree approximation of GFOMs]{theorem}{stateEvolution}
\label{thm:state-evolution}
    Let $t\ge 0$ be a constant independent of $n$ and $A=A(n)$ be a sequence of random matrices satisfying \cref{assump:A-entries}. 
    Let $x_t\in \R^n$ be the state of a GFOM and let $\widehat{x}_t$ be its tree approximation.
    Then $x_t \eqinf \widehat{x}_t$. In particular,
    \begin{align}
         \label{eq:concl-tree-1} \qquad \qquad \norm{x_t - \widehat{x}_t}_\infty = \widetilde{O}(n^{-1/2})\text{ with high probability}\,.
    \end{align}
\end{restatable}

\begin{proof}
    We can prove $x_t \eqinf \widehat{x}_t$ inductively.
    By \cref{lem:eqinf-properties}, each of the combinatorially negligible diagrams
    in $x_t$ remains combinatorially negligible at time $t+1$.
    Meanwhile, the combinatorially non-negligible tree diagrams in $\widehat{x}_t$ get updated to $\widehat{x}_{t+1}$.
    The error bound \cref{eq:concl-tree-1} follows from \cref{lem:eqinf-almostsure}.
\end{proof}

\begin{remark}\label{rem:lower-order}
    The leading order guarantee of \cref{thm:state-evolution} 
    is best possible in general
    (up to logarithmic factors). Similar but more complicated equations can be given
    for the lower-order error terms in \cref{eq:concl-tree-1}. For example, since the other connected diagrams with $E$ edges and $V$ vertices
    have magnitude $n^{(V - 1 - E)/2}$, the first lower-order term of order $n^{-1/2}$ consists
    of connected diagrams with exactly one cycle.
    The GFOM operations on this set of diagrams describe how the error evolves at this order.
\end{remark}

\begin{remark}
    One technical caveat of our analysis is that many nonlinearities used in applications are not polynomial functions (e.g. \textnormal{ReLU}, $\tanh$). We note that existing polynomial approximation arguments in the literature (see for example \cite{montanari2022equivalence, ivkov2023semidefinite}) should apply here to prove that the tree approximation holds for GFOMs with Lipschitz denoisers $f_t$ up to arbitrarily small $\frac 1 {\sqrt n}\|\cdot\|_2$ error. This is however strictly weaker than the guarantees of \cref{thm:state-evolution}.
\end{remark}

\subsection{General state evolution}
\label{sec:asymptotic-space}

From the ideas established so far, we directly deduce \emph{state evolution} for GFOM algorithms, capturing that
the coordinates of $x_t$ are asymptotically independent trajectories
of an explicit random variable $X_t$.
Recall the definition of the asymptotic state $X_t$ from \cref{def:asymptotic-state}.

To state the theorem, asymptotic independence is extended from \cref{def:asymptotic-indpt} to $\R^{t+1}$-valued random variables
in the natural way. 
\begin{definition}[Asymptotic independence of trajectories]
    A family of random variables $(X_{n,i})_{n\in\N,i\in \calI_n}$ taking values in $\R^{t+1}$ is \textit{asymptotically independent} if:
    \begin{align*} 
    &\forall q \in \N.\; \exists \eps=\eps(q) \underset{n\to\infty}{\longrightarrow} 0. \; \forall k  \in \N^{\calI_n\times [t+1]} : \sum_{\substack{i \in \calI_n\\ j \in [t+1]}} k_{ij} = q.\;\\
    &\qquad\abs{\E\left[ \prod_{\substack{i \in \calI_n, j \in [t+1]}} X_{n,i,j}^{k_{ij}}\right] - \prod_{i \in \calI_n}\E\left[\prod_{j \in [t+1]} X_{n,i,j}^{k_{ij}}\right]} \le \eps(q)\,.
    \end{align*}
\end{definition}

\begin{theorem}[General state evolution]
    \label{thm:general-state-evolution}
    Let $t$ be a constant and $A = A(n)$ be a sequence of random matrices satisfying \cref{assump:A-entries}. Let $x_t\in \R^n$ be the state of a GFOM and let $X_t$ be the asymptotic state of $x_t$. Then: \footnote{It is natural to wonder whether (ii) can be derived as a consequence of (i) in \cref{thm:general-state-evolution}. The answer is a resounding no. Even given good control over the distribution of individual coordinates $x_{t,i}$ it is crucial to ensure that errors do not correlate adversarially when summed. This is precisely the kind of heuristic assumption made when using the cavity method.}
    \begin{enumerate}[(i)]
        \item For each $i \in [n]$, $(x_{0,i}, \dots, x_{t,i}) \tod (X_0, \dots, X_t)$.
        Furthermore, the coordinates' trajectories $\{(x_{0,i},\dots, x_{t,i}) : i \in [n]\}$ are asymptotically independent.
        \item $\frac{1}{n} \sum_{i = 1}^n x_{t,i} \eqinf \E[X_t]$ and therefore,
    \[
        \frac{1}{n} \sum_{i = 1}^n x_{t,i} = \E[X_t] + \widetilde{O}(n^{-\frac 1 2})\text{ with high probability}\,.
    \]
        \item $X_t$ satisfies the explicit recurrence defined at the end of \cref{sec:calculus-rules}.
    \end{enumerate}
\end{theorem}

\begin{proof}
    For (i), by \cref{lem:method-of-moments}, it suffices to check that
    all of the constant-order joint moments of $x_{t,i}$ converge to the joint moments of $X_t$. This follows
    from convergence of the moments of every diagram $Z_\al$ to those of $Z_\al^\infty$ in the diagram classification \cref{thm:classification}.

    Part (ii) will be proven in 
    \cref{sec:onsager-correction} as the following lemma.
    \begin{restatable}{lemma}{empiricalExpectation}\label{lem:scalar-empirical-expectation}
    Let $x$ be a vector diagram expression with asymptotic state $X \in \Om$.
    Then as scalar diagrams, $\frac{1}{n}\sum_{i = 1}^n x_i \eqinf \E\left[X\right].$
    \end{restatable}
    
    For (iii), the tree approximation $x_t = \widehat{x}_t$ holds by \cref{thm:state-evolution}.
    The asymptotic state $X_t$ corresponding to $\widehat{x}_t$ then satisfies the explicit recursion on trees presented
    in \cref{sec:calculus-rules}. 
\end{proof}

We conclude this section by working out a few lemmas which help compute asymptotic
states.
We will use them in \cref{sec:amp} to compute
the state evolution of approximate
message passing.

The set of asymptotic states $\Om$ has an inner product coming from the expectation over the Gaussians $(Z_\sig^\infty)_{\sig \in \calS}$.
Since these random variables are independent Gaussians, the multivariate Hermite polynomials $(Z^\infty_\tau)_{\tau\in\calT}$ form an orthogonal basis of $\Om$ with respect to this inner product. Recall the $+$ and $-$ operators
from \cref{def:plus-minus-operators}.

\begin{fact}\label{lem:bijection}
    $+$ and $-$ are bijections between $\calT$ and $\calS$ which are inverses of each other and
    preserve $\abs{\Aut(\tau)}$.
\end{fact}
A key observation is that $X^+$ is always a centered Gaussian random variable for any $X \in \Om$,
since every resulting tree is in $\calS$.

\begin{fact}\label{cor:pm-cancel}
    For all $X \in \Om$, $(X^+)^- = X$ and $(X^-)^+$ is the orthogonal projection
    of $X$ to the subspace spanned by $\calS$.
\end{fact}

We deduce that $+$ and $-$ are adjoint operators on $\Om$:

\begin{lemma}\label{lem:adjoint}
    For all $X, Y \in \Om$,
    $\E\left[X^+Y\right] = \E\left[XY^-\right]$.
\end{lemma}
\begin{proof}
    Since $(Z^\infty_\tau)_{\tau \in \calT}$ is a
    basis of the vector space $\Om$, it suffices to check this
    for each pair of basis elements $\tau, \rho \in \calT$.
    By orthogonality,
    $\E\left[Z^\infty_{\tau^+}Z^\infty_\rho\right]$ is nonzero if and only if $\tau^+ = \rho$ and in this case it takes value $\abs{\Aut(\tau^+)}$.
    By \cref{lem:bijection}, this occurs if and only if $\rho\in\calS$ and $\tau = \rho^-$. Moreover, in this case the value is also $\abs{\Aut(\tau^+)} = \abs{\Aut(\tau)}$, as needed.
\end{proof}

\begin{lemma}\label{cor:adjoint}
    For all $X, Y \in \Om$,
    $\E\left[XY\right] = \E\left[X^+Y^+\right]$ and $\E\left[(X^-)^2\right] \leq \E\left[X^2\right]$.
\end{lemma}
\begin{proof}
    For the first statement, apply \cref{lem:adjoint} on $X$ and $Y^+$, then use \cref{cor:pm-cancel}.
    For the second statement, apply \cref{lem:adjoint} on $X^-$ and $X$ to get $\E\left[(X^-)^+ X\right] = \E\left[(X^-)^2\right]$. Since $(X^-)^+$ projects away some terms from $X$ by \cref{cor:pm-cancel},
    the left-hand side is upper bounded by $\E\left[X^2\right]$.
\end{proof}

\section{Belief Propagation, AMP, and the Cavity Method}
\label{sec:cavity-method}

With the tree approximation in hand, we describe how to use it
to implement cavity method reasoning about nonlinear iterative algorithms.
\begin{itemize}
    \item In \cref{sec:intro-cavity}, we give background on the cavity method and how it can be used
    to predict the asymptotic trajectory of
    message-passing algorithms.
    \item In \cref{sec:bp-amp}, we prove the asymptotic equivalence of BP and AMP on mean-field models (\cref{thm:bp-amp}).
    The proof precisely follows the
    structure described earlier: we 
    reproduce a folklore physics argument in
    \cref{sec:bp-amp-heuristic} and make it directly
    rigorous in \cref{sec:bp-amp-formal}. 
    \item In \cref{sec:cavity-formal}, we use the same
    technology to prove a fundamental 
    assumption of the
    cavity method: the asymptotic independence of messages
    incoming at a vertex (\cref{thm:cavity-formal}).
    
    \item In \cref{sec:amp}, we give a new proof of the state
    evolution formula for BP/AMP (\cref{thm:amp-state}). 
    
    \item In \cref{sec:montanari}, we reinterpret Montanari's
    algorithm for optimizing spin glass Hamiltonians
    through the lens of the asymptotic tree space.
\end{itemize}

\subsection{Background on the cavity method}
\label{sec:intro-cavity}

Belief Propagation (BP) and Approximate Message Passing (AMP)
are the main class of nonlinear iterative algorithms that
are studied using physical techniques.
BP is a general tool for statistical inference on graphical models which performs exact inference when the underlying graph is a tree.
The behavior of ``loopy BP''
on interaction graphs with cycles is more subtle; the \emph{cavity method} can be used to predict the asymptotic dynamics of loopy BP on mean-field models (i.e. when the underlying graphical model is dense and random).

We first explain the idea behind the cavity method on the example
of the replica-symmetric belief propagation iteration
for the Sherrington--Kirkpatrick (SK) model, which is
the original setting in which the method was conceived
by Mézard, Parisi, and Virasoro~\cite[Chapter V]{mezard1987spinglasstheoryandbeyond}. The goal here is
to estimate the marginals of the following Gibbs
distribution on $x\in \{-1,1\}^n$:
\[
    p(x) \propto \exp\left(\beta x^\top A x + h\sum_{i=1}^n x_i\right)\,,
\]
where $A$ is a random symmetric matrix with i.i.d.
$\calN(0,1/n)$ entries and $\beta,h>0$ are fixed parameters.
We will focus on a particular regime of $(\beta,h)$ known as the replica-symmetric or high temperature region of the SK model.

Let $m_i = \E_{x \sim p}[x_i]$. By isolating a single coordinate $i\in [n]$ and looking at the influence of other coordinates on it,  Mézard, Parisi, and Virasoro derive 
the {\em cavity equations}, which are fixed-point equations approximately satisfied by $m_i$,
\begin{align}
    m_{i\to j} = f\left(\sum_{\substack{k=1\\k\neq j}}^n A_{ik} m_{k\to i}\right)\,,\quad m_i \approx f\left(\sum_{k=1}^n A_{ik} m_{k\to i}\right)\,,\label{eq:cavity-equations}
\end{align}
where $f(x)=\tanh(\beta x+h)$ and $m_{i \to j}$ are new variables. Algorithmically,  we can think of an iterative \emph{belief propagation} algorithm that tries to compute a solution to these equations,
\begin{align}
    m_{i\to j}^{t+1} = f\left(\sum_{\substack{k=1\\k\neq j}}^n A_{ik} m_{k\to i}^t\right)\,,\quad m_i^{t+1} = f\left(\sum_{k=1}^n A_{ik} m_{k\to i}^t\right)\,,\label{eq:bp-sk}
\end{align}
initialized at say $m_{i \to j}^0 = 1$.
This iteration occurs on a set of {\em cavity messages} $m_{i \to j}$ for $i,j \in [n]$
which conceptually
are ``the belief of vertex $i$ about its own value, disregarding $j$''.

The physical techniques predict the asymptotic trajectory of the messages
$m^t_{i \to j}$ and the outputs $m_i^t$ in \cref{eq:bp-sk} with respect to the randomness of the matrix $A$.
They say that $m^t$ will have approximately 
independent and identically distributed entries,
\begin{align}
    m^t_i \sim {f}(Z_t)\,, \qquad \text{where } & Z_t \sim \calN(0, \sig_t^2)\,,\nonumber\\
    &\sigma_1^2 = 1\,,\quad  \sigma_{t+1}^2 = \E f(Z_t)^2\,.\label{eq:state-evolution}
\end{align}
A heuristic replica symmetric cavity approach for proving \cref{eq:state-evolution} would go as follows.
We make an \textbf{independence assumption} that the incoming terms $m_{k \to i}^t$ in the non-backtracking summation $\sum_{k = 1, k \neq j}^n A_{ik}m_{k \to i}^t$
of \cref{eq:bp-sk} are independent, as if the
messages were coming up from disjoint branches of a tree.
By symmetry, the messages are identically distributed.
Then, we appeal to the central limit theorem to deduce
\[ 
    \sum_{\substack{k = 1\\ k \neq j}}^n A_{ik}m_{k \to i}^t \sim \calN\left(0, \E\left[(m_{k \to i}^t)^2\right]\right)\,
.\]
From here, we get that the outgoing message satisfies $m_{i \to j}^t \sim f(Z_t)$ for $Z_t \sim \calN(0, \sigma_t^2)$
with $\sigma_t^2$ defined by the recurrence in \cref{eq:state-evolution}.
Using a similar argument, we get $m_i^t \sim {f}(Z_t)$. 

\cite{mezard1987spinglasstheoryandbeyond} also derived from \cref{eq:cavity-equations} a simpler form of self-consistent equations involving only the marginals themselves, known as the Thouless--Anderson--Palmer equations~\cite{thouless1977solution},
\begin{align}
    m_i \approx f\left(\sum_{\substack{k=1\\k\neq i}}^n A_{ik} m_k - \beta\left(1-\frac 1 n \sum_{k=1}^n m_k^2\right) m_i\right)\,.
    \label{eq:tap}
\end{align}
The subtracted term on the right-hand side in which $m_i$ re-occurs is the \emph{Onsager reaction term}.
In the same way that belief propagation \cref{eq:bp-sk} tries to compute solutions to the cavity equations \cref{eq:cavity-equations}, an \emph{approximate message passing} algorithm can be iterated to compute approximate solutions to \cref{eq:tap},
\begin{align}
    m^{t+1}_i = f\left(\sum_{\substack{k=1\\k \neq i}}^n A_{ik} m^t_k - \beta \left(1-\frac 1 n \sum_{k=1}^n (m_k^t)^2\right) m_i^{t-1}\right)\,.
    \label{eq:amp-sk}
\end{align}
The approximate equivalence between the BP iteration \cref{eq:bp-sk} and the AMP iteration \cref{eq:amp-sk} is a folklore cavity method argument which we elaborate
next.

\subsection{Equivalence between message-passing iterations}
\label{sec:bp-amp}

\paragraph{Belief propagation.} We consider BP iterations
on $A$ of the form
\begin{align}
    \label{eq:bp-memory}
    m^0_{i \to j} = 1\,, \qquad m^{t}_{i \to j} = f_{t}\left(\sum_{\substack{k = 1\\k \neq j}}^n A_{ik}m_{k \to i}^{t-1}, \dots, \sum_{\substack{k = 1\\k \neq j}}^nA_{ik}m_{k \to i}^0,\; m^0_{i \to j}\right)\,,\\
    \nonumber
    m^{t}_i = \widetilde{f}_{t}\left(\sum_{\substack{k = 1}}^nA_{ik}m_{k \to i}^{t-1}, \dots, \sum_{\substack{k = 1}}^nA_{ik}m_{k \to i}^0, \;m^0_{i \to j}\right)\,,
\end{align}
for a sequence of functions $f_t, \widetilde{f}_t: \R^{t+1} \to \R$. \cref{eq:bp-memory} is a generalization of \cref{eq:bp-sk} to iterations ``with memory'' i.e. that can use all the previous messages. At any timestep $t$, the $(m_{i\to j}^t)_{1\le i,j\le n}$ are \textit{cavity messages} that try to compute some information about the $i$-th variable by ignoring the edge between $i$ and $j$, while the $(m_i^t)_{1\le i\le n}$ are the output of the algorithm.

\paragraph{Approximate message passing.} On the other side, we have an \emph{approximate message passing} (AMP) algorithm of the form
\begin{align}
    w^0 = \vec{1}\,, \qquad w^{t+1} &= A f_t(w^t, \dots, w^0) - \sum_{s = 1}^t b_{s,t} f_{s-1}(w^{s-1}, \dots, w^0)\,,\label{eq:amp}\\
    m^{t} &= \widetilde{f}_t(w^t, \dots, w^0)\,,\label{eq:amp-output}
\end{align}
where $b_{s,t}$ is defined to be the scalar quantity
\[b_{s,t}=\frac 1 n \sum_{i=1}^n \frac {\partial f_t}{\partial w^s}(w^{t}_i,\ldots,w^0_{i})\,.\]
One practical advantage 
of AMP compared to BP
is that is has a smaller number of messages to track,
$O(n)$ vs $O(n^2)$.

\begin{theorem}[Equivalence of BP and AMP]\label{thm:bp-amp}
    Let $T\ge 1$ be a constant independent of $n$, $f_t, \widetilde{f}_t:\R^{t+1}\to \R$ for $t\le T$ be a sequence of polynomials independent of $n$, and $A=A(n)$ be a sequence of random matrices satisfying \cref{assump:A-entries}. Generate $m^{t, \textnormal{BP}}$ according to \cref{eq:bp-memory} and $m^{t, \textnormal{AMP}}$
    according to \cref{eq:amp-output}.
    Then 
    \[
        m^{t,\textnormal{AMP}} \eqinf m^{\textnormal{BP}}\,,
    \]
    so in particular, with high probability,
    \begin{align*}
        \norm{m^{t,\textnormal{AMP}} - m^{t, \textnormal{BP}}}_\infty = \widetilde{O}(n^{-1/2})\,.
    \end{align*}
\end{theorem}

\subsubsection{Heuristic derivation of \cref{thm:bp-amp}}
\label{sec:bp-amp-heuristic}

The equivalence between BP and AMP is folklore
in the statistical physics community, thanks to
the following simple cavity-based reasoning. It 
can be found for example
in the seminal paper \cite[Appendix A]{donoho2009message}
or the survey \cite[Section IV.E]{zdeborova2016statistical}.

We start by rewriting the BP iteration, letting $w^0=\vec{1}$ and $w_i^{t+1}=\sum_{k=1}^n A_{ik} m^t_{k\to i}$. The output of BP is computed as
\begin{align*}
    m^{t+1}_i = \widetilde{f}_{t+1}\left(w_i^{t+1}, \dots, w_i^0\right)\,.
\end{align*}
Hence it suffices to show that $w^t$ asymptotically follows the AMP iteration \cref{eq:amp}. First, \cref{eq:bp-memory} can be rewritten
\[
    m^{t+1}_{i\to j} = f_{t+1}\left(w_i^{t+1} - A_{ij} m^t_{j\to i}, \ldots, w^1_{i} - A_{ij} m^0_{j\to i},w^0_i\right)\,.
\]
Given that the entries $A_{ij}$ are on the scale of $1/\sqrt n$, which we expect to be much smaller than the magnitude of the messages, we perform a first-order Taylor approximation (the partial derivatives
are with respect to the coordinates of $f_{t+1}$ and the last coordinate is ignored because $w^0_i$ is constant):
\begin{align*}
    m^{t+1}_{i\to j}&\approx f_{t+1}\left(w^{t+1}_i, \sdots w^1_i, w_i^0\right) - A_{ij}\sum_{s = 1}^{t+1} m_{j \to i}^{s-1} \pd{f_{t+1}}{w^s}\left(w^{t+1}_i, \sdots w^1_i, w_i^0\right)\label{eq:approx1}\tag{$\ast$}\,.
\end{align*}
Plugging this approximation in the definition of $w^{t+1}_i$,
\begingroup
\allowdisplaybreaks
\begin{align*}
    w^{t+1}_i &\approx \sum_{k =1}^n A_{ik} f_t(w_k^t, \dots, w_k^0) - \sum_{k=1}^n A_{ik}^2 \sum_{s=1}^{t}   m^{s-1}_{i\to k} \pd{f_{t}}{w^s}(w_k^t, \dots, w_k^0) &\\
    &\approx \sum_{k=1}^n A_{ik} f_t(w_k^t, \dots, w_k^0) - \sum_{k=1}^n \frac 1 n \sum_{s=1}^{t}   f_{s-1}(w_i^{s-1}, \dots, w_i^0) \pd{f_t}{w^s}(w_k^t, \dots, w_k^0) & \label{eq:approx2}\tag{$\ast\ast$}\\
    &= \sum_{k=1}^n A_{ik} f_t(w_k^t, \dots, w_k^0) -  \sum_{s=1}^{t} b_{s,t}f_{s-1}(w_i^{s-1}, \dots, w_i^0)\,.
\end{align*}
\endgroup
This shows that $w^{t+1}_i$ approximately satisfies the AMP recursion \cref{eq:amp}, as desired. 

The intuition behind \cref{eq:approx2} is that because we are summing over $k$,
we may expand $A_{ik}^2$ and $m_{i \to k}^{s-1}$ on the first order and replace them by averages which do not depend on $k$:
\begin{align*}
    A_{ik}^2 &\approx \E\left[A_{ik}^2\right] = \frac{1}{n}\,,\\
    m_{i\to k}^{s-1}&= f_{s-1}\left(w_i^{s-1} - A_{ik} m^t_{k\to i}, \sdots w^1_{i} - A_{ik} m^0_{k\to i},w^0_i\right)\\
    &\approx f_{s-1}\left(w_i^{s-1}, \sdots w^0_i\right)\,.
\end{align*}

\subsubsection{Diagram proof of \cref{thm:bp-amp}}
\label{sec:bp-amp-formal}

In fact, the previous heuristic argument can be
made directly rigorous by working with the tree
approximation. It suffices to justify
\cref{eq:approx1} and \cref{eq:approx2} 
in order to prove \cref{thm:bp-amp}.

The BP iteration takes place on $m^t \in \R^{n^2}$ instead of $\R^n$
which is not captured by our previous definitions.
Most of the work below is setting up definitions to fit this iteration into our framework.
We define diagrams
for vectors $x \in \R^{n(n-1)}$ whose $(i,j)$ entry is written $x_{i \to j}$
(for simplicity, we assume $A_{ii} = 0$ so that
the messages $m_{i \to i}^t$ can be ignored).

\begin{definition}[Cavity diagrams]
    A cavity diagram is an unlabeled undirected graph $\alpha=(V(\alpha),E(\alpha))$ with two distinct, ordered root vertices $\ijrootpic$. No vertices may be isolated except for the roots.

    For any cavity diagram $\alpha$, we define $Z_\alpha\in \R^{n(n-1)}$ by
    \[
        Z_{\al, i \to j} = \sum_{\substack{\varphi \colon V(\al) \to [n]\\\textnormal{$\varphi$ injective} \\ \ph(\smallijrootpic) = (i,j)}} \prod_{\{u, v\} \in E(\al)} A_{\ph(u),\ph(v)}\,,
    \]
    for any distinct $i,j\in [n]$.
\end{definition}

Below is the representation of the first iterate of \cref{eq:bp-memory} with cavity diagrams. In the pictures, we draw an arrow from the first root to the second root to indicate the order.
If a (multi)edge exists in the graph between the roots, then the arrow
is on the edge; otherwise we use a dashed line to indicate that there is no edge.
\begingroup
\allowdisplaybreaks
\begin{align*}
    m_{i \to j}^0 &= \vcenter{\hbox{\includegraphics[height=3ex]{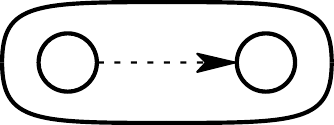}}}\\
    \sum_{\substack{k = 1\\k \neq j}}^n A_{ik} m_{k \to i}^0 &= \vcenter{\hbox{\includegraphics[height=3ex]{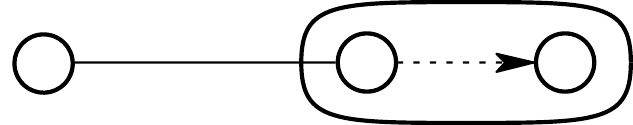}}}\\
    \sum_{k = 1}^n A_{ik}m_{k \to i}^0 &= \vcenter{\hbox{\includegraphics[height=3ex]{images-pdf/ij-fwd.pdf}}} + \vcenter{\hbox{\includegraphics[height=3ex]{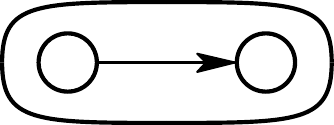}}}
\end{align*}

Multiplying $A_{ik}m_{k \to i}^t$ creates a new edge between
$k$ and $i$ in $m_{k \to i}^t\,$. Summing over $k$ ``unroots'' the first root.
A case
distinction needs to be made in the summation depending on if $k = i$ or $k = j$ or $k \not \in \{i,j\}$.
The case $k = i$ is ignored assuming that $A_{ii} = 0$.
The case $k = j$ yields the ``backward step'' while the remaining case $k \neq j$ is the ``forward step''.

To apply $f_1$, we need to multiply $i \to j$ diagrams componentwise, which is achieved by fixing/merging the roots $i,j$
and summing over the part outside the roots.
For some coefficients $c_0, c_1, c_2, \dots$ we have\footnote{The exact values of the coefficients $c_i$ are not necessary to compute.}
\begin{align*}
    m_{i \to j}^1 = f_1\left(\sum_{\substack{k = 1\\k \neq j}}^n A_{ik} m_{k \to i}^0\right) &= c_0\;\vcenter{\hbox{\includegraphics[height=3ex]{images-pdf/ij-init.pdf}}} + c_1 \vcenter{\hbox{\includegraphics[height=3ex]{images-pdf/ij-fwd.pdf}}} + c_2 \vcenter{\hbox{\includegraphics[height=8.5ex]{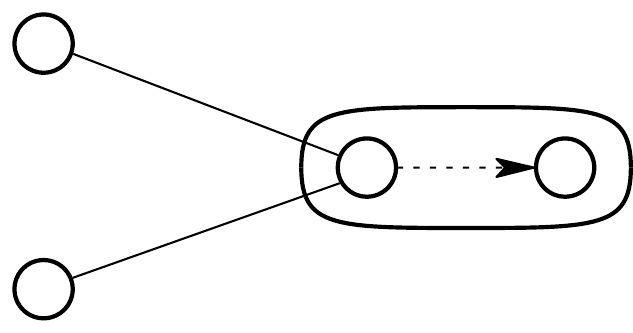}}} + \cdots
\end{align*}

The output $m^{t+1}_i$ uses the non-cavity quantities $\sum_{k=1}^n A_{ik} m_{k\to i}^t$. The cavity diagrams are converted back to the usual diagram basis as follows.

\begin{claim}[Conversion of cavity diagrams]\label{claim:cavity-to-regular}
    For any cavity diagram $\alpha$ and $i\in [n]$,
    \[
        \sum_{j=1}^n A_{ij} Z_{\alpha,j\to i} = Z_{\alpha',i}\,,
    \]
    where $\alpha'$ is the diagram (in the sense of \cref{def:diagram}) obtained from $\alpha$ by adding an edge between the two roots of $\alpha$ and unrooting the first root.
\end{claim}

Since the final output is computed by converting all cavity diagrams back to regular diagrams using the previous claim, the definition of combinatorial negligibility and the $\eqinf$ notation can be extended to cavity diagrams.
We make the following definitions.

\begin{definition}
    A cavity diagram $\alpha$ is combinatorially negligible if the diagram $\alpha'$ obtained in \cref{claim:cavity-to-regular} is combinatorially negligible. We naturally extend the $\eqinf$ notation to cavity diagrams as in \cref{def:asymptotic-equality}.
\end{definition}

\begin{claim}\label{fact:eqinf-cavity}
    Let $x$ and $x'$ be in the span of the cavity diagrams such that $x\eqinf x'$.
    If we let
    \[
        y_{i\to j} = \sum_{\substack{k=1\\k\neq j}}^n A_{ik} x_{k\to i}\,,\quad y'_{i\to j} = \sum_{\substack{k=1\\k\neq j}}^n A_{ik} x'_{k\to i}\,,
    \]
    then $y\eqinf y'$.

    If $x_1, \ldots, x_t, x'_1, \ldots, x'_t$ are in the span of cavity diagrams, $x_i\eqinf x'_i$ for all $i\in [n]$, and $f:\R^t\to \R$ is a polynomial function applied componentwise, then
    \[
        f(x_1, \ldots, x_t)\eqinf f(x'_1, \ldots, x'_t)\,.
    \]
\end{claim}

\cref{fact:eqinf-cavity} follows directly from \cref{lem:eqinf-properties}.

This completes the diagrammatic description of the belief
propagation algorithm.
We are now ready to rigorously justify the approximations made during the heuristic argument.

\begin{lemma}[\cref{eq:approx1}]
    \[
        m^{t}_{i\to j} \eqinf f_{t}\left(w^{t}_i, \ldots, w_i^0\right) - A_{ij}\sum_{s = 1}^{t} m_{j \to i}^{s-1} \pd{f_{t}}{w^s}\left(w^{t}_i, \ldots, w_i^0\right)\,.
    \]
\end{lemma}

\begin{proof}
    Since $f_t$ is a polynomial, it has an exact Taylor expansion. The terms of degree higher than $1$ in the Taylor expansion create at least $2$ edges between the roots $i$ and $j$. All cavity diagrams with $2$ edges between the roots are combinatorially negligible because the unrooting operation of \cref{claim:cavity-to-regular} adds one more edge between $i$ and $j$, and diagrams with multiedges of multiplicity $>2$ are combinatorially negligible (\cref{lem:connected-nonnegligible}).
\end{proof}

\begin{lemma}[\cref{eq:approx2}]
    \[
        \sum_{k=1}^n A_{ik}^2 m_{i\to k}^{s-1} \pd{f_t}{w^s}(w_k^t, \dots, w_k^0) \eqinf \frac 1 n f_{s-1}(w^{s-1}_i,\ldots,w^0_i) \sum_{k=1}^n \pd{f_t}{w^s}(w_k^t, \dots, w_k^0)\,.
    \]
\end{lemma}

\begin{proof}
    First, we argue about the replacement of $m_{i \to k}^{s-1}$. We have
    \[m_{i \to k}^{s-1} = f_{s-1}\left(\sum_{\substack{\el=1\\\el\neq k}}^n A_{i\el} m^{s-2}_{\el\to i}, \sdots \sum_{\substack{\el=1\\\el\neq k}}^n A_{i\el} m^{0}_{\el\to i},\; m_{i \to k}^0\right)\,.\]
    The difference between this and $f_{s-1}(w_i^{s-1}, \dots, w_i^0)$
    are the backtracking terms $A_{ik}m_{k \to i}^r$. All terms in the entire Taylor expansion of the polynomial on the right-hand side around $w_i^{s-1},\ldots,w_i^{0}$ will introduce at least one additional factor of $A_{ik}$, which combines with the $A_{ik}^2$ present in the summation over $k$ to become a negligible multiplicity $> 2$ edge (\cref{lem:connected-nonnegligible}). This shows that
    \begin{align}
        \sum_{k=1}^n A_{ik}^2 m_{i\to k}^{s-1} \pd{f_t}{w^s}(w_k^t, \dots, w_k^0) \eqinf f_{s-1}(w_i^{s-1},\ldots,w_i^0) \sum_{k=1}^n A_{ik}^2 \pd{f_t}{w^s}(w_k^t, \dots, w_k^0)\,.\label{eq:first-step-cavity}
    \end{align}
    
    Second, we argue about the replacement of $A_{ik}^2$.
    This double edge is only non-negligible if it is hanging (\cref{lem:connected-nonnegligible}).
    Among the diagrams in $\pd{f_t}{w^s}(w_k^t, \dots, w_0^t)$
    the only one which does not attach something to $k$ is the singleton diagram $\rootpic$.
    The coefficient of this diagram is the expected value (\cref{lem:constant}),
    \[\E\left[\pd{f_t}{w^s}(w_k^t, \dots, w_k^0) \right]\,.\]
    The expected value is equal to the empirical expectation up to
    negligible terms (\cref{lem:scalar-empirical-expectation}), 
    \[\E\left[\pd{f_t}{w^s}(w_k^t, \dots, w_k^0) \right] \eqinf \frac{1}{n} \sum_{k=1}^n \pd{f_t}{w^s}(w_k^t, \dots, w_k^0)\,.\]
    This implies
    \begin{align}
        \sum_{k=1}^n A_{ik}^2 \pd{f_t}{w^s}(w_k^t, \dots, w_k^0)\eqinf \frac 1 n \sum_{k=1}^n \pd{f_t}{w^s}(w_k^t, \dots, w_k^0)\,.\label{eq:second-step-cavity}
    \end{align}
    The desired statement follows from combining \cref{eq:first-step-cavity} and \cref{eq:second-step-cavity}.
\end{proof}

\begin{proof}[Proof of \cref{thm:bp-amp}]
    Replace the $\approx$ signs in the heuristic argument from \cref{sec:bp-amp-heuristic} by $\eqinf$ and use \cref{fact:eqinf-cavity} repeatedly.
\end{proof}

\subsection{Proving the cavity assumptions}
\label{sec:cavity-formal}

We examine the belief propagation iteration
\cref{eq:bp-memory} more closely.
The BP iterates have the following asymptotic structure.

\begin{lemma}\label{lem:mij-structure}
    $m_{i \to j}^t$ is asymptotically
    equivalent to a linear combination of cavity diagrams which 
    have a tree hanging off of $i$, no edges between the roots, and nothing attached to $j$.
    \begin{figure}[ht]
    \centering
    \includegraphics[scale=0.25]{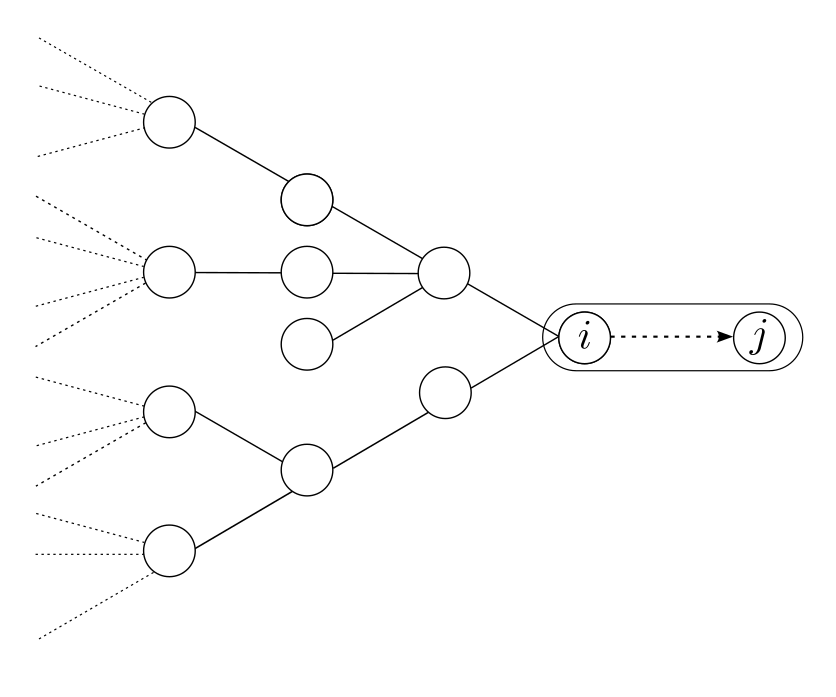} 
    \caption{
    Diagram representation of the cavity messages $m^t_{i\to j}$. Each cavity diagram in the asymptotic cavity diagram representation of $m^t_{i\to j}$ is a tree rooted at $i$.}
    \label{fig:cavity}
\end{figure}
\end{lemma}

\begin{proof}
Let $\al$ be a cavity diagram with the stated form. The vector whose $(i,j)$-th entry is $\sum_{\substack{k = 1}}^n A_{ik}Z_{\al,k \to i}$ is the sum of the diagrams which add an edge between the roots of $\al$, that can be of 3 types: (1)~the ``forward step'' diagram which puts the $j$ root as a new vertex, (2)~the ``backtracking step'' diagram which interchanges the first and second roots of $\al$, and (3)~other diagrams where $j$ intersects with a vertex from $V(\al) \setminus \{i\}$.

All diagrams of type (3) are negligible (and stay so when applying further
operations to them), because they create a cycle of length $>2$.
The backtracking step in (2) is canceled
by summing over $k \neq j$ in the belief propagation iteration. What asymptotically remains is the forward step (1) which again has the stated form. Additionally, componentwise functions preserve the stated form.
\end{proof}

\begin{theorem}
    \label{thm:cavity-formal}
    For any $j\in [n]$, the incoming messages
    at $j$, $\{m^t_{i\to j}:i\in [n]\setminus\{j\}\}$, are
    asymptotically independent (\cref{def:asymptotic-indpt}).
\end{theorem}

\begin{proof}
    When $j$ is ignored, the cavity diagrams in the asymptotic representation of $m_{i \to j}^t$ in \cref{lem:mij-structure} are equivalent to non-cavity diagrams (replacing $n$ by $n-1$).
    From the classification theorem (\cref{thm:classification}), these
    are asymptotically independent.
\end{proof}

\subsection{State evolution formula for BP/AMP}
\label{sec:amp}

We show how to simplify
the asymptotic state appearing in 
\cref{thm:general-state-evolution}
for the special case of approximate message
passing.
Recall the $+$ and $-$ operators from \cref{sec:calculus-rules}.

\begin{theorem}[Asymptotic state for AMP]
    \label{thm:amp-state}
    Under the same assumptions as \cref{thm:bp-amp}, the asymptotic state of $(w_t)_{t\le T}$ satisfies the recursion
    \begin{align}
        W_0 = 1\,,\qquad W_{t+1} = f_{t}(W_{t},\ldots,W_0)^+\,.\label{eq:amp-asymp}
    \end{align}
    In particular, $W_t$ is a centered Gaussian and for all $s,t\le T$, the covariances are
    \[
        \E \left[W_{s+1} W_{t+1}\right]= \E \left[f_s(W_s,\ldots,W_0) f_t(W_t,\ldots,W_0)\right]\,.
    \]
\end{theorem}

Combining \cref{thm:amp-state} and part (ii) of \cref{thm:general-state-evolution} recovers the typical formulation of state evolution for AMP algorithms. We note that while the formula for
computing iterates of AMP (\cref{eq:amp-output}) might look mysterious
at first sight, the AMP recursion in the
asymptotic space (\cref{eq:amp-asymp}) is much
easier to interpret.

We now prove \cref{thm:amp-state}.
Note that \cref{eq:amp} is not directly captured by the definition of a GFOM because $b_{s,t}$ requires computing an average over coordinates. This is only a technical issue: by \cref{lem:scalar-empirical-expectation}, empirical expectations are concentrated up to combinatorially negligible terms. 
Hence, the following inductive definition of a GFOM for $w_t\in\R^n$ and its corresponding asymptotic state $W_t$ is asymptotically equivalent to \cref{eq:amp}:
\begin{align}
    w_0=\vec{1}\,,\quad w_{t+1} = Af_t(w_t,\ldots,w_0) - \sum_{s=1}^t \E\left[\frac {\partial f_t}{\partial w_t}(W_t,\ldots,W_0)\right] f_{s-1}(w_{s-1},\ldots,w_0)\,.\label{eq:new-amp}
\end{align}

The Onsager correction term in \cref{eq:new-amp} will be rigorously interpreted as a backtracking term using diagrams.

\begin{lemma}\label{lem:taylor-expand}
    Let $W_1, \dots, W_t \in \Om$ be Gaussian (i.e. each $W_s$ is
    in the span of $(Z^\infty_\sig)_{\sig \in \calS}$). Then for any polynomial function $f : \R^t \to \R$,
    \[
        f(W_1, \dots, W_t)^- = \sum_{s = 1}^t  \E\left[\pd{f}{W_s}(W_1, \dots, W_t)\right] W_s^-\,.
    \]
\end{lemma}

\begin{proof}
    Expand $f(W_1, \dots, W_t)$ as
    \begin{align*}
    f(W_1,\ldots,W_t)^{\phantom{-}} &= \sum_{\sig \in \calS} c_\sig Z^\infty_\sig + \sum_{\tau \in \calT \setminus \calS} c_\tau Z^\infty_\tau\,,\\
    f(W_1,\ldots,W_t)^- &= \sum_{\sig \in \calS} c_\sig Z^\infty_{\sig^-}\,,
    \end{align*}
    for some coefficients $c_\tau\in\R$. When $\sig \in \calS$, we have
    \begin{align*}
    c_\sig\abs{\Aut(\sig)} &= \E \left[Z^\infty_\sig f(W_1,\ldots,W_t)\right] & (\text{orthogonality})\\
    &= \sum_{s=1}^{t} \E \left[Z_\sig^\infty W_s \right] \E \left[\frac{\partial f}{\partial W_s}(W_1,\ldots,W_t)\right] & (\text{\cref{lem:gaussian-ipp}})\\
    &= \sum_{s=1}^{t} \E \left[Z_{\sig^-}^\infty W_s^- \right] \E \left[\frac{\partial f}{\partial W_s}(W_1,\ldots,W_t)\right]\,. & (\text{\cref{lem:adjoint}})
    \end{align*}
    The second expectation does not depend on $\sigma$.
    Summing the first expectation over $\sigma$ produces $W_s^-$ as desired.
\end{proof}

Now we complete the proof of \cref{thm:amp-state}.

\begin{proof}[Proof of \cref{thm:amp-state}.]
    We prove by induction on $t$ that $W_{t+1}=f_{t}(W_{t},\ldots,W_0)^+$. For $t=0$, we have $w_1 = Af_0(\vec{1})$ so $W_1 = f_0(W_0)^+$ and the statement holds. 
    
    Now suppose that the statement holds for $W_1,\ldots,W_t$ for some $t<T$. The asymptotic state of $A f_t(w_t,\ldots,w_0)$ is $f_t(W_t,\ldots,W_0)^++f_t(W_t,\ldots,W_0)^-$. By the induction hypothesis and \cref{cor:pm-cancel}, for any $s\le t$,
    \[
        W_s^- = f_{s-1}(W_{s-1},\ldots,W_0)\,.
    \]
    Combining this with \cref{lem:taylor-expand}, we see that the asymptotic state of the Onsager correction term equals $f_t(W_t,\ldots,W_0)^-$. This concludes the induction.

    In particular, $W_{t+1}=f_t(W_t,\ldots,W_0)^+$ has no constant term and is in the span of $\calS$, so it has a centered Gaussian distribution. The covariances are, for all $s,t\le T$,
    \[
        \E\left[W_{s+1} W_{t+1}\right] = \E \left[f_s(W_s,\ldots,W_0)^+ f_t(W_t,\ldots,W_0)^+\right] = \E \left[f_s(W_s,\ldots,W_0) f_t(W_t,\ldots,W_0)\right]\,,
    \]
    where the last equality follows from \cref{cor:adjoint}. This completes the proof.
\end{proof}

\subsection{Montanari's iterative AMP algorithm}
\label{sec:montanari}

A special type of approximate message passing
iterations, called {\em iterative AMP},
was introduced by Montanari to
optimize Ising spin glass 
Hamiltonians~\cite{montanari2021optimization, AM20,AMS20:pSpinGlasses}. Here we reinterpret
iterative AMP and its analysis in the
asymptotic space.

The problem considered in \cite{montanari2021optimization} is to
optimize a degree-2 polynomial with random
coefficients over the hypercube (an average-case variant of the Max-Cut-Gain problem), i.e. given
$A$ satisfying \cref{assump:A-entries}, find
$x\in \{-1,1\}^n$ (approximately) solving 
\begin{align}
    \label{eq:sk-hamiltonian}
    \frac 1 n \max_{x\in \{-1,1\}^n}\langle x, Ax\rangle = \frac 1 n \max_{x\in \{-1,1\}^n} \sum_{i,j=1}^n A_{ij} x_i x_j\,.
\end{align}
The value of \cref{eq:sk-hamiltonian} is known to concentrate around the constant $2P_*\approx 1.52$~\cite{Talagrand06,carmona2006universality}.
Montanari gave an algorithm running in time
$n^{O_{\varepsilon}(1)}$ that, with high
probability over $A$ (as $n\to\infty$), finds an assignment
$x\in \{-1,1\}^n$ achieving a $(1-\varepsilon)$-approximation to \cref{eq:sk-hamiltonian}. This result is
conditional on the widely believed conjecture~\cite[Assumption 2]{montanari2021optimization}
that the problem exhibits no overlap gap.

Montanari's algorithm is an AMP iteration with
non-polynomial nonlinearities, although Ivkov and
Schramm~\cite[Lemma B.4]{ivkov2023semidefinite} proved that it
can be well-approximated by AMP with
polynomial
nonlinearities.\footnote{The assignment constructed
by the latter iteration is not precisely Boolean but it can be rounded to
$\{\pm 1\}^n$ with $o(1)$ loss 
in the objective value.} Iterative AMP \cite{montanari2021optimization} uses \cref{eq:amp} with the functions
\begin{equation}
f_t(w_t, \dots, w_0) = w_t \odot u_t(w_{t-1}, \dots, w_0) \label{eq:iamp}
\end{equation}
for chosen functions $u_t : \R^{t} \to \R$ applied componentwise, where $\odot$ denotes componentwise multiplication.
The candidate output of the algorithm is $x_T = \sum_{t = 1}^T w_t \odot u_t(w_{t-1}, \dots, w_0) = \sum_{t = 1}^T f_t(w_t, \dots, w_0)$.

The special property of iterative AMP is that it sums up \emph{independent} Gaussian
vectors $w_t$ scaled componentwise by the functions $u_t$.
The independence of the Gaussian vectors $w_t$ is contained in the state evolution for AMP as follows.
By \cref{thm:amp-state}, the asymptotic states $W_t, U_t, X_t$ of $w_t, u_t, x_t$ satisfy $U_0=W_0=1$,
\begin{align*}
    U_t = u_t(W_{t-1},\ldots,W_0)\,,\quad W_{t+1} = (U_t W_t)^+\,, \quad X_t = \sum_{s = 1}^t U_sW_s\,.\label{eq:iterative-asymptotic}
\end{align*}

\begin{claim}\label{lem:wt-depths}
    $U_t$ is in the span of trees in $\calT$ with depth at most $t-1$ and $W_t$ is in the span of trees in $\calS$ with depth exactly $t$.
\end{claim}
\begin{proof}
     Arguing inductively, as componentwise functions do not increase the depth, $U_t$ is in the span of trees from $\calT$ of depth at most $t-1$.
     In the product $U_t W_t$, the trees of depth $t$ in $W_t$ cannot be cancelled
     by any trees of lower depth from $U_t$. Therefore all trees in
     $U_tW_t$ and $W_{t+1} = (U_tW_t)^+$ have depth exactly $t$ and $t+1$ respectively, as needed.
\end{proof}
\cref{lem:wt-depths} provides a very
clear explanation of where the independent Gaussians in iterative AMP are coming from:
the $W_t$ have different depths, and Gaussian diagrams of different depths are asymptotically independent Gaussian vectors.

\paragraph{Optimality via state evolution.} The objective value achieved by the iteration can also be computed using state evolution:
\begin{align*}
    \frac 1n \langle  x_T, A x_T\rangle &\eqinf \E\left[X_T (X_T^+ + X_T^-)\right] & (\text{\cref{lem:scalar-empirical-expectation}})\\
    &= 2\E\left[X_T X_T^+\right] & (\text{\cref{lem:adjoint}})\\
    &= 2\sum_{s,t = 1}^T \E\left[U_s W_s (U_t W_t)^+\right]\\
    &= 2\sum_{s,t = 1}^T \E\left[U_s W_s W_{t+1}\right]\\
    &= 2\sum_{t = 2}^T \E\left[U_t W_t^2\right] & (\text{Independence of the $W_t$})\\
    &= 2\sum_{t = 2}^T \E\left[U_t\right] \E\left[W_t^2\right] & (\text{\cref{lem:wt-depths} and independence of the $W_t$})
\end{align*}
This gives an asymptotic description of the iterates $x_t$ (as asymptotically independent trajectories of $X_t$) and the objective value achieved by the algorithm (as above). We can now try to optimize
the best choice of the functions $u_t$ subject
to the constraint that the output point is Boolean.
This yields the following program for the value achievable by an iteration with $T$ steps (selecting $u_t(w_{t-1}, \dots, w_0)$ is equivalent to selecting $U_t$ which is measurable with respect to $W_{t-1}, \dots, W_0$):

\begin{align*}
    \max \;\;& 2\sum_{t=1}^T \E\left[U_t\right] \E[W_t^2]\\
    \textnormal{s.t. }\;\;
    &U_t \text{ is measurable w.r.t. $W_0, \dots, W_{t-1}$}\\
    &X_T = \sum_{t = 1}^T U_t W_t \in \{-1,+1\} \text{ a.s.}
\end{align*}
Note that by \cref{lem:wt-depths}, the trajectory $X_T = \sum_{t = 1}^T U_tW_t$ is a martingale.
As written, the Boolean constraint ``$X_T \in \{-1,1\}$ a.s.'' is infeasible because each step $W_t$ is a (potentially unbounded) Gaussian.
In the actual algorithm, we make $X_T$ close to $\pm 1$ then apply a final rounding step.

The remaining key step used by \cite{montanari2021optimization} is to take $T$ large
in order to approach a continuous time stochastic process $\d X_t = U_t \d B_t$.
The limiting Brownian motion only appears if we add a constraint that $\E\left[U_t^2\right] = 1$ so that $\E\left[W_t^2\right] = \E\left[U_t^2\right]\E\left[W_{t-1}^2\right] = \E\left[W_{t-1}^2\right]$ for all $t$.

This yields a continuous optimization problem for the best achievable value:
\begin{align*}
    \max \;\;& 2\int_0^1 \E\left[U_t\right] \dt\\
    \textnormal{s.t. }\;\;
    &(U_t)_{t \in [0,1]} \text{ is progressively measurable w.r.t. a Brownian motion $(B_t)_{t \in [0,1]}$}\\
    & \E\left[U_t^2\right] = 1 \text{ for all $t \in [0,1]$}\\
    &X_1 = \int_{0}^1 U_t \d B_t \in \{-1, +1\} \text{ a.s.}
\end{align*}

This continuous optimization problem is convex in $(U_t)_{t \in [0,1]}$ and dual to an ``extended Parisi formula'' for the optimal value of the SK model \cite[Section 4]{AMS20:pSpinGlasses}.
The remaining important technical step is to show that this program is well-posed, and that the maximizer of this program, which can be written in terms of the solution to the Parisi PDE, is smooth enough that it can be discretely approximated by the limit $T \to \infty$.

\section{Analyzing \texorpdfstring{$\poly(n)$}{poly(n)} Iterations}
\label{sec:polyn-iterations}

In summary so far, we have completely described the asymptotic trajectory of first-order algorithms for a \textit{constant} number of iterations.
We now discuss extensions to a number of iterations that scales with the dimension $n$ of the matrix.

A motivation for studying longer iterations is that
for problems with a hidden planted signal, it has been observed empirically that first-order iterations initialized at random can learn the planted signal. However, the standard machinery is only able to prove that these algorithms achieve recovery from an \emph{informative} initialization which has positive correlation with the planted signal. The underlying reason appears to be that ``picking up'' the signal and escaping the random initialization takes $\omega(1)$ steps, which is beyond what most previous works can handle.

\subsection{Combinatorial phase transitions}
\label{sec:combi-obstructions}

In order to show that this is a delicate question, we compute in \cref{sec:star}
that some diagrams of $\omega(1)$ size are no longer asymptotically Gaussian,
breaking the classification \cref{thm:classification}.
Larger-degree vertices in a diagram can access high moments
of the entries of other diagrams, which will detect that these quantities are not exactly Gaussian.

However, in typical first-order algorithms, high-degree diagrams only appear in a controlled way.
Thus we expect that for a class of ``nice'' GFOMs, the Gaussian tree approximation
continues to hold for many more iterations.
To demonstrate this, we examine \textit{debiased power iteration}, which is the iterative algorithm
\begin{equation}
     x_0=\vec{1}\,,\quad x_{t+1} = Ax_t - x_{t-1}\,.\label{eq:debiased-power-iteration}
\end{equation}
\cref{eq:debiased-power-iteration}
has a very simple tree approximation (the $t$-path diagram).
Note that by \cref{thm:bp-amp}, for constantly many iterations this algorithm is asymptotically equivalent to power iteration
on the non-backtracking walk matrix, which is the algorithm
\begin{align*}
    m_{0} = \vec{1}\,, \qquad m_{t+1} &= B m_{t}\,,\\
    x_{t+1,i} &= \sum_{k = 1}^n A_{ik} m_{t,k \to i}\,,
\end{align*}
where $B \in \R^{n^2 \times n^2}$ is the weighted non-backtracking walk matrix defined by $B_{i\to j,k\to \el} = A_{k\el}$ if $j=k$ and $i\neq \el$, and $B_{i\to j,k\to\el}=0$ otherwise.

We distinguish several regimes of $T=T(n)$ depending on the obstacles that arise when trying to generalize the tree approximation for \cref{eq:debiased-power-iteration} to a larger number of iterations.
\begin{itemize}
    \item When $T\ll \frac {\log n} {\log \log n}$, we expect the proofs of \cref{thm:classification} and \cref{thm:state-evolution} to generalize with minimal changes.
    The total number of diagrams that arise can be bounded by $T^{O(T)}$ which is $n^{o(1)}$ in this regime.
    
    \item When $T\approx \frac {\log n} {\log \log n}$, there are $T^{O(T)} = \textnormal{poly}(n)$ many diagrams to keep track of. This could overpower the magnitude of some cyclic diagrams, and make the naive union bound argument fail. 
    This barrier is also the one of previous non-asymptotic analyses of AMP \cite{rush2018finite, cademartori2023non}. 

    \item When $T \ll n^{\delta}$ for some small constant $\delta>0$, we show in the next subsections that the tree approximation of debiased power iteration still holds by a more careful accounting of the error terms.
    We predict that this can be extended up to $T \ll \sqrt{n}$.
      
    \item When $T\approx \sqrt n$, the tree diagrams with $T$ vertices are exponentially small in magnitude (see \cref{lem:variance}) and the number of non-tree diagrams
    starts to become overwhelmingly large. At the conceptual level, random walks of length~$\sgt \sqrt{n}$ in an $n$-vertex graph are likely to collide.
    Therefore, it is unclear whether or not the tree diagrams
    of size~$\sgt \sqrt{n}$ are significantly different from
    diagrams with cycles. This threshold also appears in recent analyses of AMP \cite{LFW23}, although it is not a barrier for their result.
\end{itemize}

\subsection{Analyzing power iteration via combinatorial walks}

For constantly many iterations of debiased power iteration, by \cref{thm:state-evolution}, we know that $x_t$ is well-approximated by the $t$-path diagram, denoted $Z_{t\textnormal{-path}}$.
Here we prove that this approximation holds much longer.
To simplify the calculation, we assume:

\begin{assumption}\label{assume:rademacher}
    Let $A$ be a random $n\times n$ symmetric matrix with $A_{ii}=0$ and $A_{ij}$ drawn independently from the uniform distribution over $\left\{-\frac{1}{\sqrt {n-1}}, \frac 1 {\sqrt {n-1}}\right\}$ for all $i<j$.
\end{assumption}

We prove that for this iterative algorithm we can extend \cref{thm:state-evolution} to a polynomial number of iterations, hence overcoming some obstructions mentioned in \cref{sec:combi-obstructions}.
A similar argument can also show that $Z_{t\text{-path}}$ remain
approximately independent Gaussians for $t$ in the same regime.
Taken together, we see that the ``usual'' state evolution formula for constantly many iterations continues to hold much longer,
up to conjecturally $\sqrt{n}$ iterations.

\begin{restatable}{theorem}{powerIteration}\label{thm:power-iteration}
    Suppose that $A=A(n)$ satisfies \cref{assume:rademacher} and generate $x_t$ according to \cref{eq:debiased-power-iteration}. Then there exist universal constants $c,\delta>0$ such that for all $t\le c n^{\delta}$,
    \[
        \left\|x_t-Z_{t\textnormal{-path}}\right\|_\infty\overset{a.s.}{\longrightarrow} 0\,.
    \]
\end{restatable}

To obtain the tree approximation of algorithms with $\poly(n)$ many iterations, we need to very carefully count combinatorial factors that were neglected in \cref{sec:combinatorial}.
The total number of diagrams in the unapproximated diagram expansion is very large, and furthermore, each diagram can arise in many different ways if it has high-degree vertices.
To perform the analysis, we decompose $x_t$ in terms
of walks of length $t$; we need to track walks instead of diagrams so that we do not throw away additional information about high-degree vertices.

Our goal is to show that the walk without any back edge (the $t$-path) dominates asymptotically. We will proceed as in the proof of \cref{thm:classification} by bounding the $q$-th moment of $x_t-Z_{\text{$t$-path}}$. This moment can be represented diagrammatically using $q$-tuples of non-backtracking walks with at least one back edge.

\begin{definition}\label{def:traversal}
    A $(q,t)$-traversal $\gamma=(\gamma_1,\ldots,\gamma_q)$ is an ordered sequence of $q$ walks, each of length $t$ and starting from the same vertex:
    \[
        \gamma_i = (\{u_{i,1}=\rootpic,u_{i,2}\},\{u_{i,2},u_{i,3}\},\ldots,\{u_{i,t},u_{i,t+1}\})\,,\quad\text{for all $i\in [q]$}.
    \]
    Each traversal $\gamma$ is naturally associated to an improper diagram $(V(\gamma),E(\gamma))$ with $V(\gamma) = \{u_{i,j}:i\in [q], j\in[t]\}$ and $E(\gamma)=\{(u_{i,j},u_{i,j+1}):i\in[q],j\in[t-1]\}$ (viewed as a multiset). We use the notation $Z_\gamma=Z_{(V(\gamma),E(\gamma))}$ following \cref{def:Zal}.
    
    \begin{itemize}
        \item A traversal is even if each edge appears an even number of times in $\bigcup_{i\in [q]} \gamma_i$.
        \item A traversal is non-backtracking if every walk of the traversal is non-backtracking, i.e. $u_{i,j+1}\neq u_{i,j-1}$ for all $i\in [q]$ and $j\in \{2,\ldots,t-1\}$.
        \item A traversal is non-full-forward if every walk of the traversal has a back edge, namely for all $i\in [q]$, there exist $j_1\neq j_2$ such that $u_{i,j_1}=u_{i,j_2}$.
    \end{itemize}

    Let $\calW_t^q$ be the set of $(q,t)$-traversals that are simultaneously even, non-backtracking, non-full-forward, and have no self-loops.
\end{definition}

\cref{def:traversal} is motivated by the following decomposition:

\begin{claim}\label{claim:walk-decomposition}
    Suppose that $x_t$ is generated according to \cref{eq:debiased-power-iteration} and $A$ satisfies \cref{assume:rademacher}. Then,
    \begin{align*}
        \E \left[(x_t-Z_{t\textnormal{-path}})^q\right] &= \sum_{\gamma\in\calW_t^q} \E \left[Z_\gamma\right]\,.
    \end{align*}
\end{claim}

We now proceed to proving \cref{thm:power-iteration}. We will bound the magnitude of $\E\left[Z_{\gamma,i}\right]$ for $\gamma\in\calW^q_t$, then count the number of traversals in $\calW^q_t$. Both bounds will depend on $\frac E 2 - V + 1$ (where $V$ is the number of vertices of the traversal and $E$ the number of edges), which quantifies how close the traversal is to a tree of double edges. 

Our first insight is that the traversals contributing to $(x_t-Z_{t\textnormal{-path}})^q$ become further from trees as $q$ increases because each walk must have a back edge.

\begin{lemma}
    For any $\gamma\in \calW_t^q$ with $V$ vertices and $E$ edges, $\frac E 2 - V + 1\ge \frac q 2$.
\end{lemma}

\begin{proof}
    Assign to each vertex all the edges going into it in $\gamma$.
    Each non-root vertex must have at least 2 incoming edges:
    the edge which explores it, and since $\gam$ is even and non-backtracking,
    an edge which revisits it a second time.
    Since $\gam$ is non-full-forward, each $\gam_i$ has a back edge; the first back edge in each $\gam_i$ yields an additional incoming edge for each $i$ (either it points to the root, which has not yet been counted, or by assumption that it is the \emph{first} back edge in $\gam_i$, it cannot cover both incident edges from the first visit). We have
    \[E \geq 2(V-1) + q\,,\]
    as needed.
\end{proof}

\begin{lemma}\label{lem:walk-magnitude}
    For any $i\in[n]$ and $\gamma\in\calW_t^q$ with $V$ vertices and $E$ edges,
    \[
        \left|\E\left[Z_{\gamma,i}\right]\right|\le O\left(n^{-\left(\frac E 2 - V + 1\right)}\right)\,.
    \]
\end{lemma}

\begin{proof}
    Using \cref{assume:rademacher}, we can directly count
    \begin{align*}
        \left|\E\left[Z_{\gamma,i}\right]\right| &\le O(1)\cdot \frac {(n-1)(n-2)\cdots (n-V+1)} {n^{\frac E 2}}\\
        &= O\left(n^{V-1 - \frac{E} 2}\right)\,.\qedhere
    \end{align*}
\end{proof}

Finally, the following lemma captures the counting of traversals. Its proof is deferred to the next subsection.

\begin{lemma}\label{lem:walk-number}
    The number of $\gamma\in\calW_t^q$ with $V$ vertices and $E$ edges is at most
    \[
        O_q(t)^{6\left(\frac{E}{2} - V+1\right)+2q}\,,
    \]
    where $O_q(\cdot)$ hides a constant depending only on $q$.
\end{lemma}

\begin{proof}[Proof of \cref{thm:power-iteration}]
    We decompose the sum over $\gamma\in\calW_t^q$ according to the value of $r=\frac E 2 - V + 1$ using \cref{lem:walk-magnitude} and \cref{lem:walk-number}:
    \[
        \E \left[(x_{t,i}-Z_{t\textnormal{-path},i})^q\right] \le O_q(t)^{2q} \sum_{r\ge \frac q 2}         O_q(t)^{6r} n^{-r}\,.
    \]
    If $t$ satisfies $t\le cn^{\delta}$ with $0<\delta<1/6$, the sum is a geometrically decreasing series and therefore it is bounded by the first term which is $O_q(t^{5q}n^{-\frac q 2})$. Under the condition $\delta<1/10$, for $q$ being a large enough integer we obtain for some $\eps > 0$,
    \[
        \E \left[(x_{t,i}-Z_{t\textnormal{-path},i})^q\right]\le O(1/n^{2+\eps})\,.
    \]
    This is enough to imply that $\|x_{t}-Z_{t\textnormal{-path}}\|_\infty
    \overset{a.s.}{\longrightarrow} 0$ by a union bound over the $n$ coordinates, then Markov's inequality and the Borel-Cantelli lemma.
\end{proof}

\subsection{Counting combinatorial walks}

Our goal here is to prove \cref{lem:walk-number}. 

In the extreme case $V\approx \frac {E} 2$ where the moment bound \cref{lem:walk-magnitude} is the weakest, typical traversals $\gamma\in\calW^q_t$ look like trees of double edges with a constant number of back edges. In this regime, most vertices will have degree exactly $4$. Following this intuition, our encoding will proceed by compressing the long paths of degree-4 vertices connected by double edges.

\begin{definition}\label{def:contracted-walk}
    For $\gamma\in \calW_t^q$, let $\gamma_c$ be the traversal obtained by replacing all maximally long paths of degree-4 vertices in $\gamma$ by a single special marked edge between the endpoints of the paths, and removing the internal vertices of the path.
    (The paths should be broken at the root so that it is not removed.)
\end{definition}

Note that these operations can create self-loops in $\gamma_c$.

\begin{lemma}\label{lem:size-contracted}
    For any $\gamma\in\calW_t^q$,
    \[
        |E(\gamma_c)|\le 3|E(\gamma)|-6(|V(\gamma)|-1)+2q\,.
    \]
\end{lemma}

\begin{proof}
    For $k\in\N$, let $V_k(\gamma)$ be the set of non-root vertices of $\gamma$ of degree exactly $k$. Since $\gamma$ is an even traversal, we get by double counting the number of edges in $\gamma$
    \[
        2|V_2(\gamma)|+4 |V_4(\gamma)| + 6\left(|V(\gamma)|-|V_2(\gamma)|-|V_4(\gamma)|-1\vphantom{\sum}\right)\le 2|E(\gamma)|\,.
    \]
    Moreover, the number of edges removed during the compression is $2|V_4(\gamma)|$. This means that
    \[
        |E(\gamma)|-|E(\gamma_c)| = 2|V_4(\gamma)|\ge  6(|V(\gamma)|-1)-4|V_2(\gamma)|-2|E(\gamma)|\,.
    \]
    Finally, since $\gam$ is non-backtracking, non-root degree-2 vertices can only be created in $\gamma$ by pairing endpoints of the walks, so that $|V_2(\gamma)|\le q/2$. The desired inequality immediately follows.
\end{proof}

We are now ready to prove \cref{lem:walk-number}.

\begin{proof}[Proof of \cref{lem:walk-number}]
    We encode a traversal $\gamma\in\calW_t^q$ as follows:
    \begin{enumerate}
        \item We first encode $\gamma_c$. 
        We write down the sequence of vertices of each walk and indicate whether each step should be the first step of a marked edge (\cref{def:contracted-walk}). 
        Every time we traverse a marked edge for the second time, instead of recording the next vertex of the walk, we record the identifier of the marked edge. We also add a single bit of information to each edge to indicate whether it is the last edge of its walk. The target space of the encoding has size $O(|E(\gamma_c)|)^{|E(\gamma_c)|}$. 
        
        \item We then expand the marked edges in $\gamma_c$ of which there are at most $|E(\gam_c)|/2$. For each marked edge, we write down the length of the path that it replaced. This can be encoded using ``stars and bars''. Initially allocating 2 edges to each marked edge, there are at most $\binom{E}{\abs{E(\gam_c)}/2}$ such objects.
    \end{enumerate}
    We claim that this encoding allows to reconstruct $\gamma$, and its length can be bounded by
    \begin{align*}
        &O(|E(\gamma_c)|)^{|E(\gamma_c)|} \binom{E}{|E(\gamma_c)|/2}
        &\le O(|E(\gamma_c)|)^{|E(\gamma_c)|} O\left(\frac{E}{|E(\gamma_c)|}\right)^{|E(\gamma_c)|/2}
        &=O_q(t)^{|E(\gamma_c)|}\,.
    \end{align*}
    The proof follows after plugging in the bound of \cref{lem:size-contracted}.
\end{proof}

\addcontentsline{toc}{section}{References}
\bibliographystyle{alpha}
{\footnotesize\bibliography{references}}

\appendix

\section{Non-asymptotic Diagram Analysis}
\label{app:non-asymptotic}

\subsection{Fourier analytic properties}
\label{sec:fourier}

In \cref{def:Zal}, for a proper $\al \in \calA$ (a graph instead of a multigraph), $Z_\al$ has entries which are homogeneous multilinear polynomials in the entries of the matrix $A$.
The next lemma shows that the proper diagrams with size at most $n$ form an orthogonal basis of
symmetric polynomials in $A$ with respect to the expectation over $A$.

\begin{lemma}\label{lem:orthogonality}
    For all $i,j \in [n]$ and distinct proper diagrams $\al, \beta \in \calA$, $\E \left[Z_{\al,i}Z_{\beta,j}\right] = 0$.
\end{lemma}
\begin{proof}
    For each distinct $S, T \subseteq \binom{[n]}{2}$, the independence and centeredness of
    the off-diagonal entries of $A$ proves that
    \[
        \E \left[\prod_{\{i,j\} \in S} A_{ij} \prod_{\{k,\el\} \in T} A_{k\el}\right] = 0\,.
    \]
    Two distinct diagrams sum over distinct sets of multilinear monomials,
    so this orthogonality extends to diagrams.
\end{proof}

The diagrams are not normalized for that inner product, but their variance can be estimated as follows:

\begin{lemma}\label{lem:variance}
    For all $i \in [n]$ and proper $\al \in \calA\setminus\{\rootpic\}$ 
    we have $\E \left[Z_{\alpha,i}\right]=0$ and
    \begin{align*}
        \E\left[Z_{\al, i}^2\right] &\underset{\phantom{n\to\infty}}{=} \abs{\Aut(\al)} \cdot \frac{(n-1)(n-2)\cdots(n-|V(\al)| + 1)}{n^{|E(\al)|}} \\
        &\underset{n\to\infty}{=} \abs{\Aut(\al)} \cdot n^{\abs{V(\al)} - 1 - \abs{E(\al)}} (1 + o(1))\,,
    \end{align*}
    where the last estimate holds when $|V(\al)| = o(\sqrt{n})$.
\end{lemma}
\begin{proof}
    When $\al$ is proper, $Z_{\alpha,i}$ is a multilinear polynomial with zero constant coefficient, and so it has expectation 0. For the second moment, we have
    \begin{align*}
        \E\left[Z_{\al, i}^2\right] &= \sum_{\substack{\text{injective }\ph_1: V(\al) \to [n]\\\ph_1(\smallrootpic) = i}} \sum_{\substack{\text{injective }\ph_2: V(\al) \to [n]\\\ph_2(\smallrootpic) = i}} \E \left[\prod_{\{u,v\} \in E(\al)} A_{\ph_1(u)\ph_1(v)} A_{\ph_2(u)\ph_2(v)}\right]\,.
    \end{align*}
    Since $\E \left[A_{jk}\right] = 0$ for $j\neq k$, the only terms with nonzero expectation have each $A_{jk}$ occurring at least twice.
    As $\ph_1$ are $\ph_2$ are injective, each $A_{jk}$ can only occur at most twice.
    Therefore, if we fix $\ph_1$ the embeddings $\ph_2$ that contribute a nonzero value are exactly
    graph isomorphisms onto $\im(\ph_1)$.
    The total number of choices for $\ph_1$ and $\ph_2$ is $(n-1)\cdots(n-\abs{V(\al)}+1) \cdot \abs{\Aut(\al)}$ and the expectation of a nonzero term is
    \[
        \prod_{\{j,k\} \in E(\al)}\E\left[A_{jk}^2\right] = \frac{1}{n^{|E(\al)|}} \,.
    \]
    This completes the proof of the first part of the statement. Under the further assumption $|V(\al)|=o(\sqrt n)$, the falling factorial can then be estimated  as
    \begin{align*}
        \left|\log\left(\frac{(n-1)\ldots(n-|V(\al)|+1)}{n^{|V(\al)|-1}}\right)\right| &\le \sum_{i=1}^{|V(\al)|-1} \left|\log\left(1-\frac i n\right)\right|\\
        &\le \sum_{i=1}^{|V(\al)|-1} \frac i n\underset{n\to\infty}{\longrightarrow} 0\,.
    \end{align*}
    This implies that $(n-1)\ldots(n-|V(\al)|+1)=(1+o(1)) n^{|V(\al)|-1}$, as desired.
\end{proof}

We can already see from the previous lemma that if $\al \in \calT$ is a tree,
then the variance of $Z_{\al,i}$ is $\Theta(1)$,
whereas if $\al$ is a connected graph with a cycle, then the variance of $Z_{\al,i}$ is $o(1)$.

We will use orthogonality repeatedly in the sequel through the following direct consequence of \cref{lem:orthogonality} and \cref{lem:variance}:

\begin{corollary}
\label{lem:constant}
    Let $x = \sum_{\textnormal{proper }\al \in \calA} c_\al Z_\al$. Then for any $\tau\in\calT$,
    \[
        \E\left[x_i Z_{\tau,i}\right] = c_\tau \E\left[Z_{\tau,i}^2\right] 
        \underset{n\to\infty}{=} c_\tau \abs{\Aut(\tau)} + o(1)\,,
    \]
    where the second estimate holds for $|V(\tau)| = o(\sqrt{n})$.

    In particular, $\E \left[x\right] = c_{\smallrootpic} \vec{1}$ where $c_{\smallrootpic}$
    is the coefficient of the singleton diagram.
\end{corollary}

\subsection{Operations on the diagram representation}
\label{sec:derivation-asymptotic}

We compute the diagrammatic effect of multiplying by $A$.
\begin{lemma}\label{lem:mat-mul}
    For all diagrams $\al \in \calA$,
    \[AZ_\al = Z_{\al^+} + \sum_{v \in V(\al)} Z_{\textnormal{contract $v$ and $\rootpic$ in $\al^+$}}\,.\]
\end{lemma}
\begin{proof}
\begin{align*}
    (AZ_\al)_i
    &= \sum_{j = 1}^n A_{ij} \sum_{\substack{\ph: V(\al) \to [n]\\\ph\text{ injective}\\\ph(\smallrootpic) = j}} \prod_{\{u,v\} \in E(\al)} A_{\ph(u)\ph(v)}\\
    &= \sum_{\substack{\ph: V(\al) \to [n]\\\ph \text{ injective}}} A_{i,\ph(\smallrootpic)}\prod_{\{u,v\} \in E(\al)} A_{\ph(u)\ph(v)}\,.
\end{align*}
The sum over $\ph$ can be partitioned based on whether $i \in \im(\ph)$.
The terms with $i \not \in \im(\ph)$ sum to $Z_{\al^+}$.
The terms with $i \in \im(\ph)$ sum to the different contractions
of $\al^+$ based
on which vertex of $\al$ is labeled $i$.
\end{proof}

Switching to componentwise operations, the combinatorics is captured by the concepts of intersection patterns and intersection diagrams.

\begin{definition}[Intersection pattern, $P \in \calP(\al_1, \dots, \al_k)$]
\label{def:intersection-pattern}
    Let $\al_1, \dots, \al_k \in \calA$. Let $\al$
    be the diagram obtained by putting all $\al_i$ at the same root.
    An intersection pattern $P$ is a partition of $V(\al)\setminus \{\rootpic\}$
    such that for all $i\in[k]$ and $v,w \in V(\al_i)\setminus\{\rootpic\}$, $v$ and $w$
    are not in the same block of the partition.

    Let $\calP(\al_1, \dots, \al_k)$ be the set of intersection patterns between $\al_1, \dots, \al_k$.
\end{definition}

\begin{definition}[Intersection diagram, $\al_P$]
\label{def:intersection-diagram}
    Let $\al \in \calA$.
    Given a partition $P$ of $V(\al)$,
    let $\al_P$ be the diagram obtained by contracting
    each block of $P$ into a single vertex. Keep all edges
    (hence there may be new multiedges or self-loops).
\end{definition}

By casing on which vertices are equal among the embeddings
of $\al_1, \dots, \al_k$ as in the proof of \cref{lem:mat-mul}, we have:
\begin{lemma}\label{lem:diagrams-product}
    For $\al_1, \dots, \al_k \in \calA$, the componentwise product of $Z_{\al_1},\ldots,Z_{\al_k}$ is
    \[Z_{\al_1} \odot \cdots \odot Z_{\al_k} = \sum_{P \in \calP(\al_1, \dots, \al_k)} Z_{\al_P}\,.\]
\end{lemma}

Next we consider these operations when restricted to the tree diagrams.
Suppose we start from $\tau \in \calT$ and compute $AZ_\tau\,$.
Which diagrams appearing in \cref{lem:mat-mul} are non-negligible?
Following the asymptotic classification of non-negligible diagrams (\cref{sec:classification}),
it is only $\tau^+$ and $\tau^-$ (the latter only appears if the root of $\tau$ has degree 1, in which case $\tau^-$ is the result of intersecting $\rootpic$ and the child of the root then removing a double edge).
Hence we conclude
\[AZ_\tau \eqinf \begin{cases}
    Z_{\tau^+} + Z_{\tau^-} & \text{ if }\tau \in \calS\\
    Z_{\tau^+}& \text{ if }\tau \in \calT \setminus \calS\,.
\end{cases}\]

Given tree diagrams $\tau_1, \dots, \tau_k \in \calT$, the asymptotically
non-negligible terms in the product in \cref{lem:diagrams-product} are identified as follows.
Let $\widetilde{\tau}$ be a non-negligible diagram appearing in the result,
i.e. $\widetilde{\tau}$ is a tree with hanging trees of double edges.
Since $\tau_1, \dots, \tau_k$ are connected,
the hanging double trees must hang off the root vertex of $\widetilde{\tau}$
in order to avoid cycles.
Additionally, they must arise as the overlap of two complete copies of the tree.
Thus the asymptotically non-negligible terms are the partial
matchings between isomorphic branches of the roots of the $\tau_i$.
Two copies of a branch $\sig \in \calS$ can be matched up into a tree
of double edges in $\abs{\Aut(\sig)}$ ways.

Based on these observations, the \textit{tree approximation} is formally
defined to be the result of applying the algorithmic operations and removing the non-trees at each step.

\begin{definition}[Tree approximation of a GFOM, $\widehat{x}_t$]
\label{def:asymptotic-state-explicit}
    Let $x_t\in\R^n$ be the state of a GFOM. We recursively define the tree approximation of $x_t$, denoted by $\widehat{x}_t$, to be a diagram expression in the span of $(Z_\tau)_{\tau \in \calT}$.
    \begin{enumerate}
        \item Initially, $\widehat{x}_0 = Z_{\smallrootpic}$.
        \item If $x_{t+1}=Ax_t$, define $\widehat{x}_{t+1} = (\widehat{x}_t)^+ + (\widehat{x}_t)^-$.
        \item If $x_{t+1}=f_t(x_t, \dots, x_0)$ coordinatewise for some polynomial $f_t:\R^t\to\R$, define
        $\widehat{x}_{t+1}$ by applying each monomial of $f_t$ to $\widehat{x}_t,\ldots,\widehat{x}_0$ separately and summing the results. To apply a monomial on $\widehat{x}_t, \dots, \widehat{x}_0$, expand each $\widehat{x}_s$ in the diagram basis and sum all the cross product terms. The result of multiplying $q$ tree diagrams $\tau_1, \dots, \tau_q \in \calT$ is
        \[\sum_{M \in \calM(\tau_1, \dots, \tau_q)} c_M Z_{\tau_M}\,,\]
        where:
        \begin{enumerate}[(a)]
            \item $\calM(\tau_1, \dots, \tau_q)$ is the set of (partial) matchings
            of isomorphic branches of $\tau_1, \dots, \tau_q$ such that no two branches from the same $\tau_i$ are matched.
            \item $\tau_M$ is the tree obtained by merging the roots of $\tau_1, \dots, \tau_q$ and removing all subtrees matched in $M$.
            \item $c_M = \prod_{\{\sigma, \sigma'\} \in M} \abs{\Aut(\sigma)}$.
        \end{enumerate}
    \end{enumerate}
\end{definition}

\subsection{Repeated-label diagram basis}
\label{sec:repeated-labels}

An alternative basis for the diagram space consists of diagrams
in which labels are allowed to repeat.
This representation has been defined by Ivkov and Schramm \cite[Section 3.5]{ivkov2023semidefinite}.

\begin{definition}[$\widetilde{Z}_\al$]
For a diagram $\al$ with
root $\rootpic$, define $\widetilde{Z}_\al \in \R^n$ by
\[\widetilde{Z}_{\al,i} = \sum_{\substack{\ph: V(\al) \to [n]\\\ph(\smallrootpic) = i}} \prod_{\{u,v\} \in E(\al)} A_{\ph(u)\ph(v)}\,.\]
\end{definition}

The only difference between $\widetilde{Z}_\al$ and $Z_\al$ is that the embedding
$\ph$ must be injective in $Z_\al$.
To perform the change of 
basis in one direction is as easy as replacing $\widetilde{Z}_\al$
by a sum of $Z_\al$
based on which labels are repeated.

\begin{lemma}\label{lem:change-basis}
    For $\al \in \calA$,
    \[\widetilde{Z}_\al = \sum_{P \in \calP(\al)} Z_{\al_P}\]
    where $\calP(\al)$ is the set of partitions of $V(\al)$ and $\al_P$
    contracts the blocks of $P$ (\cref{def:intersection-diagram}).
\end{lemma}
\begin{proof}
    We have
    \[\widetilde{Z}_{\al,i} = \sum_{\substack{\ph: V(\al) \to [n]\\\ph(\smallrootpic) = i}} \prod_{\{u,v\} \in E(\al)} A_{\ph(u)\ph(v)}\,.\]
    The sum over $\ph$ can be divided based on which vertices are assigned the same label.
    The terms with a given partition $P$ of $V(\al)$ are exactly $Z_{\al_P, i}$.
\end{proof}

The algorithmic operations are simpler to compute in this basis,
although the asymptotic tree approximation does not seem to be easily visible
in this basis
(the tree diagrams do not span the same space, and a diagram which is an even cycle has entries with magnitude $\Theta(1)$ in $\widetilde{Z}_\al$ but negligible entries in $Z_\al$).

Given the current representation $x_t = \sum_{\tau \in \calT} c_\tau \widetilde{Z}_\tau$
the operations have the following effects on the $\widetilde{Z}_\tau$
(non-asymptotically i.e. without taking the limit $n \to \infty$).
\begin{enumerate}[(i)]
    \item \textbf{Multiplying by $A$ extends the root.}

    We have $A \widetilde{Z}_\al = \widetilde{Z}_{\al^+}$ where $\al^+$ is obtained by extending the root by one edge.
    \item \textbf{Componentwise products graft trees together.}

    To componentwise multiply $\widetilde{Z}_\al$ and $\widetilde{Z}_\beta$,
    we ``graft'' $\al$ and $\beta$ by merging their roots.
\end{enumerate}

\begin{example}
Consider the example,
\[x_{t+1} = (Ax_t)^2 \qquad\qquad x_0 = \vec{1}\]
where $\vec{1} \in \R^n$ is the all-ones vector and the square function
is applied componentwise.
The first few iterations are,
    \begin{center}
    \begin{tabular}{c|c|c}
        $x_0 = \vec{1}$ & $x_1 = (A\vec{1})^2$ & $x_2 = (A(A\vec{1})^2)^2$\\
        $x_{0,i} = 1$ & $\displaystyle x_{1,i} = \sum_{\substack{j_1, j_2 = 1}}^n A_{ij_1}A_{ij_2}$ & $\displaystyle  x_{2,i} = \sum_{j_1,j_2 = 1}^n \sum_{k_1, k_2 = 1}^n\sum_{\el_1,\el_2 = 1}^n A_{ij_1}A_{ij_2} A_{j_1k_1} A_{j_1\el_1} A_{j_2k_2}A_{j_2\el_2}$\\
        $\vcenter{\hbox{\includegraphics[height=3ex]{images/singleton.png}}}$& $\vcenter{\hbox{\includegraphics[height=7ex]{images/2tree.png}}}$ & $\vcenter{\hbox{\includegraphics[height=18ex]{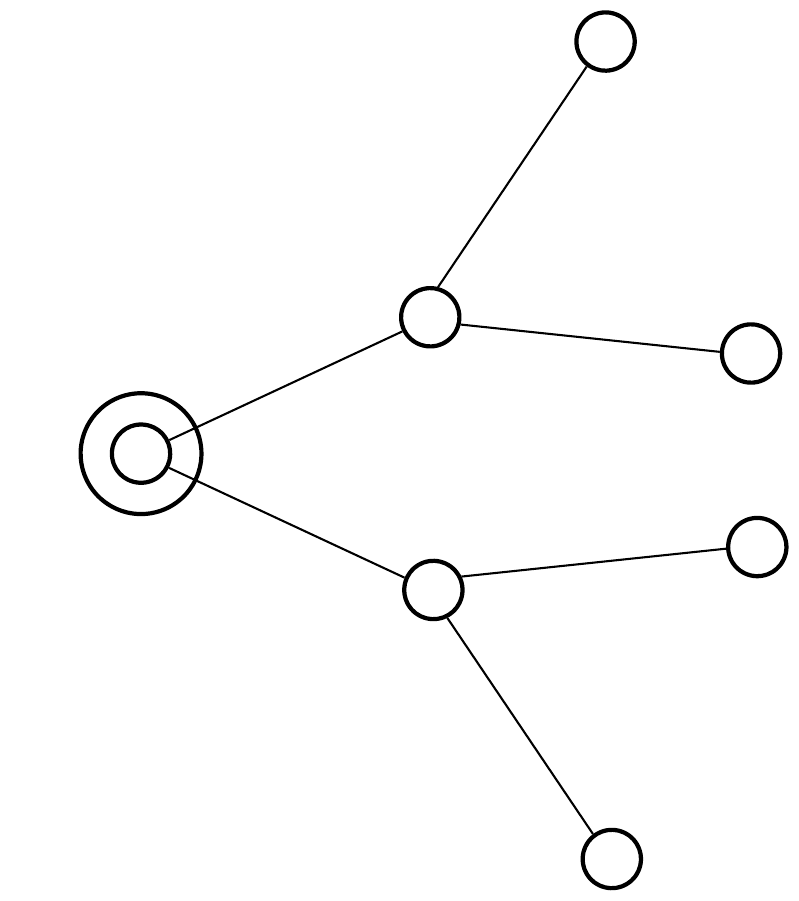}}}$
    \end{tabular}
    \end{center}
\end{example}

\section{Omitted Proofs}
\label{app:non-asymptotic-analysis}

\subsection{Removing hanging double edges}
\label{sec:edge-labels}

In order to implement the removal of hanging double edges, we introduce
an additional diagrammatic construct to track the error,
\emph{2-labeled edges}.
These terms are equal to zero when $A$ is a Rademacher matrix
and it is recommended to ignore them on a first read.

\begin{definition}[Edge-labeled diagram]
    \label{def:labeled-diagram}
    An edge-labeled diagram is a diagram in which some of the edges
    are labeled ``2''.

    We let $E(\al)$ denote the entire multiset of labeled and unlabeled edges of $\alpha$, $E_2(\al)$
    the multiset of 2-labeled edges and $E_1(\al)=E\setminus E_2(\al)$ the multiset of non-labeled edges.

    We use the convention that $\abs{E(\al)}$ counts each 2-labeled edge
twice, so that $\abs{E(\al)}$ continues to equal the degree of the polynomial $Z_{\alpha,i}$.
\end{definition}

\begin{definition}[Edge-labeled $Z_\al$]
    For an edge-labeled diagram $\al$, we define $Z_\al \in \R^n$ by
    \[Z_{\al, i} = \sum_{\substack{\textnormal{injective }\ph: V(\al) \to [n]\\\ph(\smallrootpic) = i}} \prod_{\{u,v\} \in E_1(\al)} A_{\ph(u)\ph(v)} \prod_{\{u,v\} \in E_2(\al)} \left(A_{\ph(u)\ph(v)}^2 - \frac{1}{n}\right)\,.\]
\end{definition}

The set of diagrams $\calA$ is extended to allow diagrams which may have 2-labeled edges.
The definition of $I(\al)$ from \cref{def:isolated} must also be updated to incorporate labeled edges
(because a labeled edge is mean-0, it is treated like a single edge).
\begin{definition}[Updated definition of $I(\al)$]
     For a diagram $\al \in \calA$, let $I(\al)$ be the subset
     of non-root vertices such that every edge incident to that vertex has multiplicity $\geq 2$ or is a self-loop, treating 2-labeled edges as if they were normal edges.
 \end{definition}

The following is an exact decomposition for removing hanging double edges.
\begin{lemma}\label{lem:syntactic-removal}
    Let $\al \in \calA$ be a diagram with a hanging (unlabeled) double edge. Let $\al_0$ be $\alpha$ with both the hanging double edge and corresponding hanging vertex removed, and $\al_2$ be $\alpha$ with the hanging double edge replaced by a single 2-labeled edge. Then,
        \[
            Z_\al = Z_{\al_0} - \frac{|V(\al)|-1}{n} \cdot Z_{\al_0} + Z_{\al_2}\,.
        \]
\end{lemma}

\begin{proof}
    We write:
    \begin{align*}
        Z_{\al,i} &= \sum_{\substack{\text{injective }\ph: V(\al) \to [n]\\\ph(\smallrootpic) = i}} A_{u,v}^2\prod_{\{x,y\} \in E(\al) \setminus \{\{u,v\},\{u,v\}\}} A_{\ph(x)\ph(y)}\\
        &= Z_{\alpha_2,i} + \frac 1 n \sum_{\substack{\text{injective }\ph: V(\al) \to [n]\\\ph(\smallrootpic) = i}} \prod_{\{x,y\} \in E(\al) \setminus \{\{u,v\},\{u,v\}\}} A_{\ph(x)\ph(y)}\\
        &= Z_{\alpha_2,i} + \frac {n-|V(\alpha)|+1} n \cdot Z_{\alpha_0,i} = Z_{\al_0, i} - \frac{|V(\al)| - 1}{n}Z_{\al_0, i} + Z_{\al_2,i}\,.
    \end{align*}
    The additional $n-|V(\alpha)|+1$ scaling factor comes from removing the hanging vertex.
\end{proof}

\subsection{Omitted proofs for \texorpdfstring{\cref{sec:eqinf}}{Section~\ref{sec:eqinf}}}
\label{sec:omitted-combneg}

We prove a more specific version of \cref{lem:improper-magnitude-simplified}.
\begin{lemma}\label{lem:improper-magnitude}
    Let $q\in\N, \al \in \calA,$ and $i \in [n]$. Then,
    \[
        \abs{\E\left[Z_{\al,i}^q\right]} \leq M_{q|E(\al)|} 2^{q|E(\al)|}(q\abs{V(\al)})^{q|V(\al)|} \cdot n^{\frac q 2\left(|V(\al)| - 1 - |E(\al)|+|I(\al)|\right)}\,,
    \]
    where $M_{k}$ is a bound on the $k$-th moment of the entries of $A$ (recall the notations of \cref{assump:A-entries}),
    \[
        M_{k} = \max\left(\E_{X\sim \mu} \left[\abs{X}^k\right], \;
        \E_{X\sim \mu_0} \left[\abs{X}^k\right]\right) \,.
    \]
    When $q$ and $|V(\al)|$ are $O(1)$, the overall bound reduces to
    \[
        \abs{\E\left[Z_{\al,i}^q\right]} \leq O\left( n^{\frac q 2\left(|V(\al)| - 1 - |E(\al)|+|I(\al)|\right)}\right)\,.
    \]
\end{lemma}
\begin{proof}
    We expand $\E\left[Z_{\al,i}^q\right]$ as
    \begin{align*}
         \sum_{\substack{\text{injective }\ph_1,\ldots,\ph_q : V(\al) \to [n]\\\ph_1(\smallrootpic) = \cdots=\ph_q(\smallrootpic)=i}} \E \left[\prod_{p=1}^q \left(\prod_{\{u,v\} \in E_1(\al)} A_{\ph_p(u)\ph_p(v)} \right)\left(\prod_{\{u,v\} \in E_2(\al)} \left(A_{\ph_p(u)\ph_p(v)}^2-\frac 1 n\right) \right)\right]\,.
    \end{align*}
    This is a polynomial of degree $q|E(\al)|$ in $A$ (by convention every 2-labeled edge contributes $2$ to $|E(\alpha)|$).
    We first estimate the magnitude of any summand of the sum over $\ph_1,\ldots,\ph_q$ with nonzero expectation. Each such summand can be decomposed into $2^{q|E_2(\al)|}$ terms by expanding out\footnote{The factor $2^{q|E_2(\al)|}$ may be removed with a tighter argument.} the $A_{ij}^2 - \frac{1}{n}$. This leaves monomials in the entries of $A$ of total degree at most $q|E(\al)|$. We bound the expected value of each of these monomials by $M_{q|E(\al)|} n^{-q|E(\al)|/2}$ using Hölder's inequality. This shows that any nonzero term in the summation has magnitude at most $2^{q|E_2(\al)|} M_{q|E(\al)|} n^{-q|E(\al)|/2}$.

    To bound the number of nonzero terms,
    we observe that every edge $A_{jk}$ for $j \neq k$ must occur
    zero times or at least twice in order to have nonzero expectation
    (the self-loops $A_{jj}$ can occur any number of times,
    and the 2-labeled edges $A_{jk}^2 - \frac 1 n$ must overlap at least one additional edge in order to have nonzero expectation).
    Each vertex in $V(\al)\setminus I(\al)\setminus \{\rootpic\}$ is incident to an edge of multiplicity
    1 or a 2-labeled edge, and so it must occur in at least two embeddings in order for that edge
    $A_{jk}$ to overlap and not make the expectation 0.
    This implies that the number of distinct non-root vertices among
    the embeddings is at most $q\left(|V(\al)|-1 +|I(\al)|\right)/2$ where the $-1$ is used to avoid counting the root. 

    Hence, there are at most $n^{q\left(|V(\al)|-1 +|I(\al)|\right)/2}$ ways to choose the entire image $\im(\varphi_1) \cup \ldots\cup  \im(\varphi_q)$.
    Once this is fixed, there are at most
    $(q|V(\al)|)^{q|V(\al)|}$ $q$-tuples of embeddings that map to these vertices. We conclude by combining the bound on the number of nonzero terms and the bound on the magnitude of each of these terms.
\end{proof}

\eqinfAlmostSure*

\begin{proof}
    By assumption, $x-y$ is a sum of combinatorially negligible terms. We first focus on a single one
    of them, say $a_n Z_\al$. For any $\varepsilon>0\,, q \in \N$ and $i\in [n]$, we have
    \begin{align*}
         \Pr\left(|a_n Z_{\al,i}|\ge \varepsilon\right) &\le \frac{\E|a_n Z_{\al,i}|^q} {\varepsilon^q}
         &\text{(Markov's inequality)}\\
         &\le \frac 1 {\varepsilon^q} M_{q|E(\al)|}2^{q|E(\al)|}(q|V(\al)|)^{q|V(\al)|} \cdot n^{-\frac q2}
         &\text{(\cref{lem:improper-magnitude})}\\
         &\le \frac 1 {\varepsilon^q}  (q|E(\al)|)^{O(q)} 2^{q|E(\al)|}(q|V(\al)|)^{q|V(\al)|} \cdot n^{-\frac q2} & (\text{subgaussianity of $A_{ij}$})\\
         &= \exp\left(O(q\log q) - \frac q 2 \log n + q\log (1/\eps)\right)\,.
     \end{align*}
     Picking $q=\log n$ and $\eps=q^C n^{-1/2}$
     and taking the constant $C$ large enough we
     can make the probability an arbitrarily
     small inverse polynomial in $n$. Then we
     take a union bound over all 
     $i\in [n]$ and all combinatorially negligible term appearing in $x-y$ (there are constantly many such terms by definition).
\end{proof}

\eqinfProperties*
\begin{proof} 
    It suffices to prove that for a combinatorially negligible term $n^{-k} Z_\al$:
    \begin{enumerate}[(i)]
        \item All terms in the diagram representation of $n^{-k} AZ_\al$ are combinatorially negligible.
        \item Let $n^{-\el} Z_\beta$ be any term of combinatorial order $1$ or combinatorially negligible.
        Then all terms in the diagram representation of the componentwise product $n^{-(k+\el)} Z_\al \odot Z_\beta$ are combinatorially negligible,
        where $\odot$ is the componentwise product.
    \end{enumerate}

    For (i), the diagram representation of $A Z_\al$ is given by \cref{lem:mat-mul}. In the term $\al^+$ without intersections,
    \[
        |V(\al^+)| = |V(\al)| + 1\,, \qquad |I(\al^+)| = |I(\al)|\,,  \qquad |E(\al^+)| = |E(\al)| + 1\,.
    \]
    From this we can check that $n^{-k} Z_{\alpha^+}$ is still combinatorially negligible.

    In a term $\beta$ corresponding to an intersection
    between the new root and a vertex of $\al$,
    \[
        |V(\beta)| = |V(\al)|\,, \qquad |I(\beta)| \le |I(\al)| + 1 \,, \qquad |E(\beta)| = |E(\al)| + 1\,.
    \]
   The second inequality follows from the observation that the only vertices from $\al$ whose neighborhood structure can be affected by the intersection are the root of $\al$ (which does not contribute to $|I(\al)|$) and the intersected vertex. Hence, $n^{-k} Z_\beta$ is also combinatorially negligible.

    For (ii), the diagram representation of $Z_\al\odot Z_\beta$ is given by \cref{lem:diagrams-product}. Fix an intersection pattern $P\in\calP(\al,\beta)$ that has $b$ blocks and denote by $\gamma$ the resulting diagram. Then,
    \begin{align*}
        |V(\gamma)| &= b+1 \,,\\
        |E(\gamma)| &= |E(\al)|+|E(\beta)| \,,\\
        |I(\gamma)| &\le |I(\al)|+|I(\beta)| + |V(\al)| + |V(\beta)| - b - 2 \,.
    \end{align*}
    The last inequality is proven by observing that for a non-root vertex that is neither in $I(\al)$ nor $I(\beta)$ to contribute to $I(\gamma)$, it must intersect another vertex. Moreover, there are at most $|V(\al)|+|V(\beta)|-b-2$ intersected non-root vertices in $\gamma$.
    
    Putting everything together,
    \begin{align*}
        &|V(\gamma)| - 1 - |E(\gamma)| + |I(\gamma)|\\
        \leq\;&|V(\al)| - 1 - |E(\al)| + |I(\al)| + |V(\beta)| - 1 - |E(\beta)| + |I(\beta)|\\
        <\;&2(k+l)\,,
    \end{align*}
    since $n^{-k} Z_\al$ is combinatorially negligible and $n^{-\el} Z_\beta$ is at most order 1. This concludes the proof.
\end{proof}

Using the 2-labeled edges introduced in \cref{sec:edge-labels}, we can
implement the removal of hanging double edges.
\removeDoubleEdge*
\begin{proof}
    Starting from the decomposition of \cref{lem:syntactic-removal},
    \[
        a_n Z_\al = a_n Z_{\al_0} - a_n \frac{|V(\al)| - 1}{n} Z_{\al_0} + a_n Z_{\al_2}\,,
    \]
    we claim that the first term is combinatorially order 1, and the second and third terms are combinatorially negligible.
    Comparing $\al_0$ to $\al$, two edges and one vertex in $I(\al)$ are removed.
    This does not change the combinatorial order.
    The second term scales down by $n$ and this becomes
    negligible (by assumption $|V(\al)|$ is constant).
    In the third term, $\abs{I(\al_2)} < \abs{I(\al)}$ to take into account the hanging vertex, while $|V(\al)| = |V(\al_2)|$ and $|E(\al)| = |E(\al_2)|$
    remain unchanged, making the
    term negligible.
    We remind the reader that $|E(\al)| = |E(\al_2)|$ because $|E(\al_2)|$
    counts 2-labeled edges twice.
\end{proof}

\cref{def:combinatorially-negligible}
includes the coefficient $a_n$ in the definition in order
to incorporate factors of $\frac 1 n$ on some error terms such as those in the proof above.

\subsection{Scalar diagrams}
\label{sec:scalar-diagrams}

We collect the properties of scalar diagrams (\cref{def:scalar-diagrams}) which naturally generalize those of vector diagrams.
We omit the proofs of the results in this section, as they are direct modifications of their vector analogs.

    First, the scalar diagrams are an orthogonal basis for scalar functions of $A$.
    \begin{lemma}
        For any proper $\alpha\in\calA_{\scalar}$:
        \begin{itemize}
            \item For any proper $\beta\in\calA_{\scalar}$ such that $\beta\neq \alpha$, $\E\left[Z_\al Z_\beta\right] = 0$.
            \item $\E \left[Z_\alpha\right]=0$ if $\alpha$ is not a singleton.
            \item The second moment of $Z_\al$ is
            \begin{align*}
                \E\left[Z_\al^2\right] &\underset{\phantom{n\to\infty}}{=} \abs{\Aut(\al)} \cdot \frac{n(n-1)\cdots(n-|V(\al)|+1)}{n^{|E(\al)|}} \\
                &\underset{n\to\infty}{=} \abs{\Aut(\al)} \cdot n^{|V(\al)|-|E(\al)|}(1+o(1))\,,
            \end{align*}
            where the last estimate holds whenever $|V(\al)|=o(\sqrt n)$.
        \end{itemize}
    \end{lemma}
    \begin{proof}
        Analogous to \cref{lem:orthogonality} and \cref{lem:variance}.
    \end{proof}

    When scalar and vector diagrams are multiplied together, the result can be expressed in terms of diagrams by extending the notion of intersection patterns $\calP(\al_1, \dots, \al_k)$ (\cref{def:intersection-pattern}) and intersection diagrams (\cref{def:intersection-diagram})
    to allow scalar and vector diagrams simultaneously. The ``unintersected'' diagram
    consists of adding all the scalar diagrams as floating components to the vector diagrams, which are put at the same root.
    Intersection patterns are partitions of this vertex set
    such that no two vertices from the same diagram are matched.
    \begin{lemma}\label{lem:scalar-product}
        Let $\al_1, \dots, \al_k$ be either scalar or vector diagrams.
        Then \[
        Z_{\al_1}\cdots Z_{\al_k}  = \sum_{P \in \calP(\al_1,\dots, \al_k)} Z_{\al_P}\,,\]
        where the product is componentwise for the vector diagrams.
    \end{lemma}
    \begin{proof}
        Analogous to \cref{lem:diagrams-product}.
    \end{proof}

We define $I(\al)$ for scalar diagrams exactly as in \cref{def:isolated}.

\begin{lemma}\label{lem:improper-magnitude-scalar}
    Let $q\in\N, \al \in \calA_\scalar,$ and $i \in [n]$. Then,
    \[
        \abs{\E\left[Z_{\al}^q\right]} \leq M_{q|E(\al)|} 2^{q|E(\al)|}  (q\abs{V(\al)})^{q|V(\al)|} \cdot n^{\frac q 2\left(|V(\al)| - |E(\al)|+|I(\al)|\right)}\,,
    \]
    where $M_k$ is defined as in \cref{lem:improper-magnitude}.
    When $q$ and $|V(\al)|$ are $O(1)$, this reduces to
    \[
        \abs{\E\left[Z_{\al}^q\right]} \leq O\left( n^{\frac q 2\left(|V(\al)| - |E(\al)|+|I(\al)|\right)}\right)\,.
    \]
\end{lemma}

\begin{proof}
    Analogous to \cref{lem:improper-magnitude}.
\end{proof}

\begin{definition}[Combinatorially negligible and order 1 scalar]
\label{def:combinatorially-negligible-scalar}
    Let $(a_n)_{n\in\N}$ be a sequence of real-valued coefficients with $a_n = \Theta(n^{-k})$, where $k\ge 0$ is such that $2k \in \Z$. Let $\al \in \calA_{\scalar}$ be a scalar diagram. 
    \begin{itemize}
    \item We say that $a_n Z_\al$ is \emph{combinatorially negligible}
    if
    \[|V(\al)| - |E(\al)| + |I(\al)| \leq 2k - 1\,.\]
    \item We say that $a_n Z_\al$ has \emph{combinatorial order 1} if
    \[|V(\al)| - |E(\al)| + |I(\al)| = 2k\,.\]
    \end{itemize}

    We define $\eqinf$ for scalar diagram expressions exactly as in \cref{def:asymptotic-equality}.
\end{definition}

\begin{lemma}\label{lem:almost-sure-scalar}
    Let $x$ and $y$ be scalar diagram expressions with $x\eqinf y$. Then $|x-y| = \widetilde{O}(n^{-1/2})$ with high probability.
\end{lemma}
\begin{proof}
    Analogous to \cref{lem:eqinf-almostsure}.
\end{proof}

\begin{lemma}\label{lem:comb-neg-scalar}
    Let $a_n Z_\al$ be a combinatorially negligible scalar term.
    Let $b_n Z_\beta$ be any scalar or vector term of combinatorial order at most $1$.
    Then all terms in the product $a_n b_n Z_\al Z_\beta$ are combinatorially negligible.
\end{lemma}
\begin{proof}
    Analogous to \cref{lem:eqinf-properties}.
\end{proof}

In \cref{lem:connected-nonnegligible}, we characterized the connected vector diagrams which are combinatorially order 1.
We now similarly characterize the order 1 scalar diagrams.
\begin{lemma}\label{lem:scalar-nonnegligible}
    Let $\al \in \calA_{\scalar}$ be a scalar diagram with $c$ connected components, $c_I$ of which contain only vertices in $I(\al)$.
Then 
$n^{-(c+c_I)/2}Z_\al$ is combinatorially negligible or combinatorially order 1, and it is combinatorially order 1 if and only if the following conditions hold simultaneously:
\begin{enumerate}[(i)]
    \item Every multiedge has multiplicity 1 or 2.
    \item There are no cycles.
    \item In each component, the subgraph of multiplicity 1 edges is empty or a connected graph
    (i.e. the multiplicity 2 edges consist of hanging trees)
    \item There are no self-loops or 2-labeled edges (\cref{sec:edge-labels}).
\end{enumerate}
\end{lemma}
\begin{proof}
    We proceed as in the proof of \cref{lem:connected-nonnegligible}.
    In each connected component $C$ containing at least one vertex $s\in V(\al)\setminus I(\al)$, we run a breadth-first search from $s$, assigning the multiedges used to explore a vertex to that vertex. This assigns at least one edge to every vertex in $C\setminus \{s\}$, and at least two edges to every vertex in $I(\al)\cap C$. This encoding argument shows that
    \begin{align}
        2\abs{I(\al)\cap C} + \abs{(V(\al)\setminus I(\al))\cap C}-1 \le \abs{E(C)}\,,\label{eq:encoding-first}
    \end{align}
    where $E(C)$ denotes the set of edges in the connected component $C$.
    
    In each connected component $C$ composed only of vertices in $I(\al)$, we run a breadth-first search from an arbitrary vertex, and obtain
    \begin{align}
        2(\abs{I(\al)\cap C}-1)=\abs{V(\al)\cap C}+\abs{I(\al)\cap C}-2\,\le \abs{E(C)}\,.\label{eq:encoding-second}
    \end{align}
    Summing \cref{eq:encoding-first} and \cref{eq:encoding-second} over all connected components, we obtain
    \[
        \abs{V(\al)}-\abs{E(\al)}+\abs{I(\al)}\le (c-c_I) + 2c_I  = c + c_I\,.
    \]
    This shows that $n^{-(c+c_I)/2} Z_\al$ is combinatorially negligible or combinatorially order 1, and it is combinatorially order 1 if and only if equality holds in the argument. This happens if and only if there is no cycle, multiplicity $\sgt 2$ edges, self-loops, or 2-labeled edges anywhere; and if the graph induced by the multiplicity 1 multiedges is connected.
\end{proof}

With this result in hand, we can now characterize the order-1 vector diagrams with several connected components:

\begin{corollary}\label{cor:classify-full}
    Let $\al \in \calA$ be a vector diagram with $c$ floating components, $c_I$
of which consist only of vertices in $I(\al)$. Then $n^{-(c+c_I)/2}Z_\al$ is combinatorially order 1 if and only if both the floating components (viewed as one scalar diagram) scaled by $n^{-(c+c_I)/2}$ and the component of the root are combinatorially order 1.
\end{corollary}
\begin{proof}
    \cref{def:combinatorially-negligible} sums across the root and floating components, so we may
    apply both \cref{lem:connected-nonnegligible} and \cref{lem:scalar-nonnegligible}.
\end{proof}

\subsection{Classification of diagrams}
\label{sec:diagram-classification-proof}

\begin{lemma}\label{lem:Ztau-normal}
    For all $\sig \in \calS$ and $i\in [n]$, $Z_{\sig, i} \overset{d}{\longrightarrow} \calN(0, \abs{\Aut(\sig)})$. Similarly, for all $\tau\in\calT_{\scalar}$, $n^{-\frac 1 2} Z_\tau\overset{d}{\longrightarrow} \calN(0, \abs{\Aut(\tau)})$.
\end{lemma}
\begin{proof}
    We prove that the moments $\E\left [Z_{\sig,i}^q\right]$ match the Gaussian moments and use \cref{lem:method-of-moments}.
    
    Let $q \in \N$ be a constant independent of $n$.
    First, we expand the product $Z_{\sig, i}^q$ in the diagram basis
    using \cref{lem:diagrams-product}.
    Using \cref{lem:connected-nonnegligible}, the only combinatorially order
    1 terms occur when there are no cycles, all multiedges have multiplicity 1 or 2,
    and the multiplicity 2 edges form hanging trees.
    Any term with an edge of multiplicity 1 disappears when we take the
    expectation $\E\left [Z_{\sig,i}^q\right]$,
    while the diagrams which are entirely hanging trees are equal to $\rootpic$ up to combinatorially negligible terms (\cref{lem:remove-double-edge}). Further, $\rootpic$ has expectation 1, and by \cref{lem:improper-magnitude} each of the combinatorially
    negligible terms has expectation $O(n^{-1/2})$.
    Thus, $\E\left[Z_{\sig,i}^q\right]$ equals the number of ways to create hanging
    trees of double edges, up to a term that converges to 0 as $n\to\infty$.

    For each of the $q$ copies of $\sig$, the single edge incident
    to the root must be paired with another such edge.
    This extends to an automorphism of the entire subtree.
    In conclusion, $\E \left[Z_{\sigma,i}^q\right]$ converges to $\abs{\Aut(\sigma)}^{q/2}$ times the number of perfect matchings on $q$ objects, and we conclude by \cref{lem:gaussian-moments} and \cref{lem:method-of-moments}.
    The proof for the scalar case is analogous.
\end{proof}

\begin{lemma}\label{lem:Ztau-hermite}
    If $\tau \in \calT$ consists of
    $d_\sig$ copies of the subtrees $\sig \in \calS$, then
    \[Z_{\tau} \eqinf \prod_{\sig \in \calS} h_{d_\sig}(Z_{\sig}; \abs{\Aut(\sig)})\,.\]
    For $\rho \in \calF_{\scalar}$ with $c$ components and
    consisting of $d_\tau$ copies of each tree $\tau \in \calT_{\scalar}$,
    \[n^{-\frac c 2} Z_\rho \eqinf \prod_{\tau \in \calT_{\scalar}} h_{d_\tau}\left(n^{-\frac 1 2}Z_{\tau}; \abs{\Aut(\tau)}\right)\,.\]
\end{lemma}
\begin{proof}
    We first expand $h_{d}(Z_{\sig}; \abs{\Aut(\sig)})$ in the diagram basis using \cref{lem:diagrams-product}
    and identify the dominant terms, i.e. those which are combinatorially order 1.
    As in the proof of \cref{lem:Ztau-normal}, the combinatorially order 1 terms in each monomial $Z_{\sig, i}^k$ consist of
    pairing up copies of the tree $\sig$:
    \[
        Z_{\sig}^k \eqinf \sum_{M \in \calM(k)} \abs{\Aut(\sig)}^{|M|} Z_{k - 2|M|\text{ copies of }\sig}\,,
    \]
    where $\calM(k)$ is the set of partial matchings on $k$ objects.
    Now we use the combinatorial interpretation of Hermite polynomials (\cref{fact:hermite-matchings}),
    \begin{align*}
        h_d(Z_\sig; \abs{\Aut(\sig)}) &= \sum_{N \in \calM(d)} (-1)^{|N|} \abs{\Aut(\sig)}^{|N|} Z_{\sig}^{d - 2|N|}\\
        &\eqinf \sum_{N \in \calM(d)} (-1)^{|N|} \abs{\Aut(\sig)}^{|N|} \sum_{M \in \calM(d - 2|N|)}  \abs{\Aut(\sig)}^{|M|} Z_{d - 2|N| - 2|M| \text{ copies of }\sig}\\
        &= \sum_{M' \in \calM(d)}  \abs{\Aut(\sig)}^{|M'|} Z_{d - 2|M'|\text{ copies of }\sig} \sum_{N \subseteq M'} (-1)^{|N|}\\
        &= Z_{d\text{ copies of }\sig}\,.
    \end{align*}
    This completes the argument when $\tau$ consists of several copies of a single $\sig \in \calS$.
    If $\sigma,\sigma'\in\calS$ are distinct, using again \cref{lem:diagrams-product} and \cref{lem:connected-nonnegligible}, we can check that
    \[
        Z_{\text{$d$ copies of $\sig$}}\odot Z_{\text{$d'$ copies of $\sig'$}} \eqinf Z_{\text{$d$ copies of $\sig$ and $d'$ copies of $\sig'$}}\,.
    \]
    The proof then follows by applying these arguments inductively, and 
extends analogously to scalar diagrams.
\end{proof}

\begin{lemma}\label{lem:scalar-to-vector}
    Let $\al \in \calF$ have $c$ floating components.
    Let $\al_{\smallrootpic}$ be the component of the root
    and $\al_{\textnormal{float}}$ be the floating components. Then
    $n^{- \frac c 2}Z_\al \eqinf n^{- \frac c 2} Z_{\al_{\textnormal{float}}}Z_{\al_{\smallrootpic}}$.
\end{lemma}
\begin{proof}
    The product $n^{-\frac c 2}Z_{\al_{\textnormal{float}}}Z_{\al_{\smallrootpic}}$ can be expanded
    in the diagram basis using \cref{lem:scalar-product}.
    We claim that the only non combinatorially negligible diagram
    is the one without intersections, which equals $n^{-\frac c 2}Z_\al$.
    When an intersection occurs, it can only be between the root component and a floating component.
    The new component of the root is at most combinatorially order 1
    (this is a property of all connected vector diagrams, \cref{lem:connected-nonnegligible}),
    so there is an 
    ``extra'' factor of $\frac{1}{\sqrt{n}}$ from the lost component which makes the intersection term negligible.
\end{proof}

\begin{lemma}
\label{lem:Ztau-indpt}
    $\left\{Z_{\sig, i} : \sig \in \calS, i \in [n]\right\} \cup \left\{n^{-\frac 1 2}Z_\tau : \tau \in \calT_{\scalar}\right\}$ are asymptotically
    independent.
\end{lemma}
\begin{proof}
    Fix constants $q,r \in \N$.
    We proceed by computing the moment of a set of diagrams $\sig_1, \dots, \sig_q \in \calS$
    rooted at $i_1, \dots, i_q \in [n]$ and $\tau_1, \dots, \tau_r \in \calT_{\scalar}$:
    \begin{align}
        \E\left[\prod_{p=1}^q Z_{\sig_p, i_p} \prod_{p=1}^r n^{-\frac 1 2}Z_{\tau_p}\right]\,.\label{eq:moment-independence}   
    \end{align}
    Let $|V| = \sum_{p = 1}^q |V(\sig_p)| + \sum_{p = 1}^r |V(\tau_p)|$ and $|E| = \sum_{p = 1}^q |E(\sig_p)| + \sum_{p = 1}^r |E(\tau_p)|$.
    Let $q_{\textnormal{distinct}}$ be the number of distinct roots, i.e. the number of distinct elements in $\{i_1, \dots, i_q\}$.

    Expanding \cref{eq:moment-independence} gives a sum over embeddings of the diagrams.
    We will prove that the dominant terms factor across
    the distinct $(\sigma_p, i_p)$ and $\tau_p$; they correspond
    to pairing up isomorphic $\sig_p$ at each distinct root and isomorphic $\tau_p$.

    Each nonzero term in the expansion of \cref{eq:moment-independence} equals $n^{-(|E|+ r)/2}$
    (when every edge appears exactly twice) or $O(n^{-(|E|+ r)/2})$ (in general) by \cref{assump:A-entries}.
    We partition the summation based on the intersection pattern as in \cref{def:intersection-pattern}.
    For a given intersection pattern, letting $I$ be the union of the images of the embeddings,
    the number of terms with this pattern
    is $(1-o(1))\cdot n^{|I| - q_{\textnormal{distinct}}}$ because the $q_{\textnormal{distinct}}$ root vertices are fixed.
    In an embedding with nonzero expectation, every edge appears at least
    twice, so every non-root vertex is in at least two embeddings.
    Applying this bound to all of the non-root vertices in $I$,
    \[|I| \leq q_{\textnormal{distinct}} + \frac {|V| - q}  2\,.\]
    Multiplying the value of each term times the number of terms, the total contribution of this intersection pattern is
    \[n^{|I| - q_{\textnormal{distinct}} - \frac{|E|+r} 2} \le n^{\frac 1 2\left(|V| - q-|E|-r\right)}\,.\]
    Since the individual diagrams are connected, the exponent is nonpositive.
    The dominant terms occur exactly when $|I| = q_{\textnormal{distinct}} + (|V| - q)/2$, equivalently all of the non-root vertices intersect exactly
    one other non-root vertex. Each edge must occur at least twice,
    and this condition implies that each edge occurs exactly twice
    in the dominant terms.

    We claim that the only way that each edge and vertex can be in exactly two embeddings is if isomorphic
    $\sig_p$ and $\tau_p$ are paired.
    Indeed, by connectivity of $\sig_p$ and $\tau_p$, sharing one edge extends to an isomorphism.
    Furthermore, because non-root vertices must intersect other non-root
    vertices in the dominant terms, we have that no pairs can be made between
    $\sig_p$ and $\tau_{p'}$, or between $\sig_p$ and $\sig_{p'}$
    which have distinct roots.
\end{proof}

\cref{thm:classification} follows from \cref{lem:Ztau-normal}, \cref{lem:Ztau-hermite}, \cref{lem:scalar-to-vector}, and \cref{lem:Ztau-indpt}.
The constant-order joint moments of all the diagrams
are summarized into the next theorem which generalizes \cref{thm:intro-classification}.

\begin{theorem}\label{all-moments}
    Suppose that $A = A(n)$ is a sequence of random matrices satisfying \cref{assump:A-entries}.
    For all $\al_1, \dots, \al_k \in \calA, \; i_1, \dots, i_k \in [n]$ and $\beta_1, \dots, \beta_\el \in \calA_\scalar$ (allowing repetitions anywhere), 
    \[
        \E\left[\prod_{j = 1}^k n^{-C(\al_j)/2}Z_{\al_j, i_j} \prod_{j = 1}^\ell n^{-C(\beta_j)/2}Z_{\beta_j} \right] = \E\left[\prod_{j = 1}^k Z_{\al_j,i_j}^\infty \prod_{j = 1}^\ell Z^\infty_{\beta_j} \right] + O(n^{-\frac 12})\,,
    \]
    where $C(\al)$ is the number of floating components of $\al$,
    and where the asymptotic random variables $(Z_{\al,i}^\infty)_{\al \in \calA, i \in [n]}$ and $(Z_\beta^\infty)_{\beta \in \calA_\scalar}$ are defined as:
    \[
    \allowdisplaybreaks
    \begin{cases}
        Z_{\sig,i}^\infty \sim \calN(0, \left|\Aut(\sig)\right|) \text{ independently}& \text{if $\sig \in \calS$}\\
        Z_\tau^\infty \sim \calN(0, \left|\Aut(\tau)\right|) \text{ independently}& \text{if $\tau \in \calT_\scalar$}\\
        \displaystyle Z_{\rho,i}^\infty = \prod_{\sigma \in \calS} h_{d_\sig}(Z_{\sig,i}^\infty; \left|\Aut(\sig)\right|) \prod_{\tau \in \calT_\scalar} h_{d_{\tau}}(Z_\tau^\infty; \left|\Aut(\tau)\right|) & \text{if $\rho \in \calF$}\\
        \displaystyle Z_\rho^\infty = \prod_{\tau \in \calT_\scalar} h_{d_{\tau}}(Z_\tau^\infty; \left|\Aut(\tau)\right|) & \text{if $\rho \in \calF_\scalar$}\\
        Z_{\al,i}^\infty = Z_{\al_0, i}^\infty  \text{ and }Z_\beta^\infty = Z_{\beta_0}^\infty & \text{if removing hanging double edges}\\&\text{creates $\al_0 \in \calF$ or $\beta_0 \in \calF_\scalar$}\\
        Z_{\al,i}^\infty = Z_\beta^\infty = 0 & \text{if removing hanging double edges}\\&\text{is not in $\calF$ or $\calF_\scalar$} 
    \end{cases}
    \]
\end{theorem}

\subsection{Handling empirical expectations}
\label{sec:onsager-correction}

Empirical expectations are highly concentrated and the following lemma confirms this.
Note that the empirical expectations in the Onsager correction for AMP (\cref{sec:amp}) will create floating
components in the diagrams of the algorithmic state, but all such diagrams will be negligible.

\empiricalExpectation*

\begin{proof}
The effect of summing a vector diagram $Z_\al=(Z_{\al,i})_{i\in [n]}$ over $i$ is to unroot $\al$, converting it to a scalar diagram. We prove this operation makes every diagram combinatorially negligible,
except for the constant term.
For $k \geq 0$ and a vector diagram $\al \in \calA$:
    \begin{enumerate}[(i)]
        \item If $a_n Z_\al$ is combinatorially negligible,
        then $\frac{a_n}{n}\sum_{i=1}^n Z_{\al,i}$ is a combinatorially negligible scalar term.
        \item If $a_n Z_\al$ has combinatorial order 1,
        and the root of $\al$ is incident to at least
        one edge of multiplicity 1,
        then $\frac{a_n}{n}\sum_{i=1}^n Z_{\al,i}$ is a combinatorially negligible scalar term.
    \end{enumerate}
    Unrooting a vector diagram does not change the number of vertices nor the number of edges. During this operation, the number of vertices in $I(\al)$ stays the same if the root is adjacent to an edge of multiplicity 1; otherwise it increases by at most 1. We readily check from the definition that the extra $\frac{1}{n}$ makes the resulting scalar terms combinatorially negligible.

Now let $\widehat{x} \eqinf x$ be the tree approximation.
The difference $x - \widehat{x}$ consists of combinatorially negligible terms
which stay negligible by part (i) above.
The trees in $\calT$ become negligible by part (ii) above with the exception of
the singleton tree which becomes $1$.
The singleton has coefficient $\E[\widehat{x}_1] = \E[X]$ since the other trees
are mean-zero (\cref{lem:constant}).
\end{proof}

\section{High-degree tree diagrams are not Gaussian}
\label{sec:star}

We compute that the star-shaped diagram with $\log n$
leaves and the root at a leaf is not Gaussian
(its fourth moment is significantly larger than the square of its second moment), demonstrating that care must be taken when studying diagrams of superconstant size.\footnote{Similarly, adding an edge between two of the leaves creates a cyclic diagram with negligible variance but non-negligible fourth moment.}
This diagram appears after only $T = O(\log \log n)$ iterations in 
the recursion
\[x_1 = A \vec{1} \qquad x_{t+1} = (x_t)^2 \qquad x_{T+1} = A x_T\,.\]
However, we expect that this diagram does not contribute significantly to nicer GFOMs that strictly alternate between multiplication by $A$ and constant-degree componentwise operations.

Fixing $d$, let $\gam$ denote ($d$-star graph)$^+$.
We compute that $\E\left[Z_{\gam,1}^4\right] \gg \E\left[Z_{\gam,1}^2\right]^2$ when $d \approx \log n$.
By \cref{lem:variance}, the variance is
\[\E\left[Z_{\gam, 1}^2\right] = (1+o(1))\abs{\Aut(\gam)} = (1+o(1))d!\,.\]
When computing the fourth moment $\E\left[Z_{\gam, 1}^4\right]$ for constant $d$, the terms that are dominant
consist of (1) a perfect matching between the four edges incident to the root, (2) perfect matchings between their $d$ children.
There are $3(d!)^2$ such terms, recovering the fourth moment of a Gaussian with variance $d!$.

For $d = \log n$, another type of term becomes dominant.
These are the terms where all four edges incident to the root are equal,
then we have a perfect matching on $4d$ objects divided into four groups
of size $d$ such that no two objects from the same group are matched.
Denote the latter set of matchings by $\calM(d,d,d,d)$.
\begin{lemma}\label{lem:four-matchings}
    Up to a multiplicative $\poly(d)$ factor, $|\calM(d,d,d,d)| \gtrsim 3^d (d!)^2$.
\end{lemma}
These terms come with a $\frac 1 n$ factor due to the multiplicity 4 edge.
When $d = \Omega(\log n)$, the extra factor of $3^d$ overpowers the $\frac 1 n$ and makes the fourth moment much larger than the
the squared variance $(d!)^2$.

\begin{proof}[Proof of \cref{lem:four-matchings}.]
    We establish a recursion.
    There are $(3d)(3d-1)\cdots(2d+1)$ ways to match up the objects in the first group, which can be partitioned
    in $O(d^2)$ ways depending on how many objects in each other group are matched.
    We will recurse on the ``maximum-entropy'' case in which the
    first group matches 
    $d/3$ elements from each other group, using the following claim.
    \begin{claim}\label{claim:max-matching}
        Let $d,k \in \N$ such that $\frac{d}{k-1}$ is an integer.
        Counting the matchings between $d$ objects and a subset of $(k-1)d$ objects
        in $k-1$ groups,
        as a function of the number of objects matched in each group,
        the number of matchings is maximized when there are $\frac{d}{k-1}$
        matched elements per group.
    \end{claim}
    \begin{proof}[Proof of \cref{claim:max-matching}.]
        Letting $n_1, \dots, n_{k-1}$ be the number of matched elements per group, we may directly compute this number as
        $\prod_{i=1}^{k-1}(d)_{n_i}$
        where $(d)_k = d(d-1)\cdots (d-k+1)$ is the falling factorial.
        When $n_i$ and $n_j$ are replaced by $n_i - 1$ and $n_j + 1$, the ratio of new to old values is
        \begin{align*}
            \frac{d-n_j}{d - n_i + 1}
        \end{align*}
        which is at least 1 if $n_i \geq n_j + 1$. Hence the $n_i$ are equal at the maximum.
    \end{proof}

    Using \cref{claim:max-matching}, up to a factor of $O(d^2)$,
    \begin{align*}
    |\calM(d,d,d,d)| &\gtrsim (3d)(3d-1)\cdots (2d+1)|\calM(2d/3,2d/3,2d/3)|\\
    &\asymp \left(\frac{3d}{e}\right)^{3d} \left(\frac{e}{2d}\right)^{2d} |\calM(2d/3,2d/3,2d/3)|
    \end{align*}
    where the second equality holds up to a $\poly(d)$ factor
    by Stirling's approximation:
    \begin{fact}[Stirling's approximation]
    \label{claim:stirling}
        Up to a multiplicative $\poly(d)$ factor, $d! \asymp \left(\frac{d}{e}\right)^d$.
    \end{fact}
    Recursing via the same principle,
    \begin{align*}
    |\calM(2d/3,2d/3,2d/3)| &\gtrsim (4d/3)(4d/3-1)\cdots(2d/3+1)|\calM(d/3,d/3)|\\
    &= (4d/3)(4d/3-1)\cdots(2d/3+1)(d/3)!\\
    &\asymp \left(\frac{4d}{3e}\right)^{4d/3}\left(\frac{3e}{2d}\right)^{2d/3}\left(\frac{d}{3e}\right)^{d/3} & \text{(\cref{claim:stirling})}
    \end{align*}
    In total,
    \[|\calM(d,d,d,d)| \gtrsim 3^d \left(\frac{d}{e}\right)^{2d}\,.\qedhere\]
\end{proof}

\end{document}